\documentclass{article}

\usepackage{fullpage}

\title{Nearly-Linear Time LP Solvers and Rounding Algorithms for Scheduling Problems}

\author{Shi Li\\ State Key Laboratory for Novel Software Technology, \\Nanjing University, \\Nanjing, Jiangsu Province, China.\\ \href{mailto:shili@nju.edu.cn}{shili@nju.edu.cn}}

\usepackage[utf8]{inputenc}

\usepackage{multirow}

\ifdefined\CR\else 
\usepackage{graphicx} 
\usepackage{amsmath, amsthm, amssymb, amsfonts}
\usepackage{thmtools} 
\usepackage{thm-restate}
\usepackage[colorlinks]{hyperref}
\fi

\graphicspath{{./figs/}}

\usepackage{enumitem}
\setlist{topsep=3pt,itemsep=0pt}

\usepackage{algorithm}
\usepackage{algorithmicx}
\usepackage{algpseudocode}
\algtext*{EndWhile}
\algtext*{EndIf}
\algtext*{EndFor}
\algtext*{EndFunction}
\algtext*{EndProcedure}

\ifdefined\CR\else
\newtheorem{theorem}{Theorem}[section]
\newtheorem{definition}[theorem]{Definition}
\newtheorem{lemma}[theorem]{Lemma}
\newtheorem{claim}[theorem]{Claim}
\fi
\newtheorem{coro}[theorem]{Corollary}
\newtheorem{obs}[theorem]{Observation}

\usepackage{framed}

\newcommand{\ceil}[1]{{\left\lceil#1\right\rceil}}
\newcommand{\floor}[1]{{\left\lfloor#1\right\rfloor}}
\newcommand{\set}[1]{{\left\{#1\right\}}}

\newcommand{\R}{{\mathbb{R}}}
\newcommand{\Z}{{\mathbb{Z}}}

\newcommand{\bbD}{{\mathbb{D}}}

\mathchardef\mhyphen="2D
\DeclareMathOperator*\E{\mathbb{E}}
\DeclareMathOperator*\union{\bigcup}

\newcommand{\poly}{\mathrm{poly}}

\newcommand{\bfC}{{\mathbf{C}}}

\newcommand{\bfP}{{\mathbf{P}}}

\newcommand{\bfa}{{\mathbf{a}}}
\newcommand{\bfb}{{\mathbf{b}}}
\newcommand{\bfc}{{\mathbf{c}}}

\newcommand{\bff}{{\mathbf{f}}}
\newcommand{\bfg}{{\mathbf{g}}}
\newcommand{\bfu}{{\mathbf{u}}}
\newcommand{\bfv}{{\mathbf{v}}}
\newcommand{\bfx}{{\mathbf{x}}}
\newcommand{\bfy}{{\mathbf{y}}}
\newcommand{\bfz}{{\mathbf{z}}}

\newcommand{\calA}{{\mathcal{A}}}
\newcommand{\calF}{{\mathcal{F}}}

\newcommand{\calI}{{\mathcal{I}}}

\newcommand{\calO}{{\mathcal{O}}}
\newcommand{\calQ}{{\mathcal{Q}}}
\newcommand{\calR}{{\mathcal{R}}}

\newcommand{\calT}{{\mathcal{T}}}

\newcommand{\opt}{{\mathsf{opt}}}
\newcommand{\lp}{{\mathsf{lp}}}

\newcommand{\sfd}{{\mathsf{d}}}

\newcommand{\busy}{{\mathsf{busy}}}
\newcommand{\idle}{{\mathsf{idle}}}

\newcommand{\val}{{\mathsf{val}}}
\newcommand{\supp}{{\mathsf{supp}}}

\newcommand{\inclen}{\mathsf{inc}\mhyphen\mathsf{len}}
\newcommand{\listscheduling}{\mathsf{list}\mhyphen\mathsf{scheduling}}

\newcommand{\tprec}{\mathrm{prec}}

\newcommand{\load}{{\mathsf{load}}}
\newcommand{\vecload}{\mathsf{loads}}

\allowdisplaybreaks

\begin{document}
		\maketitle
	
	\begin{abstract}
		We study nearly-linear time approximation algorithms for non-preemptive scheduling problems in two settings: the unrelated machine setting, and the identical machine with job precedence constraints setting, under the well-studied objectives such as makespan and weighted completion time. For many problems, we develop nearly-linear time approximation algorithms with approximation ratios matching the current best ones achieved in polynomial time.
		
		Our main technique is linear programming relaxation.  For the unrelated machine setting, we formulate mixed packing and covering LP relaxations of nearly-linear size, and solve them approximately using the nearly-linear time solver of Young. For the makespan objective, we develop a rounding algorithm with $(2+\epsilon)$-approximation ratio. For the weighted completion time objective, we prove the LP is as strong as the rectangle LP used by Im and Li, leading to a nearly-linear time $(1.45 + \epsilon)$-approximation for the problem. 
		
		For problems in the identical machine with precedence constraints setting, the precedence constraints can not be formulated as packing or covering constraints.   To achieve the nearly-linear running time, we define a polytope for the constraints, and leverage the multiplicative weight update (MWU) method with an oracle which always returns solutions in the polytope.
		
	\end{abstract}

	
	\section{Introduction}
	\label{sec:intro}
	
	Scheduling theory is an important sub-area of combinatorial optimization, operations research and approximation algorithms. Over the past few decades, advanced techniques have been developed to design approximation algorithms for numerous scheduling problems, among which mathematical relaxation is a prominent one. The algorithms based on the technique follow a two-step framework: solve some linear/convex/semi-definite programming relaxation for the problem to obtain a fractional schedule, and round it into an integral one.  The main focus of the algorithm design in the literature has been the best approximation ratios that can be achieved in polynomial time.    Many of the LPs used have size much larger than that of the input, and a general convex/semi-definite program requires a large polynomial time to solve, making these algorithms impractical. 
	
	To overcome the running time issue, we design approximate LP-based scheduling algorithms that run in \emph{nearly-linear} time. We focus on two well-studied non-preemptive scheduling settings:
	\begin{enumerate}
		\item {\bf Unrelated machine setting.} We are given a set $J$ of $n$ jobs, a set $M$ of $m$ machines, a bipartite graph $G = (M, J, E)$ between $M$ and $J$, and a processing time $p_{ij} \in \Z_{> 0}$ for every $ij  \in E$, indicating the time it takes to process job $j$ on machine $i$. If $ij \notin E$, then the job $j$ can not be processed on machine $i$.
		The output of a problem in this setting is an assignment $\sigma \in M^J$ of jobs to machines so that $\sigma_jj \in E$ for every $j \in J$.  This indicates that we process the job $j$ on machine $\sigma_j$. 
		
		\item {\bf Identical machine with job precedence constraints setting.} In this setting, we are given  a set $J$ of $n$ jobs,  each job $j \in J$ with a processing time $p_j \in \Z_{\geq 0}$,  and the number $m \geq 1$ of identical machines. There are precedence constraints of the form $j \prec j'$, indicating that the job $j'$ can only start after job $j$ completes. The output of a problem in the setting is a completion time vector $(C_j)_{j \in J} \in \Z_{\geq 0}^J$, meaning that a job $j \in J$ is processed during the time interval $(C_j - p_j, C_j]$. We need $C_j \geq p_j$ for every $j \in J$, $C_j \leq C_{j'} - p_{j'}$ for every $j \prec j'$, and every integer $t \geq 1$ is contained in $(C_j - p_j, C_j]$ for at most $m$ jobs $j \in J$.\footnote{It is a folklore that if the last property is satisfied, we can assign $\{(C_j - p_j], j \in [J]\}$ to $m$ machines so that the intervals assigned to each machine are disjoint.}  
	\end{enumerate}
	
	The main objective function we focus on is \emph{weighted completion time}: We are additionally given a weight $w_j \in \Z_{> 0}$ for every job $j \in J$, and the goal of the problem is to minimize $\sum_{j \in J}w_j C_j$, where $C_j$ is the completion time of $j$ on its assigned machine.  For the second setting, this is explicitly given by the output. For the first setting, given the assignment  $\sigma \in M^J$ of jobs to machines, it is well-known that the Smith's rule\footnote{By this rule, we schedule jobs $j$ assigned to a machine $i$ using non-decreasing order of $p_{ij}/w_j$.} gives the optimum order on each machine $i$. 
	For the first setting, we also consider the objective of minimizing the \emph{makespan}, which is defined as $\max_i \sum_{j\in \sigma^{-1}(i)} p_{ij}$, i.e., the maximum load over all machines.
	
	It is convenient for us to use the classic three-field notation $\alpha|\beta|\gamma$ in \cite{GLLR79} to denote scheduling problems studied in this paper.\footnote{In the notation, $\alpha$ indicates the machine model, $\beta$ gives the set of additional constraints, and $\gamma$ is the objective. $\alpha = R$ and $\alpha = P$ denote the unrelated and identical machine settings respectively, and $\tprec \in \beta$ indicates that jobs have precedence constraints. $\gamma = C_{\max}$ and $\gamma = \sum_jw_jC_j$ denote the makespan and weighted completion time objectives respectively.}  The makespan and weighted completion time minimization problems in the unrelated machine setting are denoted as $R||C_{\max}$ and $R||\sum_j w_jC_j$ respectively.  The problem to minimize weighted completion time in the identical machine with job precedence constraint setting is denoted as $P|\tprec|\sum_j w_j C_j$. We will also consider special cases of the problem, and give their notations when we discuss them.
	
	There is a rich literature on designing approximation algorithms for these problems. For the unrelated makespan minimization problem, i.e., $R||C_{\max}$,  the classic result of Lenstra, Shmoys and Tardos \cite{LST90} gives a $2$-approximation, which remains the state-of-the-art result. The problem is NP-hard to approximate within a factor of better than $1.5$. 
	Plotkin, Shmoys and Tardos \cite{PST95} studied fast approximation algorithms for the problem, as an application of  their packing and covering LP solver. They developed a randomized $(2+\epsilon)$-approximation algorithm in time $\tilde O_\epsilon(mn)$.\footnote{In this paper, we use $\tilde O_{\epsilon}(\cdot)$ to hide a factor that is poly-logarithmic in the input size of the instance being considered, which will be clear from the context, and polynomial in $1/\epsilon$, where $\epsilon$ is a precision parameter. An algorithm is nearly-linear if its running time is $\tilde O_\epsilon(\text{input size})$.} So their algorithm is nearly-linear if $|E| = \Theta(m n)$. Much work on the problem has focused on a special setting called the restricted assignment setting \cite{Sve12, JR17, JR20}, where there is an intrinsic size $p_j \in \Z_{>0}$ for every $j \in J$, and for every $ij  \in E$ we have $p_{ij} = p_j$.  
	
	
	For the unrelated machine weighted completion time problem, i.e., $R||\sum_j w_jC_j$,  many independent rounding algorithms 
	achieve an approximation ratio of $1.5$ \cite{SS02, Sku01, SS99, Li20}.  
	Bansal, Svensson and Srinivasan \cite{BSS16} showed that the barrier of $1.5$ is inherent for this type of algorithms. 
	%
	To overcome the barrier, they developed a novel dependent rounding scheme and a lifted SDP relaxation for the problem, leading to a $(1.5-1/2160000)$-approximation algorithm. The ratio has been improved to $1.5 - 1/6000$ by Li \cite{Li20}, to $1.488$ by Im and Shadloo \cite{IS20} and to the current best ratio of $1.45$ by Im and Li \cite{IL23}. The three subsequent works are based on the rectangle LP relaxation for the problem. 
	
	There is a vast literature on the problem of minimizing weighted completion time in the identical machine with job precedence constraints setting, i.e., the problem $P|\tprec|\sum_j w_j C_j$.  A special case of the problem where there is only one machine (i.e., $m = 1$), denoted as $1|\tprec|\sum_j w_j C_j$, is already non-trivial. Hall et al.\ \cite{HSS97} developed a $2$-approximation for the problem, which is the best possible under some stronger version of the unique game conjecture introduced by Bansal and Khot \cite{BK09}. Another special case that is considered moderately in the literature is when all jobs have unit-size, denoted as $P|\tprec, p_j = 1|\sum_j w_jC_j$.  Munier, Queyranne and Schulz \cite{MQS98} gave approximation ratios of $3$ and $4$ for the special case and the general problem $P|\tprec|\sum_j w_j C_j$ respectively. The ratios were improved to $1+\sqrt{2}$ and $2 + 2\ln2$ by Li \cite{Li20}. 
	Most algorithms \cite{HSS97, MQS98, QS06, Li20} for $P|\tprec|\sum_j w_jC_j$ and the two special cases use the following framework: Solve some linear/convex program to obtain an order of the jobs respecting the precedence constraints. For every job in this order, schedule it as early as possible, without violating the precedence and $m$-machine constraints.
	\smallskip
	
	Most of the results we discussed focused on optimizing the approximation ratios with polynomial time algorithms. Albeit being polynomial, the running times in these results are often very large.  For LP-based algorithms, this may be caused by two factors. First, the size of an LP might already be large w.r.t the input size. Consider a typical time-indexed LP relaxation in the unrelated machine setting,  one need a variable for every triple $ijs$ with $ij \in E$ and $s$ being the starting time. Assuming the number of possible starting times is linear in $n$, the number of variables in the LP is already $\Theta(n|E|)$; the size of the LP can only be bigger. Second, these algorithms often use a general LP solver, which has a large running time w.r.t the size of the LP. There is a vast literature in recent years on designing exact and approximate general LP solvers. Here we could only include a few representative results. To solve a linear program with $\bar n$ variables, $\bar m$ constraints and $\bar N$ non-zero coefficients up to a precision of $\epsilon$, Lee and Sidford \cite{LS15} developed an algorithm with running time $\tilde O\big((\bar N + \bar m^2)\sqrt{\bar m}\log\frac1\epsilon\big)$. Lee, Song and Zhang \cite{LSZ19} gave an algorithm with running time $\tilde O(\bar n^\omega\log\frac1\epsilon)$,\footnote{The result requires that the LP does not have redundant constraints.} where $\omega \approx 2.373$ is the current best exponent for matrix multiplication. Brand, Lee, Sidford and Song \cite{BLS20} provided a $\tilde O(\bar m \bar n+\bar n^3)$ time randomized algorithm that solves the LP exactly with high probability; the running time is nearly linear if the constraint matrix is dense and tall. However, to solve general linear programs, these running times are at least quadratic, even if the LP has a linear size. Convex or semi-definite programming based algorithms need to solve the CP/SDP using the interior point or ellipsoid methods, which are often time-consuming.
	
	\subsection{Our Results}
	To overcome the above issue, we design approximation algorithms for scheduling problems, that run in \emph{nearly-linear} time, i.e., in time $\tilde O_\epsilon(\text{input size})$. So,  up to a $\poly(\log n, 1/\epsilon)$-factor, our running times are the best possible. Some of the algorithms  we developed have been studied empirically \cite{Als22}.  In the unrelated machine setting, $G = (M, J, E)$ denotes the bipartite graph between $M$ and $J$, and a nearly-linear time is of order $\tilde O_\epsilon(|E|)$.  For the identical machine with precedence constraints setting, we use $\kappa$ to denote the number of precedence constraints. A nearly-linear time algorithm runs in time $\tilde O_\epsilon(n + \kappa)$. Unlike the polynomial running time scenario, we can not assume $\prec$ is transitive, as it may dramatically increase the number of precedence constraints to quadratic. Moreover, the best known algorithm computing the transitive closure of the precedence constraints takes $O(n\kappa)$ time \cite{Pur70}. 
	
	For many problems, including $R||C_{\max}, R||\sum_j w_j C_j, 1|\tprec|\sum_j w_j C_j$ and $P|\tprec, p_j=1|\sum_j w_j C_j$, our nearly-linear time algorithms achieve the correspondent best known  polynomial-time approximation ratios, due to Lenstra, Shmoys and Tardos \cite{LST90}, Im and Li \cite{IL23}, Hall et~al.\ \cite{HSS97}, and Li \cite{Li20} respectively.	 
	
	\begin{theorem}
		\label{thm:main-R-Cmax}
		For any $\epsilon > 0$, there is a $\tilde O_\epsilon(|E|)$-time $(2+\epsilon)$-approximation algorithm for $R||C_{\max}$, i.e., the makespan minimization problem on unrelated machines. 
	\end{theorem}
	
	For the problem $R||\sum_j w_jC_j$, we believe that showing that the rectangle LP can be approximated in nearly-linear time is interesting on its own. So we give two theorems for the problem. Refer to LP\eqref{LP:rectangle} for the formal description of the rectangle LP for the problem. 
	\begin{theorem}
		\label{thm:main-R-wC-LP} Consider an instance of $R||\sum_jw_jC_j$ and the rectangle LP \eqref{LP:rectangle} for the instance.
		Let $\epsilon > 0$ and $\lp_{\eqref{LP:rectangle}}$ be the value of the LP. Then in $\tilde O_\epsilon(|E|)$ time, we can construct a solution $\bfz$ to the LP such that: 
		\begin{itemize}
			\item $\bfz$ satisfies all the constraints in the LP, except that the constraint at most one job is processed on any machine at any time may be violated by a factor of $1 + \epsilon$. (Formally, Constraint \eqref{LPC:rectangle-capacity} is only satisfied with the right-side replaced by $1+\epsilon$.)
			\item The value of $\bfz$ to the LP is at most $(1+\epsilon)\lp_{\eqref{LP:rectangle}}$. 
		\end{itemize}
	\end{theorem}
	In the theorem, our $\bfz$ will be represented by the list of non-zero coordinates and their values.  Then, we show that the rounding algorithm of Im and Li \cite{IL23} can indeed run in time nearly-linear on the support size of the LP solution. This gives the following theorem. 
	\begin{theorem}
		\label{thm:main-R-wC}
		For any $\epsilon > 0$, there is a $\tilde O_\epsilon(|E|)$-time $(1.45+\epsilon)$-approximation algorithm for $R||\sum_j w_jC_j$, i.e., the weighted completion time minimization problem on unrelated machines. 
	\end{theorem}
	
	The following two theorems are for $1|\tprec|\sum_j w_j C_j$ and $P|\tprec, p_j = 1|\sum_j w_jC_j$.
	
	\begin{theorem}
		\label{thm:main-1-Prec-wC}
		For any $\epsilon > 0$, there is a $\tilde O_\epsilon((n + \kappa)\log p_{\max})$-time $(2+\epsilon)$-approximation algorithm for $1|\tprec|\sum_j w_jC_j$, i.e., the weighted completion time problem on a single machine with precedence constraints, where $p_{\max}:=\max_{j \in J}p_j$ is the maximum job size.
	\end{theorem}
	So the algorithm runs in nearly-linear time only when $p_{\max}$ is polynomially bounded.
	
	\begin{theorem}
		\label{thm:main-P-Prec-pj1-wC}
		For any $\epsilon > 0$, there is a $\tilde O_\epsilon(n + \kappa)$-time $(1 + \sqrt{2}+\epsilon)$-approximation algorithm for $P|\tprec, p_j = 1|\sum_j w_jC_j$, i.e., the weighted completion time problem on identical machines with unit-size jobs and precedence constraints.
	\end{theorem}
	
	Along the way of algorithm design for the identical machine with precedence constraints setting, we developed a nearly-linear time $(1+\epsilon)$-approximation algorithm for the single commodity network flow problem in directed acyclic graphs, with bounded supplies and demands on sources and sinks, but infinite capacities on edges.

	Recently there has been a lot of progress on solving maximum flow problem on undirected and directed graphs.  For undirected graphs, the problem can be approximated within a factor of $1+\epsilon$ in nearly-linear time \cite{KLO14, Pen16,She17}, and solved exactly with a slightly weaker running time of $m^{1 + o(1)}$ (this is called \emph{almost-linear} time) \cite{BGS21}.  It was open whether an almost-linear  running time can be achieved for solving maximum flow on directed graphs.\footnote{By repeatedly solving maximum flow instances on residual graphs, one can convert an approximate maximum flow algorithm on directed graphs to an exact algorithm, without much loss on the running time. So for directed graphs, allowing $(1+\epsilon)$-approximation does not give much advantage. } This was resolved in the affirmative by a recent breakthrough due to Chen et al.\ \cite{CKL22}: They developed an algorithm that computes exact maximum flows on directed graphs with polynomially bounded integral capacities in $m^{1+o(1)}$ time.  
	Thus, we could use the result as a black-box for our problem, if we allow the running time to be almost-linear. Nevertheless as our theme is to design \emph{nearly-linear} time algorithms, we include in 
	\ifdefined\CR
	the full version of the paper  %
	\else
	Appendix~\ref{sec:appendix-networkflow}
	\fi
	our approximate maximum-flow algorithm for the special case with this running time. To the best of our knowledge, this was not known before.

	For the general precedence-constrained scheduling  problem $P|\tprec|\sum_j w_j C_j$ (on multiple machines with variant job lengths), we achieve an $O(1)$-approximation algorithm in nearly-linear time. However, the approximation ratio of the algorithm is $6 + \epsilon$, which is worse than the best polynomial-time ratio of $2+2\ln2$ due to Li \cite{Li20}.
	\begin{theorem}
		\label{thm:main-P-Prec-wC}
		For any $\epsilon > 0$, there is a $\tilde O_\epsilon((n + \kappa)\log p_{\max})$-time $(6+\epsilon)$-approximation algorithm for $P|\tprec|\sum_j w_jC_j$, i.e., the weighted completion time minimization problem on identical machines with precedence constraints, where $p_{\max}:=\max_{j \in J}p_j$ is the maximum job size.
	\end{theorem}
	
	\subsection{Our Techniques}
	All of our algorithms are based on linear programming: We design an LP relaxation of nearly-linear size, solve it in nearly-linear time to obtain a $(1+\epsilon)$-approximate solution, and round the solution into an integral schedule in nearly-linear time. 
	
	For $R||C_{\max}$, the natural LP relaxation has $O(|E|)$ size, and the mixed packing and covering form.  Thus it can be solved within a factor of $1+\epsilon$ by the algorithm of Young \cite{You14} in $\tilde O_\epsilon(|E|)$ time.  In particular,  the algorithm outputs a $(1+\epsilon)$-approximate solution that violates the constraints by a factor of $1\pm \epsilon$, in $O\left(\frac{\bar N\log \bar m}{\epsilon^2}\right) = \tilde O_\epsilon(\bar N)$ time, where $\bar m$ and $\bar N$ are the number of constraints and non-zero coefficients in the LP respectively.  	To round the fractional solution, we apply the grouping technique of \cite{ST93} for the so called generalized assignment problem, but with a $(1+\epsilon)$-slack. This gives us a bipartite graph $H = (V, J, E_H)$ satisfying $|N_H(J')| \geq (1+\epsilon)|J'|$ for every $J' \subseteq J$, where $N_H(J')$ is the set of neighbors of $J'$ in $H$.  This allows us to find a matching in $H$ that covers $J$ in nearly-linear time, which leads to a $(2 + \epsilon)$-approximate solution, matching the current best approximation of $2$ in \cite{LST90}.  We remark that the $\tilde O_\epsilon(mn)$-running time of \cite{PST95} comes from both solving the LP, and rounding the LP solution.  So even with the nearly-linear time mixed covering and packing LP solver, the algorithm of \cite{PST95} still requires $\tilde O_\epsilon(mn)$ time. 	
	
	For the problem $R||\sum_j w_jC_j$, 
	we give a nearly-linear size mixed packing and covering LP  that (up to a factor of $1+O(\epsilon)$) is equivalent  to the rectangle LP used by Li \cite{Li20}, Im and Shadloo \cite{IS20}, Im and Li \cite{IL23}. In the rectangle LP, there is a variable $x_{ijs}$ indicating if a job $j$ is scheduled on the machine $i$ and has starting time $s$, and constraints that at most one job is processed at any time on any machine.   To reduce the size of the LP to $\tilde O_{\epsilon}(|E|)$, we partition the time horizon into \emph{windows}, with lengths geometrically increasing by a factor of $1+\epsilon$. We distinguish between two types of scheduling intervals: If a job is scheduled within a window on some machine $i$ (we call this an inside-window interval), then we do not need to capture the precise location of the scheduling interval. On the other hand, if the job starts and ends at two different windows (we call the interval an cross-window interval), we will approximately capture its starting and ending times. To do so, we divide each window into $1/\epsilon$ sub-windows, and let the LP variables capture the two sub-windows containing the starting and completion times.  In the LP, we require all the cross-window intervals incur a congestion of 1: any point $t$ is covered by at most 1 fraction of cross-window intervals. Then we require the total volume of jobs processed inside each window is at most its length.  We show that up to a factor of $1+O(\epsilon)$, a solution to the LP can be converted to one for the rectangle LP with no large cost.  Roughly speaking, the width of window is small compared to its position and so we do not need to know the precise location of an inside-window-interval.  For a cross-window-interval, we may incur an error on its length that is about $\epsilon$ times the total length of its starting window and ending window.  As a sub-window has a small length, and a cross-window-interval covers some window-boundary, the total error incurred will also be small. 
	
	We proceed to our techniques for the weighted completion time problems in the identical machine with precedence constraints setting, i.e., the problem $P|\tprec|\sum_j w_j C_j$ and its special cases.  Due to the precedence constraints, the LP relaxations do not have the mixed packing and covering form anymore. 
	Nevertheless, the multiplicative weight update (MWU) framework can still be applied.  We enclose the precedence constraints in a polytope $\calQ$. In each iteration of the MWU framework, we guarantee that all these constraints are satisfied, i.e., the vector we obtain is in $\calQ$.  Other than the precedence constraints, we have $\tilde O_\epsilon(\log p_{\max})$ packing inequalities correspondent the $m$-machine constraint. This is due to that we can round completion times to integer powers of $1+\epsilon$.  
	
	The number of iterations the MWU framework takes is $\tilde O_\epsilon(\bar m)$, where $\bar m$ is the number of packing constraints in the LP, without counting the constraints for $\calQ$.  Fortunately we have $\bar m = \tilde O_\epsilon(\log p_{\max})$. To obtain the claimed $\tilde O_\epsilon((n + \kappa)\log p_{\max})$ time,  we need to run each iteration of MWU in nearly-linear time.  The bottleneck comes from finding a vector in $\calQ$ satisfying one aggregated packing constraint, that maximizes a linear objective with non-negative coefficients.
	
	A key technical contribution of our paper is an oracle for the problem. For an appropriately defined directed acyclic graph $G = (V, E)$, the polytope $\calQ$ can be formulated as $\{\bfy \in [0, 1]^V: y_v \leq y_u, \forall vu \in E\}$. For two given row vectors $\bfa, \bfb \in \R_{\geq 0}^V$, the aggregated LP in each iteration of MWU is: $\max \bfa\bfy$ subject to $\bfy \in \calQ$ and $\bfb\bfy \leq 1$. Using LP duality, the problem is reduced to the special single commodity maximum flow problem we introduced:  We have bounded supplies and demands on sources and sinks, but infinite capacities on edges. When allowing a $(1+\epsilon)$-approximation for the scheduling problem, we need to find a flow whose value is at least the maximum value for the instance with sink capacities scaled by $\frac{1}{1+\epsilon}$.  This is done by our nearly-linear time maximum-flow algorithm for the special case.  
	
	\subsection{Other Related Work}
	The makespan minimization problem in the identical machine setting with precedence constraints, i.e., the problem $P|\tprec|C_{\max}$, is another classic problem in scheduling theory. The seminal work of Graham \cite{Gra69} gives a simple greedy algorithm that achieves a $2$-approximation. On the negative side, Lenstra and Rinnooy Kan \cite{LR78} proved a $(4/3-\epsilon)$-hardness for the problem. Under the stronger version of the Unique Game Conjecture (UGC) introduced by Bansal and Khot \cite{BK09}, Svensson \cite{Sve10} showed that the problem is hard to approximate within a factor of $2-\epsilon$ for any $\epsilon > 0$.  
	Much work has focused on the special case where $m = O(1)$ and all jobs have size $1$ \cite{LR21, Garg18, Li21}, for which obtaining a PTAS is a long-standing open problem.

	The multiplicative weight update (MWU) method for solving linear programs has played an important role in a wide range of applications. Some of its foundational work can be found in a beautiful survey by Arora, Hazan and Kale \cite{AHK12}.   
	There has been a vast literature on solving packing, covering, and mixed packing and covering LPs  approximately to a factor of $1+\epsilon$ using iterative methods \cite{SM90, PST95, LN93, You95, GK07, KY07, KY13, You14, AO15, CQ18}.  In particular, to solve a mixed packing and covering LP with $\bar n$ variables, $\bar m$ constraints and $\bar N$ non-zero coefficients, the algorithm of Young \cite{You14} returns $(1+\epsilon)$-approximation deterministically in $O\left(\frac{\bar N\ln \bar m}{\epsilon^2}\right)$ time.  The dependence on $\epsilon$ has been improved slightly by Chekuri and Quanrud \cite{CQ18}, who gave a randomized algorithm with running time $\tilde O\left(\frac{\bar N}{\epsilon} + \frac{\bar m}{\epsilon^2} + \frac{\bar n}{\epsilon^3}\right)$, where $\tilde O(\cdot)$ hides a poly-logarithmic factor.

	There has been a recent surge of interest in designing fast or nearly-linear time approximation algorithms for combinatorial optimization problems \cite{CQ17, CQ18a, CHQ20, CGL20, Li21, BG21}.\smallskip
	
	\noindent{\bf Organization.}\  \ The rest of the paper is organized as follows. In Section~\ref{sec:prelim}, we define some elementary notations used across the paper, and describe the result of Young \cite{You14} on solving mixed packing and covering LPs, and a template solver for packing LPs over an ``easy'' polytope. In Sections~\ref{sec:unrelated-makespan} and \ref{sec:unrelated-wC}, we present our results for $R||C_{\max}$ and $R||\sum_j w_j C_j$. 	
	\ifdefined\CR
	Due to the page limit, we leave our algorithms for $P|\tprec|\sum_j w_jC_j$ and the two special cases $1|\tprec|\sum_j w_jC_j$ and $P|\tprec, p_j = 1|\sum_j w_jC_j$ to the full version of the paper. The full version also contains other technicalities, such as how to handle the case where input integers are not polynomially bounded,  how to reduce problems to the promise versions and how to use the self-balancing binary search tree data structure to run a list scheduling algorithm.
	\else
	\fi
	
	
	\section{Preliminaries}
	\label{sec:prelim}
	We use bold lowercase letters to denote vectors, and their correspondent italic letters to denote their coordinates.  We use bold uppercase letters to denote matrices. $\bf0$ and $\bf1$ are used to denote the all-$0$ and all-$1$ vectors whose domain can be inferred from the context.  Given a template vector $\bfv$ over some finite domain,  and a subset $S$ of the domain, let $v(S) := \sum_{e \in S}v_e$ be the sum of $v$-values over elements in $S$.
	
	Given an (undirected) graph $H = (V_H, E_H)$, we use $\delta_{H}(v), N_{H}(v), \delta_{H}(U), N_{H}(U)$ to respectively denote the sets of incident edges of $v \in V_H$, neighbors of $v$, edges between the set $U\subseteq V_H$ and $V_H \setminus U$, and vertices in $V_H \setminus U$ with at least one neighbor in $U$, in the graph $H$.  Given a directed graph $H = (V_H, E_H)$, 
	for every $v \in V_H$, we use $\delta^+_H(v)$ and $\delta^-_H(v)$ to denote the sets of outgoing and incoming edges of $v$ respectively. For every $U \subseteq V_H$, let $\delta^+_H(U):=\{uv  \in E_H: u \in U, v \notin U\}$ and $\delta^-_H(U):=\{uv  \in E_H: u \notin U, v \in U\}$ be the sets of edges from $U$ to $V_H\setminus U$ and from $V_H\setminus U$ to $U$ respectively. 		
	When $H = G$ for the graph $G$ in the context (which can be undirected or directed), we omit the subscript $H$ in the notations.
	
	For cleanness of exposition, we use $\tilde O_{\epsilon}(\cdot)$ to hide factors that are polynomial in $\frac1\epsilon$ and poly-logarithmic in the size of the input. 
	As we gave the first nearly-linear time algorithms for the studied problems, the hidden factors are small compared to the improvements we make.   The final approximation ratios we get have an additive factor of $O(\epsilon)$ (instead of $\epsilon$); but it can be reduced to $\epsilon$ if we start from a smaller $\epsilon$.  By default, for an (undirected or directed) graph $H = (V_H, E_H)$ we deal with, we assume every vertex is incident to at least one edge so $|E_H| = \Omega(V_H)$. 
	%
	%
	For any $a \in \R$, we define $(a)_+$ as $\max\{a, 0\}$. 

	\subsection{Nearly-Linear Time Mixed Packing and Covering LP Solver}
	A mixed packing and covering LP is an LP of the following form: 
	\begin{align}
		\text{find } \bfx \qquad \text{such that} \qquad \bfx \geq 0,\qquad \bfP \bfx \leq \bf1\qquad\text{and}\qquad\bfC \bfx \geq \bf1, \tag{\textrm{MPC}} \label{LP:MPC}
	\end{align}
	where $\bfP \in \R_{\geq 0}^{\bar m_{\bfP} \times \bar n}$ and $\bfC \in \R_{\geq 0}^{\bar m_{\bfC} \times \bar n}$ for some positive integers $\bar n, \bar m_{\bfP}, \bar m_{\bfC}$.  Let $\bar m = \bar m_\bfP + \bar m_\bfC$ and $\bar N$ be the total number of non-zeros in $\bfP$ and $\bfC$.  Young \cite{You14} developed a nearly-linear time algorithm that solves \eqref{LP:MPC} approximately: 
	\begin{theorem}[\cite{You14}] \label{thm:fast-LP-solver}
		Given an instance of \eqref{LP:MPC} and $\epsilon > 0$, there is an $O\left(\frac{\bar N \log \bar m}{\epsilon^2}\right)$-time algorithm that either claims \eqref{LP:MPC} is infeasible, or outputs an $\bfx \in \R_{\geq 0}^{\bar n}$ such that $\bfP \bfx \leq (1+\epsilon)\bf1$ and $\bfC \bfx \geq \frac{\bf1}{1+\epsilon}$.
	\end{theorem}
	
	\subsection{Template Packing LP Solver over a Simple Polytope}
	\label{subsec:temp-lp-solver}
	In this section, we describe a template MWU-based LP solver for a packing linear program with an additional requirement that the solution is inside an ``easy'' polytope $\calQ$.  The framework we describe here is introduced in \cite{CJV15} and later reformulated in \cite{CQ17}. 
	
	Let $\bfP \in \R_{\geq 0}^{{\bar m} \times {\bar n}}$ be a non-negative matrix, with ${\bar N}$ non-zero entries. Let $\bfa \in \R_{\geq 0}^{{\bar n}}$ be a row vector, and $\calQ \subseteq \R_{\geq 0}^{{\bar n}}$ be a polytope which is defined by ``easy'' constraints.  We focus on the following linear program:
	\begin{align}
		\max\ \bfa\bfx \qquad \text{subject to} \qquad \bfx \in \calQ\qquad\text{and}\qquad\bfP\bfx \leq {\bf1}. \tag{$\mathrm{P}_\calQ$} \label{LP:packing}
	\end{align}
	Throughout the paper, we make sure all instances of \eqref{LP:packing} we deal with are feasible. 
	\begin{definition}
		Let $\epsilon \in (0, 1), \phi > 0$ be two parameters. An $(\epsilon, \phi)$-approximate solution to \eqref{LP:packing} is a vector $\bfx \in \calQ$ satisfying $\bfP\bfx \leq (1+\epsilon)\bf1$ and $\bfa\bfx \geq \bfa\bfx^* - \phi$, where $\bfx^* \in \calQ$ is the optimum solution to \eqref{LP:packing}.
	\end{definition}
	
	As a hindsight, we only allow a loss of an additive factor $\phi$ in the objective function of the LP for $P|\tprec|\sum_j w_jC_j$, which will be set to be a polynomially small term.  As is typical in a MWU framework, we need to solve the following LP where the constraints $\bfP\bfx \leq \bf1$ are aggregated into one constraint $\bfb\bfy \leq 1$, where $\bfb \in \R_{\geq 0}^{{\bar n}}$ is a row vector:
	\begin{align}
		\max\ \bfa\bfy \qquad \text{subject to} \qquad \bfy \in \calQ\qquad\text{and}\qquad\bfb\bfy \leq 1. \label{LP:packing-aggregate}
	\end{align}
	Again we guarantee all instances of \eqref{LP:packing-aggregate} we encounter are feasible. 
	\begin{definition}
		Let $\epsilon \in (0, 1), \phi > 0$ be two parameters.
		An $(\epsilon, \phi)$-approximate solution to \eqref{LP:packing-aggregate} is a vector $\bfy \in \calQ$ satisfying $\bfb\bfy \leq 1+\epsilon$ and $\bfa\bfy \geq \bfa\bfy^* - \phi$, where $\bfy^*$ is the optimum solution to the LP.  An $(\epsilon, \phi)$-oracle for \eqref{LP:packing-aggregate} is an algorithm that,  given an instance of \eqref{LP:packing-aggregate}, and $\epsilon  \in (0, 1), \phi > 0$,  outputs an $(\epsilon, \phi)$-approximate solution $\bfy$ to \eqref{LP:packing-aggregate}.
	\end{definition}
	
	
	\begin{algorithm}[ht]
		\caption{LP Solver for \eqref{LP:packing}}
		\label{alg:LP-solver}
		\begin{algorithmic}[1]
			\Require{an instance of \eqref{LP:packing}, $\epsilon \in (0, 1), \phi > 0$, and $(\epsilon, \phi)$-oracle $\calO$ for \eqref{LP:packing-aggregate}}
			\Ensure{an $(O(\epsilon), \phi)$-approximate solution $\bfx$ for \eqref{LP:packing}}
			\State $t \gets 0, \rho \gets \frac{\ln {\bar m}}{\epsilon^2}, \bfx^{(0)} \gets {\bf0} \in \R_{\geq 0}^{\bar n}, \bfu^{(0)} \gets {\bf1} \in \R_{\geq 0}^{\bar m}$  
			\Statex\Comment{$\bfx^{(t)}$'s are column vectors and $\bfu^{(t)}$'s are row vectors}
			\While{$t < 1$}\label{step:lp-solver-main-loop}
			\State define $\bfb := \frac{\bfu^{(t)}}{|\bfu^{(t)}|}\bfP$, and run the oracle $\calO$ for \eqref{LP:packing-aggregate} to obtain an $(\epsilon, \phi)$-approximate solution $\bfy$ for \eqref{LP:packing-aggregate} \label{step:lp-solver-use-oracle}
			\State $\displaystyle \delta \gets \min\left\{\min_{i \in [{\bar m}]} \frac{1}{\rho \cdot  \bfP_i\bfy}, 1-t \right\}$\label{step:lp-solver-delta} \medskip
			\For{every $i \in [{\bar m}]$}
			$u^{(t+\delta)}_i \gets u^{(t)}_i \cdot \exp\big(\delta\epsilon \rho \cdot  \bfP_i \bfy\big)$ 
			\EndFor
			\State $\bfx^{(t + \delta)} \gets \bfx^{(t)} + \delta \bfy, t \gets t + \delta$
			\EndWhile
			\State \Return $\bfx:=\bfx^{(1)}$
		\end{algorithmic}
	\end{algorithm}
	
	The template LP solver is described in Algorithm~\ref{alg:LP-solver}, where we use $\bfP_i$ to denote the $i$-th row vector of $\bfP$.  By our assumption that \eqref{LP:packing} is feasible, the instance of \eqref{LP:packing-aggregate} defined in every execution of Step~\ref{step:lp-solver-use-oracle} is also feasible.  The performance of the algorithm is summarized in the following theorem. 
	
	\begin{restatable}{theorem}{thmtemplatelpsolver}
		Algorithm \ref{alg:LP-solver} will return an $(O(\epsilon), \phi)$-approximate solution $\bfx$ to \eqref{LP:packing}, within $O(\frac{{\bar m} \log {\bar m}}{\epsilon^2})$ iterations of Loop~\ref{step:lp-solver-main-loop}.
	\end{restatable}
	
	\begin{proof}
		Focus on one iteration of Loop~\ref{step:lp-solver-main-loop}. Let $t$ be the value of $t$ at the beginning of the iteration,  $\bfy$ and $\delta$ be the $\bfy$ and $\delta$ obtained in Step~\ref{step:lp-solver-use-oracle} and \ref{step:lp-solver-delta} in the iteration respectively.  Then we have
		\begin{align*}
			&\quad |\bfu^{(t+\delta)}| = \sum_{i \in [{\bar m}]} u^{(t + \delta)}_i = \sum_{i \in [{\bar m}]} u^{(t)}_i \exp(\delta\epsilon \rho\cdot  \bfP_i\bfy) \leq \sum_{i \in [{\bar m}]} u^{(t)}_i (1+(1+\epsilon)\epsilon\cdot \delta\rho \cdot \bfP_i \bfy) \\
			&= |\bfu^{(t)}| + (1+\epsilon)\epsilon\delta\rho\cdot\bfu^{(t)}\bfP\bfy
			\leq |\bfu^{(t)}| + (1+\epsilon)^2\epsilon\delta\rho \cdot |\bfu^{(t)}| \leq |\bfu^{(t)}| \exp((1+\epsilon)^2\epsilon\delta \rho).
		\end{align*}
		The inequality in the first line is by that $\delta\rho \cdot \bfP_i \bfy \in [0, 1]$ for every $i \in [{\bar m}]$ and $e^{\epsilon \theta} \leq 1 + \epsilon\theta + (\epsilon\theta)^2 \leq 1 + \epsilon\theta + \epsilon^2\theta$ for every $\epsilon\in[0, 1]$ and $\theta \in [0, 1]$.  The first inequality in the second line is by that $\frac{\bfu^{(t)}}{|\bfu^{(t)}|}\bfP\bfy  = \bfb\bfy \leq 1 + \epsilon$.
		%
		
		Combining the inequality over all iterations, we have 
		\begin{align}
			|\bfu^{(1)}| \leq |\bfu^{(0)}| \exp\left((1+\epsilon)^2\epsilon \rho\right) = {\bar m}\cdot \exp\left((1+\epsilon)^2\epsilon \rho\right). \label{inequ:u1-to-u0}
		\end{align} 
		%
		For every $i \in [{\bar m}]$, we have $u^{(1)}_i  = \exp\left(\epsilon \rho \cdot \bfP_i \bfx\right)$, where  $\bfx:=\bfx^{(1)}$ is the returned solution. So, by \eqref{inequ:u1-to-u0}, we have $\exp(\epsilon \rho  \cdot \bfP_i \bfx) \leq {\bar m} \cdot \exp((1+\epsilon)^2\epsilon \rho)$, which implies $ \bfP_i \bfx \leq \frac{\ln {\bar m}}{\epsilon \rho}  + (1+\epsilon)^2 \leq (1+\epsilon)^2+\epsilon  = 1 +O(\epsilon)$. 
		
		In the end $\bfx=\bfx^{(1)}$ is a convex combination of vectors $\bfy$ obtained in all iterations. As each $\bfy$ is in $\calQ$, we have $\bfx \in \calQ$. Moreover, for the instance of \eqref{LP:packing-aggregate} in any iteration, $\bfx^*$ is a valid solution. So, the optimum solution $\bfy^*$ to the instance of \eqref{LP:packing-aggregate} has $\bfa\bfy^* \geq \bfa\bfx^*$, and the $\bfy$ returned by the oracle has $\bfa\bfy \geq \bfa\bfy^* - \phi \geq \bfa\bfx^* - \phi$. This implies our final $\bfx$ has $\bfa\bfx \geq \bfa\bfx^* - \phi$. Therefore, $\bfx$ is a $(O(\epsilon), \phi)$-approximate solution to \eqref{LP:packing}.
		
		It remains to bound the number of iterations that Loop~\ref{step:lp-solver-main-loop} can take. 
			In every iteration of loop~\ref{step:lp-solver-main-loop} except for the last one, some $i$ has $\frac{1}{\rho \cdot \bfP_i \bfy} = \delta$, i.e., $\delta\epsilon\rho\cdot \bfP_i \bfy = \epsilon$.   We say $u_i$ is increased fully in the iteration. Notice by \eqref{inequ:u1-to-u0}, each $u_i$ can be increased fully in at most $\frac{\ln\big({\bar m} \exp((1+\epsilon)^2\epsilon\rho)\big)}{\epsilon} = \frac{\ln {\bar m}+ (1+\epsilon)^2\epsilon \rho}{\epsilon} = O\left(\frac{\ln {\bar m}}{\epsilon^2}\right)$ iterations.  This bounds the number of iterations by $O\left(\frac{{\bar m}\log {\bar m}}{\epsilon^2}\right)$ as there are ${\bar m}$ different values of $i$.
	\end{proof}

	For each iteration of loop~\ref{step:lp-solver-main-loop}, the steps other than Step~\ref{step:lp-solver-use-oracle} takes $O({\bar N})$ time. Therefore, the running time of Algorithm~\ref{alg:LP-solver} is $O\left(\frac{{\bar m}\log {\bar m} \cdot {\bar N}}{\epsilon^2}\right)$, plus the time for running the oracle $O\left(\frac{{\bar m}\log {\bar m}}{\epsilon^2}\right)$ times.
	
	
	\section{Unrelated Machine Makespan Minimization}
	\label{sec:unrelated-makespan}
	In this section, we give the nearly-linear time $(2+\epsilon)$-approximation algorithm for the unrelated machine makespan minimization problem, i.e, the problem $R||C_{\max}$. Recall that we are given a bipartite graph $G = (M, J, E)$ and a $p_{ij} \in \Z_{>0}$ for every $ij  \in E$. Recall that $N(j), N(i), \delta(j)$ and $\delta(i)$ denote the set of neighbors or incident edges of a job $j \in J$ or a machine $i \in M$, in the graph $G$.
	
	Via a standard technique described in the full version of the paper, 
	we can focus on the following promise version:
	\begin{itemize}
		\item We are given a number $P \geq \opt$, where $\opt$ is the optimal makespan of the instance, and our goal is to construct an assignment of makespan at most $(2+O(\epsilon))P$. 
	\end{itemize}	
	For some $ij  \in E$ with $p_{ij} > P$, we remove $ij $ from $E$, as the optimum solution does not use the edge.  The following is the natural LP relaxation for the problem: \vspace*{-10pt}
	
	\noindent\begin{minipage}[t]{0.35\textwidth}
		\begin{align}
			\sum_{j \in N(i)} p_{ij} x_{ij} \leq P,  \forall i \in M \label{LPC:makespan} 
		\end{align}
	\end{minipage}\hfill
	\begin{minipage}[t]{0.33\textwidth}
		\begin{align}
			\sum_{i \in N(j)} x_{ij} \geq1, \forall j \in J \label{LPC:makespan-covered}
		\end{align}
	\end{minipage}\hfill
	\begin{minipage}[t]{0.3\textwidth}
		\begin{align}
			x_{ij} \geq 0, \forall ij  \in E \label{LPC:makespan-non-negative}
		\end{align}
	\end{minipage}

	In the correspondent integer program, $x_{ij} \in \{0, 1\}$ for every $ij  \in E$ indicates whether the job $j$ is assigned to machine $i$. \eqref{LPC:makespan} requires that the makespan of the schedule to be at most $P$, \eqref{LPC:makespan-covered} requires every job to be scheduled. In the linear program, we replace the requirement that $x_{ij} \in \{0, 1\}$ with the non-negativity constraint \eqref{LPC:makespan-non-negative}. 
	
	By the promise that $P \geq \opt$, the LP is feasible.  Therefore, applying Theorem~\ref{thm:fast-LP-solver}, we can solve the LP in $\tilde O_\epsilon(|E|)$ time to obtain an approximate solution $\bfx \in [0, 1]^E$.  By scaling, we can assume \eqref{LPC:makespan-covered} holds with equalities, and \eqref{LPC:makespan} holds with right side replaced by $(1+O(\epsilon))P$.

	To round the solution to an integral assignment in $\tilde O_{\epsilon}(|E|)$-time, we use the grouping idea from \cite{ST93}: For each machine $i \in M$, we break the fractional jobs assigned to $i$ into groups, each containing $\frac{1}{1+\epsilon}$ fractional jobs.  This gives us a bipartite graph $H$ between jobs and groups. Any perfect matching (i.e., a matching covering all jobs $J$) will give a $(2+O(\epsilon))$-approximation for the makespan problem.  In $H$, every subset $J' \subseteq J$ of jobs has at least $(1+\epsilon)|J'|$ neighbors. The $(1+\epsilon)$-factor allows us to design a $\tilde O_{\epsilon}(|E|)$-time algorithm to find a matching covering all jobs $J$, as stated in the following lemma:
	\begin{restatable}{lemma}{lemmapbm}
		\label{lemma:pbm}
		Assume we are given a bipartite graph $H = (S, T, E_H)$ and $\epsilon > 0$ such that $|N_H(S')| \geq (1+\epsilon)|S'|$ for every $S' \subseteq S$.  In $O\left(\frac{|E_H|}{\epsilon}\log |S|\right)$-time, we can find a matching in $H$ covering all vertices in $S$.
	\end{restatable}
	\begin{proof}
		Let $L = \floor{\log_{1+\epsilon}|S|} + 1 > \log_{1+\epsilon}|S|$.  Then we use the shortest-augmenting path algorithm of  Hopcroft and Karp \cite{HK73} to find a matching for which there is no augmenting path of length at most $2L+1$.  The running time of the algorithm can be made to $O(|E_H| L) = O(\frac{|E_H|}{\epsilon}\log |S|)$. 
		%
		%
		It remains to show the following lemma: 
		\begin{restatable}{lemma}{lemmashortestaugmenting}
			\label{lemma:shortest-augmenting}
			Let $F$ be a matching in $H$ for which there is no augmenting path of length at most $2L+1$. Then all vertices in $S$ are matched in the matching $F$.
		\end{restatable}
		\begin{proof}
			Let $\vec{H}$ be the residual graph of $H$ w.r.t the $F$: $\vec H$ is a directed graph over $S \cup T$, for every edge $st \in E_H$, we have $st \in \vec H$, and for every $st \in F$, we have $ts \in \vec H$.   We say a vertex in $S$ is free if it is unmatched in $F$.  For every integer $\ell \in [0, L]$, define $S^\ell$ ($T^\ell$ resp.) to be the set of vertices in $S$ ($T$, resp.) to which there exists a path in $\vec{H}$ of length \emph{at most} $2\ell$ ($2\ell+1$, resp.) from a free vertex. 
			So, we have $S^0 \subseteq S^1 \subseteq S^2 \subseteq \cdots \subseteq S^L$ and $T^0 \subseteq T^1 \subseteq T^2 \subseteq \cdots \subseteq T^L$.
			
			Notice that $T^{\ell} = N_H(S^{\ell})$ for every $\ell \in [0, L]$. 
			So for every $\ell \in [0, L]$, we have $({1+\epsilon})|S^{\ell}| \leq {|T^{\ell}|}$ by the condition of the lemma. All vertices in $T^L$ are matched by our assumption that there are no augmenting paths of length at most $2L + 1$.  So for every $\ell \in [0, L-1]$, we have $|T^{\ell}| \leq |S^{\ell+1}|$ as all vertices in $T^\ell$  are matched to $S^{\ell+1}$.
			
			Combining the two statements gives us $(1+\epsilon)|S^{\ell}| \leq |S^{\ell+1}|$ for every $\ell \in [0, L-1]$. Thus $|S^L| \geq (1+\epsilon)^L |S^0|$, which contradicts the definition of $L$ and that  $|S^0| \geq 1, |S^L| \leq |S|$. 
		\end{proof}
		This finishes the proof of Lemma~\ref{lemma:pbm}.
	\end{proof}
	
	With the lemma, we prove the following theorem using the grouping technique from \cite{PST95}:
	\begin{restatable}{theorem}{thmgrouping}
		\label{thm:integral-matching}
		Given $\bfx \in [0, 1]^E$ satisfying $x(\delta(j)) = 1$ for every $j \in J$, and  $\epsilon \in (0, 1)$, there is an $O\left(\frac{|E|}{\epsilon} \log n\right)$-time algorithm that outputs an assignment $\sigma \in M^J$ of jobs to machines such that $\sigma_jj \in E$ and $x_{\sigma_jj} > 0$ for every $j \in J$,  and for every $i \in M$, we have
		\begin{align*}
			\sum_{j \in \sigma^{-1}(i)}  p_{ij} \leq  (1+\epsilon) \sum_{j \in N(i)}  p_{ij} x_{ij} + \max_{j \in \sigma^{-1}(i)} p_{ij}. \quad  (\text{Assume the maximum over $\emptyset$ is 0.})
		\end{align*}
	\end{restatable}
	
	\begin{proof}
		We construct a bipartite graph $H = (V, J, E_H)$, starting with $V = \emptyset$ and $E_H = \emptyset$. For every machine $i \in M$, we run the following procedure. See Figure~\ref{fig:grouping} for an illustration. (The notations defined in the paragraph depend on $i$;  if a notation does not contain $i$ in the subscript, it will only be used locally, in this paragraph.) 
		Let $D_i$ be the number of jobs $j$ with positive $x_{ij}$ values.  Let $j_1, j_2, \cdots, j_{D_i}$ be these jobs  $j$, sorted in non-increasing order of $p_{ij}$; that is, we have $p_{ij_1} \geq p_{ij_2} \geq \cdots \geq p_{ij_{D_i}}$. For every integer $d \in [0, D_i]$, we define $Z_d = \sum_{d' = 1}^{d} x_{ij_{d'}}$.  Let $R_i = \ceil{(1+\epsilon)Z_{D_i}} = \ceil{(1+\epsilon)x(\delta(i))}$.  For every ${r} = 1, 2, 3,\cdots, R_i$, we create a vertex $ir$ and add it to $V$. We add to $E_H$ an edge between $ir, {r} \in [R_i]$ and $j_d$, $d \in [D_i]$  if $(\frac{{r}-1}{1+\epsilon}, \frac{{r}}{1+\epsilon})\cap(Z_{d-1}, Z_d) \neq \emptyset$, and we define $y_{(ir)j_d}$ to be the length of the interval.  This finishes the construction of $H = (V, J, E_H)$, along with a vector $\bfy \in \left(0, \frac{1}{1+\epsilon}\right]^{E_H}$. 
		\begin{figure}[ht]
			\centering
			\includegraphics[width=0.8\textwidth]{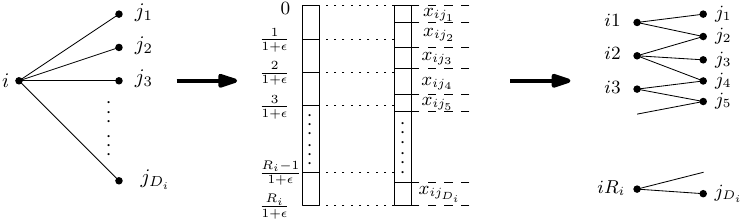}
			\caption{Construction of the $H$ for the machine $i \in M$. In the bipartite graph between $\{i1, i2, \cdots, iD_i\}$ and $\{j_1, j_2, \cdots, j_{D_i}\}$ and  there is an edge between $j_d$ and $(ir)$ iff the interval correspondent to $j_d$ intersects the interval $(\frac{r-1}{1+\epsilon}, \frac{r}{1+\epsilon})$. }
			\label{fig:grouping}
		\end{figure}
		
		The number of edges in $H$ for each $i$ is at most $D_i + R_i - 1 \leq |\delta(i)| + (1+\epsilon)x(\delta(i))$. Therefore the total number of edges we created in $H$ is at most $|E| + (1+\epsilon)|J|  = O(|E|)$. For every $ij  \in E$, we have $\sum_{r:(ir)j\in E_H}y_{(ir)j} = x_{ij}$. This implies that for every $j \in J$, we have $y(\delta_H(j)) = 1$. For every $ir\in V$, we have $y\big(\delta_H(ir)\big) \leq \frac{1}{1+\epsilon}$, and the inequality holds with equality except when $r = R_i$. 
		
		For every set $J' \subseteq J$, we have $|N_{H}(J')| \geq (1+\epsilon)|J'|$, as we can view $\bfy$ as a fractional matching in $H$ where every $j \in J$ is matched to an extent of 1 and every $ir \in V$ is matched to an extent of at most $\frac{1}{1+\epsilon}$. Then we can use Lemma~\ref{lemma:pbm} \footnote{We need to switch the left and right sides when going from the bipartite graph $H$ in Theorem~\ref{thm:integral-matching} to that in Lemma~\ref{lemma:pbm}. That is, we set $S = J$ and $T = V$. } to find a matching in $H$ that covers all jobs $J$.  The running time of the algorithm is $O\left(\frac{|E_H|}{\epsilon}\log n\right) = O\left(\frac{|E|}{\epsilon}\log n\right)$. The matching  gives an assignment $\sigma \in M^J$: If $j$ is matched to $ir$, then define $\sigma_j = i$.  Fix some $i \in M$ with $\sigma^{-1}(i) \neq \emptyset$; we upper bound $\sum_{j \in \sigma^{-1}(i)}  p_{ij}$:
		\begin{align*}
			&\quad \sum_{j \in \sigma^{-1}(i)}  p_{ij} \leq\quad \max_{j \in\sigma^{-1}(i)}p_{ij} + \sum_{{r} = 2}^{R_i} \max_{j \in N_H(ir)}p_{ij} \\
			&\leq\quad \max_{j \in\sigma^{-1}(i)}p_{ij} + (1+\epsilon)\sum_{{r} = 2}^{R_i} \sum_{j \in N_H(i(r-1))} p_{ij}y_{(i(r-1))j}\\
			&\leq\quad \max_{j \in\sigma^{-1}(i)}p_{ij} + (1+\epsilon)\sum_{{r} = 1}^{R_i} \sum_{j \in N_H(ir)} p_{ij}y_{(ir)j} 
			\ =\  \max_{j \in\sigma^{-1}(i)}p_{ij}+(1+\epsilon)\sum_{j \in N(i)}p_{ij}x_{ij}.
		\end{align*}
		To see the first inequality, notice that the job $j'$ matched to $i1$ (if it exists) has $p_{ij'} \leq \max_{j \in \sigma^{-1}(i)}p_{ij}$, and the job $j'$ matched to each $ir$, ${r} \in [2, R_i]$, has $p_{ij'} \leq \max_{j \in \delta_H((ir))}p_{ij}$. Consider the second inequality. For every ${r} \in [2, R_i]$, any $j \in \delta_H(ir)$ and any $j' \in  \delta_H(i(r-1))$, we have $p_{ij} \leq p_{i j'}$.  Moreover, for every ${r} \in [2,R_i]$, we have $y\big(\delta_H(i(r-1))\big) = \frac{1}{1+\epsilon}$. The inequality in the third line follows from replacing ${r}$ with ${r}+1$. The equality holds since for every $ij  \in E$ we have $\sum_{{r}:(ir)j \in E_H}y_{(ir)j} = x_{ij}$. 
	\end{proof}
	
	We can then apply Theorem~\ref{thm:integral-matching} with the solution $\bfx$ we obtained from solving LP(\ref{LPC:makespan}-\ref{LPC:makespan-non-negative}).  Clearly we have $\max_{j\in\sigma^{-1}(i)} p_{ij} \leq P$ for every $i \in M$. So, the total load on any machine $i$ is at most $P + (1+\epsilon)\cdot \sum_{j \in N(i)}  p_{ij} x_{ij} \leq P + (1+\epsilon)\cdot (1+O(\epsilon)) P = (2+O(\epsilon))P$, as \eqref{LPC:makespan} is satisfied with right side replaced by $(1+O(\epsilon))P$. This finishes the analysis of the algorithm for $R||C_{\max}$ and proves Theorem~\ref{thm:main-R-Cmax}.
	
	
	\section{Unrelated Machine Weighted Completion Time Minimization}
	\label{sec:unrelated-wC}
	
	In this section, we give our nearly-linear time algorithm for $R||\sum_j w_jC_j$, with an approximation ratio of  $1.45 + \epsilon$, matching the current best ratio of Im and Li \cite{IL23} achieved in polynomial time. Our result is based on formulating an LP relaxation that is equivalent to the rectangle LP introduced by  Li \cite{Li20}. The new LP relaxation has a nearly-linear size and the mixed packing and covering form; thus it can be solved in nearly-linear time using Theorem~\ref{thm:fast-LP-solver}.  We describe the rectangle LP (LP\eqref{LP:rectangle}), our new LP relaxation (LP\eqref{LP:RwC}) and show their equivalence in  Sections~\ref{subsec:rectangle-LP}, \ref{subsec:RwC-new-LP} and \ref{subsec:R-wC-equivalance} respectively. 
	
	In \ifdefined\CR the full version of the paper \else Appendix~\ref{sec:R-wC-rounding} \fi 
	we show how to construct a solution to LP\eqref{LP:rectangle} from one to LP\eqref{LP:RwC} in nearly-linear time, finishing the proof of Theorem~\ref{thm:main-R-wC-LP}. We also show in
	\ifdefined\CR the full version \else Appendix~\ref{sec:R-wC-rounding} \fi
	that the rounding algorithm of Im and Li can run in nearly-linear time; this finishes the proof of Theorem~\ref{thm:main-R-wC}. Throughout the section, we assume all processing times are integers bounded by a polynomial of $n$. The general case is handled in
	\ifdefined\CR the full version\else Appendix~\ref{subsec:wc-big-p-w}\fi.
	
	\subsection{Rectangle LP Relaxation}
	\label{subsec:rectangle-LP}
	We describe the rectangle LP relaxation for $R||\sum_j w_jC_j$ introduced by Li \cite{Li20}.  Let $T= \sum_{j \in J} \max_{i \in N(j)} p_{ij}$ so that any schedule will complete by time $T$. The following is the rectangle LP:
	\begin{align}
		\min \qquad \sum_{j \in J}w_j \sum_{i \in N(j), s \in [0, T)} z_{ijs}(s+p_{ij}) \label{LP:rectangle}
	\end{align}\vspace*{-20pt}
	
	\noindent\begin{minipage}{0.54\textwidth}
		\begin{align}
			\sum_{i \in N(j), s\in [0, T)}z_{ijs} &\geq 1 & &\forall j \in J \label{LPC:rectangle-scheduled}\\
			\sum_{j \in N(i), s \in [t - p_{ij}, t)}z_{ijs} &\leq 1 & &\forall i \in M, t \in [T] \label{LPC:rectangle-capacity}
		\end{align}
	\end{minipage}\hfill
	\begin{minipage}{0.44\textwidth}
		\begin{align}
			z_{ijs} &= 0 & &\forall ij \in E, s > T - p_{ij} \label{LPC:rectangle-no-late} \\[15pt]
			z_{ijs} &\geq 0	& &\forall ij \in E, s \in [0, T) \label{LPC:rectangle-non-negative} \\[3pt]\nonumber
		\end{align}
	\end{minipage}\smallskip
	
	In the correspondent integer program, $z_{ijs}$ for every $ij \in E$ and integer $s \in [0, T)$ indicates if job $j$ is scheduled on machine $i$, with starting time $s$.  The objective gives the weighted completion time of the schedule.  \eqref{LPC:rectangle-scheduled} requires that every job $j$ is scheduled. \eqref{LPC:rectangle-capacity} requires that at any time on machine $i$, at most one job is being processed. \eqref{LPC:rectangle-no-late} ensures that no jobs complete after time $T$. \eqref{LPC:rectangle-non-negative} is the non-negativity constraint.  Im and Li \cite{IL23} showed that given a solution $\bfz$ to LP\eqref{LP:rectangle}, one can round it to an integral schedule, whose weighted completion time in expectation is at most $1.45$ times the value of $\bfz$.
	
	\subsection{A Nearly-Linear Size LP Relaxation}
	\label{subsec:RwC-new-LP}
	
	In this section we formulate the relaxation that can be solved in nearly-linear time, and prove its equivalence to LP\eqref{LP:rectangle} in Section~\ref{subsec:R-wC-equivalance}. We create a list of time points as follows: $T_0 = 0$,  $T_d = \floor{(1+\epsilon)T_{d-1}}+1$ for every integer $d \geq 1$.   Define $D = O(\frac{\log n}{\epsilon})$ to be the smallest integer so that $T_D \geq T$.  We call $(T_{d-1}, T_d]$ the $d$-th \emph{window}, and the time points $T_0, T_1, \cdots, T_D$ \emph{window boundaries} (or simply boundaries). Define $\Delta_d = T_d - T_{d-1}$ to be the length of the $d$-th window.
	
	Let $\eta_d := \ceil{\epsilon \Delta_d}$. We partition $(T_{d-1}, T_d]$ into \emph{sub-windows} of length $\eta_d$, except that the last sub-window may be shorter.  Then $q_d := \ceil{\frac{\Delta_d}{\eta_d}} \leq \frac1\epsilon$ is the number of sub-windows of $(T_{d-1}, T_d]$. Let $\tau^{(d)}_0 = T_{d-1}, \tau^{(d)}_1, \tau^{(d)}_2, \cdots, \tau^{(d)}_{q_d} = T_d$ be the boundaries of the $q_d$ sub-windows. 
	
	We describe the variables in the LP.  For every $ij \in E$ and $d \in [D]$ with $p_{ij} \leq \Delta_d$, we introduce a variable $x_{ijd}$, indicating if $j$ is scheduled on $i$ inside the $d$-th window.   
	Let $S_j$ and $C_j$ be the starting and completion time of $j$ in the target optimum schedule (which the algorithm does not know).  For every $ij \in E, 1 \leq d \leq e < D$, integers $u \in [0, q_d), v \in [0, q_{e+1})$, we \emph{may} introduce a variable $y_{ijdeuv}$ indicating if $j$ is scheduled on $i$, $S_j \in [\tau^{(d)}_u, \tau^{(d)}_{u+1})$ and $C_j \in (\tau^{(e+1)}_v, \tau^{(e+1)}_{v+1}]$.  That means, the scheduling interval $(S_j, C_j]$ of $j$ contains the $d'$-th window for every $d' \in [d+1, e]$, and a non-empty part of the $d$-th and $(e+1)$-th windows. $u$ and $v$ approximately give the volumes of $j$ processed in the two windows. It is disjoint from all other windows.  As a hindsight, a sub-window is short enough and we can afford to incur an error equaling its length for every window. We only introduce a $y$-variable if the correspondent event can happen. That is, the following conditions need to be satisfied for the existence of $y_{ijdeuv}$: $\tau^{(e+1)}_v - \tau^{(d)}_{u+1} + 2 \leq p_{ij}\leq \tau^{(e+1)}_{v+1} - \tau^{(d)}_u$.
	Notice when $y_{ijdeuv} = 1$, then the scheduling interval $(S_j, C_j]$ of $j$ intersects at least two windows.
	
	For a variable $y_{ijdeuv}$ and an integer $d' \in [D]$, we define
	\begin{align*}
		Q_{ijdeuvd'} := \begin{cases}
			0 & \text{if } d'  \leq d-1 \text{ or } d' \geq e+2\\
			\Delta_{d'} & \text{if } d+1 \leq d' \leq e\\
			T_d - \tau^{(d)}_{u+1} + 1 & \text{if } d' = d\\
			\tau^{(e+1)}_v - T_e + 1 & \text{if } d' = e+1
		\end{cases}.
	\end{align*}
	Assuming $j$ starts at time $\tau^{(d)}_{u+1}-1$ and ends at time $\tau^{(e+1)}_v + 1$ on machine $i$, we have that $Q_{ijdeuvd'}$ is the volume of job $j$ processed in the $d'$-th window. So, if $y_{ijdeuv} = 1$ in a schedule,  then $Q_{ijdeuvd'}$ gives a lower bound on the volume.
	
	We say the quadruple $deuv$ left-covers the pair $d'u'$ if the sub-window $(\tau^{(d')}_{u'-1}, \tau^{(d')}_{u'}]$ is between the sub-windows $(\tau^{(d)}_u, \tau^{(d)}_{u+1}]$ (exclusive) and $(\tau^{(e+1)}_v, \tau^{(e+1)}_{v+1}]$ (inclusive) in the time horizon, or if $(\tau^{(d')}_{u'-1}, \tau^{(d')}_{u'}]= (\tau^{(d)}_u, \tau^{(d)}_{u+1}]$  and $\tau^{(d')}_{u'} - \tau^{(d')}_{u'-1}  = 1$. So if $deuv$ left-covers $d'u'$ and $y_{ijdeuv} = 1$, then the scheduling interval of $j$ will surely cover the left-most time unit of the sub-window $(\tau^{(d')}_{u'-1}, \tau^{(d')}_{u'}]$.
	
	With the necessary definitions, we can formulate the LP relaxation as LP\eqref{LP:RwC}. For the sake of convenience, we assume if a variable does not exist, then it is not included in a summation.  
	\begin{equation}
		\min \qquad \sum_{ijd} w_j \cdot (T_{d-1}+1) \cdot x_{ijd} + \sum_{ijdeuv} w_j \cdot (T_e+1) \cdot y_{ijdeuv} \label{LP:RwC}
	\end{equation}\vspace*{-10pt}
		\begin{align}
			\sum_{id} x_{ijd} + \sum_{ideuv} y_{ijdeuv} &\geq 1 &\quad & \forall j \in J \label{LPC:RwC-scheduled}\\
			\sum_{\substack{jdeuv\ :\ deuv\\\text{left-covers } d'u'}} y_{ijdeuv} &\leq 1 &\quad &\forall i \in M, d' \in [D], u' \in [q_{d'}] \label{LPC:RwC-capacity-cross}\\
			\sum_{j} p_{ij} \cdot x_{ijd'} + \sum_{jdeuv}Q_{ijdeuvd'}\cdot y_{ijdeuv} &\leq \Delta_{d'} &\quad &\forall i \in M, d' \in [D]\label{LPC:RwC-capacity-both}
		\end{align}\vspace*{-15pt}
		\begin{align}
			\text{all variables are non-negative} \label{LPC:RwC-non-negative}
		\end{align}
		
		Consider the correspondent integer program and an integral schedule. If $x_{ijd} = 1$, then the completion time of $j$ is in $(T_{d-1}, T_d]$. If $y_{ijdeuv} = 1$, then it is in $(T_e, T_{e+1}]$. So, the objective \eqref{LP:RwC} approximates and underestimates the total weighted completion time of the schedule.\footnote{A more precise estimation for the case $y_{ijdeuv}=1$ is $\tau^{(e+1)}_v+1$. But the estimation $T_e + 1$ is good enough.} \eqref{LPC:RwC-scheduled} requires that every job is scheduled: either the scheduling interval of a job $j$ is inside some window, or it overlaps with at least two windows.  \eqref{LPC:RwC-capacity-cross} follows from the definition of $deuv$ left-covering $d'u'$.  If $x_{ijd'} = 1$, then the $p_{ij}$ units of job $j$ is processed in the $d'$-th window on machine $i$. If $y_{ijdeuv} = 1$, then at least $Q_{ijdeuvd'}$ units is processed. So \eqref{LPC:RwC-capacity-both} is valid as the volume of the jobs processed in the $d'$-th window is at most $\Delta_{d'}$.  Therefore, LP\eqref{LP:RwC} is valid, and its value is at most the weighted completion time of the optimum schedule for the instance.
		
		There are at most $D|E|$ $x$-variables. We then count the number of tuples $ijdeuv$ such that $y_{ijdeuv}$ is a variable in the LP. For fixed $ij \in E, d \in D$ and $u \in [0, q_d)$, there are at most $O(1)$ possibilities for $(e, v)$, as the lengths of sub-windows do not decrease from left to right in the time horizon, except for the last sub-window of each window.  Hence the number of $y$-variables is at most $O\big(\frac{D|E|}{\epsilon}\big) = O\big(\frac{|E|\log n}{\epsilon^2}\big)$.\footnote{By cutting job lengths $p_{ij}$ by a factor of $\epsilon$, one can reduce the number of $y$ variables to $O\Big(|E|\big(\frac{1}{\epsilon^2} + \frac{\log n}{\epsilon}\big)\Big)$. But we prioritize on giving a clean algorithm, rather than optimizing the $\poly(\log n, \frac1\epsilon)$-factor in the running time.}  The number of constraints is $O(n + \frac{m|D|}{\epsilon}) = O\big(n + \frac{m\log n}{\epsilon^2}\big)$. The number of non-zeros is at most $O\big(\frac{|E|\log^2 n}{\epsilon^4}\big)$ as each variable appears in at most $O(\frac{D}{\epsilon})$ constraints.  
		
		Therefore, by Theorem~\ref{thm:fast-LP-solver}, in $O\left(\frac{|E|\log^3 n}{\epsilon^6}\right) = \tilde O_\epsilon(|E|)$ time, we can find a solution $(\bfx, \bfy)$ satisfying the following conditions: Its cost is at most $1+\epsilon$ times that of the optimum solution to the LP, all variables are non-negative, \eqref{LPC:RwC-scheduled} holds with equalities, and \eqref{LPC:RwC-capacity-cross} and \eqref{LPC:RwC-capacity-both} hold with right sides replaced by $1+\epsilon$ and $(1+\epsilon)\Delta_{d'}$ respectively.  For convenience, we call such a solution a $(1+\epsilon)$-approximate solution to LP\eqref{LP:RwC}; but keep in mind that it may violate \eqref{LPC:RwC-capacity-cross} and \eqref{LPC:RwC-capacity-both} by a factor of $1+\epsilon$. 
		
		\subsection{Equivalence of LP\eqref{LP:RwC} and LP\eqref{LP:rectangle}}
		\label{subsec:R-wC-equivalance}
		We use $\lp_{\eqref{LP:rectangle}}$ and $\lp_{\eqref{LP:RwC}}$ to denote the values of LP\eqref{LP:rectangle} and LP\eqref{LP:RwC} respectively.  It is easy to show that $\lp_{\eqref{LP:RwC}} \leq \lp_{\eqref{LP:rectangle}}$, as one can convert a solution to LP\eqref{LP:rectangle} into one to LP\eqref{LP:RwC} with smaller or equal value.  The following theorem gives the other direction, proving the equivalence of the two LPs up to a $1 + O(\epsilon)$ factor:
		\begin{theorem}[Equivalence of LP\eqref{LP:RwC} and LP\eqref{LP:rectangle}]
			\label{thm:RwC-exists-zcirc}
			Let $(\bfx, \bfy)$ be a $(1+\epsilon)$-approximate solution to LP\eqref{LP:RwC}. 
			%
			%
			Then in $\tilde O_\epsilon(|E|)$-time we can find a solution $\bfz$ to LP\eqref{LP:rectangle} except that \eqref{LPC:rectangle-capacity} is only satisfied with the right-side replaced by $1+\epsilon$,  such that the following is true for an absolute constant $c \geq 1$, every $ij \in E$ and integer $t \geq 0$:
			\begin{align*}
				\sum_{s + p_{ij} > (1+c\epsilon)t} z_{ijs} \leq \sum_{d: T_{d-1} + 1 >  t}x_{ijd} + \sum_{deuv: T_e + 1 > t} y_{ijdeuv}.
			\end{align*}
		\end{theorem}

		In words, for every $ij \in E$, and every time $t \geq 0$, the fraction of job $j$ scheduled on $i$ with completion time after $(1+c\epsilon)t$ in $\bfz$ is at most the fraction with completion time after $t$ in $(\bfx, \bfy)$. Then, the following corollary is immediate: 
		\begin{coro}
			Let $(\bfx, \bfy)$ and $\bfz$ be defined as in Theorem~\ref{thm:RwC-exists-zcirc}. Then
			the value of  $\bfz$ to LP\eqref{LP:rectangle} is at most $1+c\epsilon$ times that of $(\bfx, \bfy)$ to LP\eqref{LP:RwC}. This implies that the value of $\bfz$ to LP\eqref{LP:rectangle} is at most $(1+c\epsilon)(1+\epsilon)\cdot\lp_{\eqref{LP:RwC}} \leq (1+O(\epsilon))\cdot\lp_{\eqref{LP:rectangle}}$.
		\end{coro}
		
		To better present the ideas behind the proof, we only show the existence of such a vector $\bfz$ in this section. That is, we are not concerned with the running time of the algorithm that constructs $\bfz$.  In  
		\ifdefined\CR the full version of the paper \else Appendix~\ref{subsec:constructing-z} \fi
		we show how $\bfz$ can be constructed in nearly-linear time. 
		
		So the rest of this section is dedicated to proving the existence of $\bfz$ satisfying the conditions in  Theorem~\ref{thm:RwC-exists-zcirc}. 
		Till the end, we fix the solution $(\bfx, \bfy)$. We assume all variables in $(\bfx, \bfy)$ have values being integer multiplies of $1/\Phi$, and $(1+\epsilon)\Phi$ is an integer, for a large enough integer $\Phi > 0$. We fix a machine $i \in M$ and show how to construct the $\bfz$ values for this $i$.   We create $(1+\epsilon)\Phi$ \emph{mini-machines}, each serving as $1/\Phi$ fraction of the machine $i$. We create two types of \emph{mini-jobs}: 
		\begin{itemize}
			\item For every variable $x_{ijd}$ with positive value, we create $\Phi x_{ijd}$ mini-jobs of length $p_{ij}$; we call them \emph{inside-mini-jobs}.  Each such inside-mini-job has an \emph{intended completion time} of $T_{d-1} + 1$; this is the estimation used in the objective \eqref{LP:RwC}.  
			\item For every variable $y_{ijdeuv}$ with positive value, we create $\Phi y_{ijdeuv}$ mini-jobs of length $\tau^{(e+1)}_v - \tau^{(d)}_{u+1}+2$; we call them \emph{cross-mini-jobs}. Notice the length may be smaller than $p_{ij}$.  Similarly, the cross-mini-jobs have an \emph{intended completion time} of $T_e  + 1$. We define the \emph{blocking interval} of these mini-jobs to be the union of the sub-windows $(\tau^{(d')}_{u'-1}, \tau^{(d')}_{u'}]$ such that $deuv$ left-covers $d'u'$. This is indeed an interval. As \eqref{LPC:RwC-capacity-cross} holds with right side replaced by $1+\epsilon$, every time point is covered by blocking intervals of at most $(1+\epsilon)\Phi$ cross-mini-jobs.
		\end{itemize}
		
		Our goal becomes to schedule all the mini-jobs on the $(1+\epsilon)\Phi$ mini-machines integrally, guaranteeing that the completion time of each mini-job is at most $1 + 5\epsilon$ times its intended completion time (after we extend the lengths of cross-mini-jobs). The construction of the schedule is given in Algorithm~\ref{alg:schedule-on-mini-machines}; recall that we are not concerned with the running time in this proof. 	 	The solution $\bfz$ to LP\eqref{LP:rectangle} will be the integral schedule scaled by a factor of $\frac{1}{\Phi}$: $z_{ijs}$ is $\frac{1}{\Phi}$ times the number of mini-jobs for $j$ that start at time $s$ in the schedule. 
		
		\begin{algorithm}[h]
			\caption{Scheduling of mini-jobs on mini-machines for a machine $i \in M$.}
			\label{alg:schedule-on-mini-machines}
			\begin{algorithmic}[1]
				\State define a vector $\sigma:\text{cross-mini-jobs} \to \text{mini-machines}$, so that for every mini-machine $h$, the blocking intervals of $\sigma^{-1}(h)$ are disjoint. \label{step:mini-schedule-sigma}
				\For{$d' \gets 1$ to $D$} \label{step:mini-schedule-loop}
				\For{every cross-mini-job $k$ for some variable $y_{ijdeuv}$ with $d \leq d' \leq e + 1$} \label{step:mini-schedule-inner-loop-cross}
				\State $\load_{\sigma_k} \gets \load_{\sigma_k} + Q_{ijdeuvd'}$ \label{step:mini-schedule-cross-load}
				\If{$d' = e + 1$} append $k$ to the mini-machine $\sigma_k$ \EndIf \label{step:mini-schedule-cross-append}
				\EndFor
				\For{every inside-mini-job $k$ for some variable $x_{ijd'}$} \label{step:mini-schedule-inner-loop-inside}
				\State let $h$ be the mini-machine with the smallest $\load_h$ \label{step:mini-schedule-choose-inside}
				\State $\load_h \gets \load_h + p_{ij}$, append $k$ to the mini-machine $h$ \label{step:mini-schedule-assign-inside}
				\EndFor
				\EndFor
				\State \label{step:mini-schedule-extend} extend the length of each cross-mini-job for a variable $y_{ijdeuv}$ to $p_{ij}$ in the constructed schedule
			\end{algorithmic}
		\end{algorithm}
		
		Step~\ref{step:mini-schedule-sigma} of Algorithm~\ref{alg:schedule-on-mini-machines} is possible since each point is covered by at most $(1+\epsilon)\Phi$ blocking intervals. When we schedule an inside-mini-job $k$ on a mini-machine $h$, we increase $\load_h$ by the length of $k$ (Step~\ref{step:mini-schedule-assign-inside}). The scheduling of a cross-mini-job $k$ for some variable $y_{ijdeuv}$ is done differently. First the mini-machine $\sigma_k$ for $k$ is pre-defined. Second, we append $k$ to $\sigma_k$ only in iteration $d' = e+1$ (Step~\ref{step:mini-schedule-cross-append}), but we add the length of $k$ to $\load_{\sigma_k}$ piece by piece: In iterations $d' = d, d+1, \cdots, e+1$, we increase the load by $Q_{ijdeuvd'}$ (Step~\ref{step:mini-schedule-cross-load}).  Still we ensure that the load to $\sigma_k$ contributed by $k$ is equal to the length of $k$.  A mini-job for a variable $y_{ijdeuv}$ may have length smaller than the desired length $p_{ij}$, so in Step~\ref{step:mini-schedule-extend} we extend these mini-jobs.
		
		\begin{restatable}{lemma}{lemmaloadadded}
			\label{lemma:load-added}
			At the end of iteration $d'$ of Loop~\ref{step:mini-schedule-loop}, every mini-machine has a load of at most $T_{d'} - 1 + \Delta_{d'}$.
		\end{restatable}
		\begin{proof}
			There are two types of loads added to mini-machines during iteration $d'$ of Loop~\ref{step:mini-schedule-loop}: those from cross-mini-jobs, and those from inside-mini-jobs. The total load (from both cross- and inside-mini-jobs) added to all mini-machines is at most $(1+\epsilon)\Phi\Delta_{d'}$:  it is precisely $\Phi$ times the left-side of \eqref{LPC:RwC-capacity-both} for the machine $i$ and $d'$, which is at most $(1+\epsilon)\Phi \Delta_{d'}$ as the constraint is violated only by a factor of $1+\epsilon$.
			
			The total load from cross-mini-jobs added to a mini-machine $h$ in iteration $d'$ is at most $\Delta_{d'}$ as the blocking intervals of all mini-jobs in $\sigma^{-1}(h)$ are disjoint. We need to check the case when one mini-job $k \in \sigma^{-1}(h)$ has blocking interval ending at $\tau^{(d')}_{u'}$ and another mini-job $k' \in \sigma^{-1}(h)$ has blocking interval starting at the time. If the length of the sub-window $(\tau^{(d')}_{u'-1}, \tau^{(d')}_{u'}]$ is at least 2,  then the statement holds as we only gave 1 unit length to $k$ and $k'$ in this sub-window. If the length is 1, then because we handled the case in a special way in the definition of left-covering, we did not give any length to $k'$ for the sub-window.
			
			With the observations, we can prove the lemma. Before we add an inside-mini-job $k$ for $x_{ijd'}$ to a mini-machine $h$ in iteration $d'$, the total load of all mini-machines is strictly less than $(1+\epsilon)\Phi\sum_{d'' = 1}^{d'} \Delta_{d''} = (1+\epsilon)\Phi T_{d'}$ (as the length of $k$ has not been added to the loads yet). Therefore $\load_h<T_{d'}$ before we append $k$ to $h$. After that, we have $\load_h \leq T_{d'} - 1 + p_{ij} \leq T_{d'}-1+\Delta_{d'}$.
			
			Assume towards the contradiction that the lemma does not hold and consider the first time when the condition is violated. Assume this is at iteration $d'$, and some mini-machine has a load of at least $T_{d'} + \Delta_{d'}$. This must be due to that we add the loads from cross-mini-jobs to the machine.  By our assumption, every mini-machine has a load of at most $T_{d'-1}-1+\Delta_{d'-1}$ at the end of iteration $d'-1$. (A special case is when $d' = 1$; but this can be handled trivially.)  As we argued, we add a load of at most $\Delta_{d'}$ from cross-mini-jobs to each mini-machine $i$n iteration $d'$. Therefore after we add the loads, every mini-machine has a load of at most $T_{d' - 1} - 1 + \Delta_{d' - 1} + \Delta_{d'} = T_{d'} -1 + \Delta_{d' - 1} \leq T_{d'} - 1 + \Delta_{d'}$, a contradiction. 
		\end{proof}
		
		Now we consider how Step~\ref{step:mini-schedule-extend} changes the completion times. The length of a cross-mini-job for a variable $y_{ijdeuv}$ is increased by at most $\eta_d - 1 + \eta_{e+1}-1 \leq \epsilon \Delta_d + \epsilon \Delta_{e+1} \leq 2\epsilon(\Delta_d + \Delta_e)$.  For all cross-mini-jobs assigned to the same mini-machine $h$, the correspondent intervals $\{d, d+1, \cdots, e\}$ are disjoint.  Therefore, a mini-job scheduled in iteration $d'$ of Loop~\ref{step:mini-schedule-loop} is delayed by at most $2\epsilon(\Delta_1 + \Delta_2 + \cdots + \Delta_{d'}) = 2\epsilon T_{d'}$ units time. In the final schedule constructed by Algorithm~\ref{alg:schedule-on-mini-machines} the completion time of a mini-job scheduled in iteration $d$ is at most 
		\begin{align*}
			&T_d  - 1 + \Delta_d + 2\epsilon T_{d} \leq T_d - 1 + ((1+\epsilon) T_{d-1} + 1) - T_{d-1} + 2\epsilon T_{d} \\
			&= T_d  + \epsilon T_{d-1} + 2\epsilon T_d \leq (1+5\epsilon)(T_{d-1}+1).
		\end{align*}
		Setting $c = 5$, Theorem~\ref{thm:RwC-exists-zcirc} follows from that $T_{d-1} + 1$ is the intended completion time of the mini-job.

	\bibliographystyle{plain}
	\bibliography{reflist}

\begin{thebibliography}{10}

\bibitem{AO15}
Zeyuan Allen-Zhu and Lorenzo Orecchia.
\newblock Nearly-linear time positive {LP} solver with faster convergence rate.
\newblock In {\em Proceedings of the Forty-Seventh Annual ACM Symposium on
  Theory of Computing (STOC 2015)}, page 229–236.

\bibitem{Als22}
M{\aa}ns Alskog.
\newblock Implementation of a fast approximation algorithm for precedence
  constrained scheduling.
\newblock Master's thesis, Linköping University, Applied Mathematics, Faculty
  of Science and Engineering, 2022.

\bibitem{AHK12}
Sanjeev Arora, Elad Hazan, and Satyen Kale.
\newblock The multiplicative weights update method: a meta-algorithm and
  applications.
\newblock {\em Theory of Computing}, 8(6):121--164, 2012.

\bibitem{BK09}
Nikhil Bansal and Subhash Khot.
\newblock Optimal long code test with one free bit.
\newblock In {\em Proceedings of the 50th Annual IEEE Symposium on Foundations
  of Computer Science (FOCS 2009)}, pages 453--462.

\bibitem{BSS16}
Nikhil Bansal, Aravind Srinivasan, and Ola Svensson.
\newblock Lift-and-round to improve weighted completion time on unrelated
  machines.
\newblock In {\em Proceedings of the Forty-eighth Annual ACM Symposium on
  Theory of Computing (STOC 2016)}, pages 156--167.

\bibitem{BG21}
Yair Bartal and Lee-Ad Gottlieb.
\newblock Near-linear time approximation schemes for steiner tree and forest in
  low-dimensional spaces.
\newblock In {\em Proceedings of the 53rd Annual ACM SIGACT Symposium on Theory
  of Computing (STOC 2021)}, page 1028–1041.

\bibitem{BGS21}
A.~Bernstein, M.~Gutenberg, and T.~Saranurak.
\newblock Deterministic decremental sssp and approximate min-cost flow in
  almost-linear time.
\newblock In {\em Proceedings of the 62nd Annual IEEE Symposium on Foundations
  of Computer Science (FOCS 2021)}, pages 1000--1008.

\bibitem{BLS20}
Jan van~den Brand, Yin~Tat Lee, Aaron Sidford, and Zhao Song.
\newblock Solving tall dense linear programs in nearly linear time.
\newblock In {\em Proceedings of the 52nd Annual ACM SIGACT Symposium on Theory
  of Computing (STOC 2020)}, page 775–788.

\bibitem{CHQ20}
Chandra Chekuri, Sariel Har-Peled, and Kent Quanrud.
\newblock Fast lp-based approximations for geometric packing and covering
  problems.
\newblock In {\em Proceedings of the Thirty-First Annual ACM-SIAM Symposium on
  Discrete Algorithms (SODA 2020)}, page 1019–1038.

\bibitem{CJV15}
Chandra Chekuri, T.S. Jayram, and Jan Vondrak.
\newblock On multiplicative weight updates for concave and submodular function
  maximization.
\newblock In {\em Proceedings of the 2015 Conference on Innovations in
  Theoretical Computer Science (ITCS 2015)}, page 201–210.

\bibitem{CQ17}
Chandra Chekuri and Kent Quanrud.
\newblock Near-linear time approximation schemes for some implicit fractional
  packing problems.
\newblock In {\em Proceedings of the Twenty-Eighth Annual ACM-SIAM Symposium on
  Discrete Algorithms (SODA 2017)}, page 801–820.

\bibitem{CQ18}
Chandra Chekuri and Kent Quanrud.
\newblock Randomized {MWU} for positive {LP}s.
\newblock In {\em Proceedings of the Twenty-Ninth Annual ACM-SIAM Symposium on
  Discrete Algorithms (SODA 2018)}, page 358–377.

\bibitem{CQ18a}
Chandra Chekuri and Kent Quanrud.
\newblock Fast approximations for metric-tsp via linear programming.
\newblock {\em ArXiv}, abs/1802.01242, 2018.

\bibitem{CKL22}
Li~Chen, Rasmus Kyng, Yang~P. Liu, Richard Peng, Maximilian~Probst Gutenberg,
  and Sushant Sachdeva.
\newblock Maximum flow and minimum-cost flow in almost-linear time.
\newblock In {\em Proceedings of the 63rd Annual IEEE Symposium on Foundations
  of Computer Science (FOCS 2022)}, pages 612--623.

\bibitem{CGL20}
Julia Chuzhoy, Yu~Gao, Jason Li, Danupon Nanongkai, Richard Peng, and
  Thatchaphol Saranurak.
\newblock A deterministic algorithm for balanced cut with applications to
  dynamic connectivity, flows, and beyond.
\newblock In Sandy Irani, editor, {\em Proceedings of the 61st Annual {IEEE}
  Annual Symposium on Foundations of Computer Science (FOCS 2020)}, pages
  1158--1167.

\bibitem{Din70}
E.~A. Dinic.
\newblock Algorithm for solution of a problem of maximal flow in a network with
  power estimation.
\newblock {\em Doklady Akademii Nauk SSSR}, 194(4):1277–1280, 1970.

\bibitem{GK07}
Naveen Garg and Jochen Könemann.
\newblock Faster and simpler algorithms for multicommodity flow and other
  fractional packing problems.
\newblock {\em SIAM Journal on Computing}, 37(2):630--652, 2007.

\bibitem{Garg18}
Shashwat Garg.
\newblock Quasi-{PTAS} for scheduling with precedences using {LP} hierarchies.
\newblock In {\em Proceedings of 45th International Colloquium on Automata,
  Languages, and Programming (ICALP 2018)}, pages 59:1--59:13.

\bibitem{Gra69}
R.~L. Graham.
\newblock Bounds on multiprocessing timing anomalies.
\newblock {\em Siam Journal on Applied Mathematics}, 17(2):416--429, 1969.

\bibitem{GLLR79}
R.~L. Graham, E.~L. Lawler, J.~K. Lenstra, and A.~H. G.~Rinnooy Kan.
\newblock {Optimization and approximation in deterministic sequencing and
  scheduling: a survey}.
\newblock {\em Ann. Discrete Math.}, 4:287--326, 1979.

\bibitem{HSS97}
Leslie~A. Hall, Andreas~S. Schulz, David~B. Shmoys, and Joel Wein.
\newblock Scheduling to minimize average completion time: Off-line and on-line
  approximation algorithms.
\newblock {\em Math. Oper. Res.}, 22(3):513--544, August 1997.

\bibitem{HK73}
John~E. Hopcroft and Richard~M. Karp.
\newblock An $n^{5/2}$ algorithm for maximum matchings in bipartite graphs.
\newblock {\em SIAM Journal on Computing}, 2(4):225--231, 1973.

\bibitem{IL23}
Sungjin Im and Shi Li.
\newblock Improved approximations for unrelated machine scheduling.
\newblock In {\em Proceedings of the Thirty-Fourth Annual ACM-SIAM Symposium on
  Discrete Algorithms (SODA 2023)}, pages 2917--2946.

\bibitem{IS20}
Sungjin Im and Maryam Shadloo.
\newblock Weighted completion time minimization for unrelated machines via
  iterative fair contention resolution [extended abstract].
\newblock In {\em Proceedings of the Thirty-First Annual ACM-SIAM Symposium on
  Discrete Algorithms (SODA 2020)}, page 2790–2809.

\bibitem{JR17}
Klaus Jansen and Lars Rohwedder.
\newblock On the configuration-{LP} of the restricted assignment problem.
\newblock In {\em Proceedings of the Twenty-Eighth Annual ACM-SIAM Symposium on
  Discrete Algorithms (SODA 2017)}, page 2670–2678.

\bibitem{JR20}
Klaus Jansen and Lars Rohwedder.
\newblock A quasi-polynomial approximation for the restricted assignment
  problem.
\newblock {\em SIAM Journal on Computing}, 49(6):1083--1108, 2020.

\bibitem{KLO14}
Jonathan~A. Kelner, Yin~Tat Lee, Lorenzo Orecchia, and Aaron Sidford.
\newblock An almost-linear-time algorithm for approximate max flow in
  undirected graphs, and its multicommodity generalizations.
\newblock In {\em Proceedings of the Twenty-Fifth Annual ACM-SIAM Symposium on
  Discrete Algorithms (SODA 2014)}, page 217–226.

\bibitem{KY13}
Christos Koufogiannakis and N.~Young.
\newblock A nearly linear-time ptas for explicit fractional packing and
  covering linear programs.
\newblock {\em Algorithmica}, 70:648--674, 2013.

\bibitem{KY07}
Christos Koufogiannakis and Neal~E. Young.
\newblock Beating simplex for fractional packing and covering linear programs.
\newblock In {\em 48th Annual IEEE Symposium on Foundations of Computer Science
  (FOCS 2007)}, pages 494--504.

\bibitem{LS15}
Yin~Tat Lee and Aaron Sidford.
\newblock Efficient inverse maintenance and faster algorithms for linear
  programming.
\newblock In {\em 2015 IEEE 56th Annual Symposium on Foundations of Computer
  Science (FOCS 2015)}, pages 230--249.

\bibitem{LSZ19}
Yin~Tat Lee, Zhao Song, and Qiuyi Zhang.
\newblock Solving empirical risk minimization in the current matrix
  multiplication time.
\newblock pages 2140--2157.

\bibitem{LR78}
J.~K. Lenstra and A.~H.~G. Rinnooy~Kan.
\newblock Complexity of scheduling under precedence constraints.
\newblock {\em Oper. Res.}, 26(1):22--35, 1978.

\bibitem{LST90}
Jan~Karel Lenstra, David~B. Shmoys, and {\'E}va Tardos.
\newblock Approximation algorithms for scheduling unrelated parallel machines.
\newblock {\em Mathematical Programming}, 46:259--271, 1990.

\bibitem{LR21}
Elaine Levey and Thomas Rothvo.
\newblock A (1+$\epsilon$)-approximation for makespan scheduling with
  precedence constraints using lp hierarchies.
\newblock {\em SIAM Journal on Computing}, 50(3):STOC16--201--STOC16--217,
  2021.

\bibitem{Li21}
Jason Li.
\newblock Deterministic mincut in almost-linear time.
\newblock In {\em Proceedings of the 53rd Annual ACM SIGACT Symposium on Theory
  of Computing (STOC 2021)}, page 384–395.

\bibitem{Li20}
Shi Li.
\newblock Scheduling to minimize total weighted completion time via
  time-indexed linear programming relaxations.
\newblock {\em SIAM Journal on Computing}, 49(4):FOCS17--409--FOCS17--440,
  2020.

\bibitem{LN93}
Michael Luby and Noam Nisan.
\newblock A parallel approximation algorithm for positive linear programming.
\newblock In {\em Proceedings of the Twenty-Fifth Annual ACM Symposium on
  Theory of Computing (STOC 1993)}, page 448–457.

\bibitem{MQS98}
Alix Munier, Maurice Queyranne, and Andreas~S. Schulz.
\newblock Approximation bounds for a general class of precedence constrained
  parallel machine scheduling problems.
\newblock In {\em Integer Programming and Combinatorial Optimization (IPCO
  1998)}, pages 367--382, 1998.

\bibitem{Pen16}
Richard Peng.
\newblock Approximate undirected maximum flows in ${O}(m \text{polylog}(n))$
  time.
\newblock In {\em Proceedings of the Twenty-Seventh Annual ACM-SIAM Symposium
  on Discrete Algorithms (SODA 2016)}, page 1862–1867.

\bibitem{PST95}
Serge~A. Plotkin, David~B. Shmoys, and Éva Tardos.
\newblock Fast approximation algorithms for fractional packing and covering
  problems.
\newblock {\em Mathematics of Operations Research}, 20(2):257--301, 1995.

\bibitem{Pur70}
Paul Purdom.
\newblock A transitive closure algorithm.
\newblock {\em BIT Numerical Mathematics}, 10:76--94, 1970.

\bibitem{QS06}
Maurice Queyranne and Andreas~S. Schulz.
\newblock Approximation bounds for a general class of precedence constrained
  parallel machine scheduling problems.
\newblock {\em SIAM J. Comput.}, 35(5):1241--1253, May 2006.

\bibitem{SS02}
Andreas~S. Schulz and Martin Skutella.
\newblock Scheduling unrelated machines by randomized rounding.
\newblock {\em SIAM J. Discret. Math.}, 15(4):450--469, April 2002.

\bibitem{SS99}
Jay Sethuraman and Mark~S. Squillante.
\newblock Optimal scheduling of multiclass parallel machines.
\newblock In {\em Proceedings of the Tenth Annual ACM-SIAM Symposium on
  Discrete Algorithms (SODA 1999)}, pages 963--964.

\bibitem{SM90}
Farhad Shahrokhi and D.~W. Matula.
\newblock The maximum concurrent flow problem.
\newblock {\em J. ACM}, 37(2):318–334, April 1990.

\bibitem{She17}
Jonah Sherman.
\newblock Area-convexity, $\ell_\infty$ regularization, and undirected
  multicommodity flow.
\newblock In {\em Proceedings of the 49th Annual ACM SIGACT Symposium on Theory
  of Computing (STOC 2017)}, page 452–460.

\bibitem{ST93}
David~B. Shmoys and \'{E}va Tardos.
\newblock An approximation algorithm for the generalized assignment problem.
\newblock {\em Math. Program.}, 62(1–3):461–474, feb 1993.

\bibitem{Sku01}
Martin Skutella.
\newblock Convex quadratic and semidefinite programming relaxations in
  scheduling.
\newblock {\em J. ACM}, 48(2):206--242, March 2001.

\bibitem{ST83}
Daniel~D. Sleator and Robert {Endre Tarjan}.
\newblock A data structure for dynamic trees.
\newblock {\em Journal of Computer and System Sciences}, 26(3):362--391, 1983.

\bibitem{Sve10}
Ola Svensson.
\newblock Conditional hardness of precedence constrained scheduling on
  identical machines.
\newblock In {\em Proceedings of the Forty-second ACM Symposium on Theory of
  Computing (STOC 2010)}, pages 745--754.

\bibitem{Sve12}
Ola Svensson.
\newblock Santa {C}laus schedules jobs on unrelated machines.
\newblock {\em SIAM Journal on Computing}, 41(5):1318--1341, 2012.

\bibitem{You95}
Neal~E. Young.
\newblock Randomized rounding without solving the linear program.
\newblock In {\em Proceedings of the Sixth Annual ACM-SIAM Symposium on
  Discrete Algorithms (SODA 1995)}, page 170–178.

\bibitem{You14}
Neal~E. Young.
\newblock Nearly linear-work algorithms for mixed packing/covering and
  facility-location linear programs, 2014.

\end{thebibliography}
	\appendix

\section{Precedence Constrained Scheduling on Identical Machines to Minimize Weighted Completion Time}
\label{sec:prec}
In this section, we give our nearly-linear time algorithms for  $P|\tprec|\sum_j w_jC_j$, and its two special cases $1|\tprec|\sum_j w_jC_j$ and $P|\tprec, p_j = 1|\sum_j w_jC_j$. The approximation ratios for the two special cases are $2+\epsilon$ and $1+\sqrt{2} + \epsilon$ respectively, matching their correspondent current best ones achievable in polynomial time.

We describe our LP relaxation in Section~\ref{subsec:prec-LP} and the rounding algorithms in Section~\ref{subsec:prec-rounding}. The oracle for solving \eqref{LP:packing-aggregate} in the MWU framework is described in Section~\ref{subsec:oracle}, with the key component on solving the network flow problem deferred to Section~\ref{sec:appendix-networkflow}.

\subsection{Linear Programming Relaxation} 
\label{subsec:prec-LP}


We describe the LP relaxation for $P|\tprec|\sum_j w_jC_j$.  To concentrate on the main ideas, we assume $p_{\max}:= \max_{j \in J}p_j$ is bounded by $\poly(n)$, and deter the general case to Appendix~\ref{sec:prec-super-p}. We remark that a direct implementation of our algorithm would give a $\tilde O_\epsilon((n + \kappa)\log^3p_{\max})$ running time; additional ideas are needed to reduce the $\log^3p_{\max}$ term to $\log p_{\max}$.

For every $j \in J$, let $q_j$ be the maximum total length of jobs in a precedence chain ending at $j$. This can be computed in $O(n + \kappa)$ time using dynamic programming.  We define a list of completion times as follows: let $\tau_0 =0$, $\tau_d = (1+\epsilon)^{d-1}$ for every integer $d \geq 1$. Let $D$ be the smallest integer such that $\tau_D \geq p(J)$.  Then $D = O\left(\frac{\log n}{\epsilon}\right)$ since we assumed $p(J) = \poly(n)$. For every integer $d \in [0, D-1]$ we define $\eta_d := \tau_{d+1} - \tau_d$. 

Let $d^{\min}_{j} = 0$ and $d^{\max}_{j} = D$ for every $j \in J$.  Later in the super-polynomial $p_{\max}$ case, we define $d^{\min}_{j}$'s and $d^{\max}_{j}$'s differently. 
The linear program is defined by the objective \eqref{LP:prec} and constraints (\ref{LPC:prec-same-j}-\ref{LPC:prec-1}).
\begin{align}
	\min \qquad w(J) \tau_D - \sum_{j \in J} w_j \sum_{d = 1}^{D-1} \eta_d x_{jd} \label{LP:prec}
\end{align}\vspace*{-25pt}

\noindent\begin{minipage}[t]{0.55\textwidth}
	\begin{align}
		x_{jd}&\leq x_{j(d+1)} & &\forall j \in J, d \in [0, D) \label{LPC:prec-same-j}\\
		x_{jd} &\geq x_{j'd} &  &\forall j \prec j', d \in [0, D]  \label{LPC:prec-prec}\\
		\sum_{j \in J} p_j x_{jd} &\leq m\tau_d & &\forall d \in [D] \label{LPC:prec-capacity}
	\end{align}
\end{minipage}\hfill
\begin{minipage}[t]{0.45\textwidth}
	\begin{align}
		x_{jd} &= 0 & &\forall j \in J, d \in [0, d^{\min}_j]\text{ or } \tau_d < q_j\label{LPC:prec-0}\\
		x_{jd} &= 1 & &\forall j \in J, d \in [d^{\max}_j, D] \label{LPC:prec-1}
	\end{align}
\end{minipage}\bigskip

In the correspondent 0/1-integer program, $x_{jd}$ indicates whether $j$ has completion time at most $\tau_d$.   \eqref{LPC:prec-same-j} says if $j$ has completion time at most $\tau_d$, then it has completion time at most $\tau_{d+1}$. \eqref{LPC:prec-prec} requires that for two jobs $j \prec j'$,  if $j'$ has completion time at most $\tau_d$, then so does $j$.  \eqref{LPC:prec-capacity} requires the total size of jobs with completion time at most $\tau_d$ to be at most $m \tau_d$ for every $d \in [0, D]$. \eqref{LPC:prec-0} says a job can not complete before $q_j$. For the $p_{\max} = \poly(n)$ setting, the condition $d \in [0, d^{\min}_j]$ is redundant.\footnote{We may assume there are no jobs $j$ with $q_j = 0$ since they can be removed.} \eqref{LPC:prec-1} says a job must complete before or at time $\tau_D$. In the linear program, \eqref{LPC:prec-0}, \eqref{LPC:prec-1} and \eqref{LPC:prec-same-j} bound all variables in $[0, 1]$. 

It remains to discuss the objective \eqref{LP:prec}. Let $\lp$ be the value of LP\eqref{LP:prec}, and $\opt$ be the weighted completion time of the optimum schedule. We prove the following lemma:
\begin{restatable}{lemma}{lemmapreclptoopt}
	\label{lemma:prec-lp-to-opt}
	$\lp \leq (1+\epsilon)\opt$.
\end{restatable}
\begin{proof}
	Let $\tilde \bfx^* \in \{0, 1\}^{J \times [0, D]}$ be the solution correspondent to the optimum schedule: $\tilde x^*_{jd} \in \{0, 1\}$ indicates if $j$ has completion time at most $\tau_d$ in the schedule. Then, we have
	\begin{align*}
		\opt &\geq \frac{1}{1+\epsilon}\sum_{j \in J}w_j \sum_{d = 1}^{D} (\tilde x^*_{jd} - \tilde x^*_{j(d-1)}) \tau_d= \frac{1}{1+\epsilon}\sum_{j \in J} w_j \left(\sum_{d = 1}^{D-1} \tilde x^*_{jd}\big(\tau_d - \tau_{d+1}\big) + \tau_D\right) \\
		&= \frac{1}{1+\epsilon}\sum_{j \in J} w_j\left(\tau_D  - \sum_{d=1}^{D-1} \eta_d \tilde x^*_{jd}\right) \geq \frac{\lp}{1+\epsilon}.
	\end{align*}
	To see the inequality in the first line, note that a job $j \in J$ with $\tilde x^*_{j(d-1)} = 0$ and $\tilde x^*_{jd} = 1$ has completion time at least $\frac{\tau_d}{1+\epsilon}$.  The equality in the line is by rearranging of terms, and that $\tilde x^*_{j, 0} = 0$ and $\tilde x^*_{j, D} = 1$ for every $j \in J$.  The equality in the second line is by the definition of $\eta_d$'s.  The inequality in the line is by that $\tilde \bfx^*$ is a valid solution to the LP.  Therefore, we have $\lp \leq (1+\epsilon)\opt$, finishing the proof of the lemma.
\end{proof}

We use the template algorithm (Algorithm~\ref{alg:LP-solver}) to solve the LP. We remove the variables that are fixed to $0$ or $1$ and define a directed graph $G = (V, E)$ as follows: $V$ is the set of $jd$ pairs for which $x_{jd}$ is not fixed to $0$ or $1$.  We add an edge to $E$ from $jd$ to $j'd'$ if we have the constraint $x_{jd} \leq x_{j', d'}$ in \eqref{LPC:prec-same-j} or \eqref{LPC:prec-prec}.  Define $$\calQ:= \{ \bfx \in [0, 1]^V: x_v \leq x_u, \forall vu \in E\}.$$
Let $\bfP \in \R_{\geq 0}^{[0, D] \times V}$ so that \eqref{LPC:prec-capacity} can be written as $\bfP \bfx \leq \bf1$.  Notice that each $v \in V$ participates in exactly one row of $\bfP$ and thus $\bfP$ has ${\bar N} := |V|$ non-zeros. Let $a_{jd} = w_j \eta_d$ for every $jd \in V$. Then minimizing \eqref{LP:prec} is equivalent to maximizing $\bfa\bfx$. Our LP becomes $\max \bfa\bfx$ subject to $\bfx \in \calQ$ and $\bfP\bfx \leq {\bf1}$,  which is exactly \eqref{LP:packing}. Let $\bfx^*$ be the optimum solution to the LP.


To apply Algorithm~\ref{alg:LP-solver}, we need an $(\epsilon, \phi)$-oracle for \eqref{LP:packing-aggregate} with some appropriate value of $\phi$. This is summarized in the following theorem, which we prove in Section~\ref{subsec:oracle}. 
\begin{restatable}{theorem}{thmprecaggregatemain}
	\label{thm:prec-aggregate-main}
	Let $G = (V, E)$ be a directed acyclic graph and $\calQ := \{\bfy \in [0, 1]^V: y_v \leq y_u, \forall vu \in E\}$. Let $\bfb, \bfa \in \R_{\geq 0}^V$  be two row vectors. 
	Let $\bfy^*$ be the $\bfy \in \calQ$ satisfying $\bfb\bfy\leq 1$ with the maximum $\bfa\bfy$. Let $\epsilon \in (0, 1), \phi \in (0, |\bfa|_1/2)$.
	Then, in $\tilde O_\epsilon\left(|E|\cdot\log^2\frac{|\bfa|_1}{\phi}\right)$ time, we can find a $\bfy \in \calQ$ satisfying $\bfb\bfy \leq 1 + \epsilon$ and $\bfa\bfy \geq \bfa\bfy^* - \phi$.
\end{restatable}


We run Algorithm~\ref{alg:LP-solver} on our instance of \eqref{LP:packing} defined by $\calQ, \bfP$ and $\bfa$, with the $(\epsilon, \phi)$-oracle given in Theorem~\ref{thm:prec-aggregate-main} to output an $(O(\epsilon), \phi)$-approximate solution to \eqref{LP:packing}, where $\phi = \epsilon \cdot w(J) \leq \epsilon \cdot \opt$. Then the $\bfx$ returned by the template LP solver has $\bfx \in \calQ, \bfP\bfx \leq (1 + O(\epsilon))\bf1$ and $\bfa\bfx \geq \bfa\bfx^* - \phi$. Then, we have $w(J)\tau_D - \bfa\bfx \leq w(J)\tau_D - \bfa\bfx^* + \phi = \lp + \phi \leq (1+\epsilon)\opt + \epsilon\cdot\opt = (1+2\epsilon)\opt$.  

The running time of Algorithm~\ref{alg:LP-solver}, excluding Step~\ref{step:lp-solver-use-oracle}, is $O\left(\frac{{\bar m}\ln {\bar m} \cdot {\bar N}}{\epsilon^2}\right) =\tilde O_{\epsilon}(n)$, as ${\bar m} = D = O\left(\frac{\log n}{\epsilon}\right) = \tilde O_{\epsilon}(1)$ and ${\bar N} = O\left(\frac{n \log n}{\epsilon}\right) = \tilde O_\epsilon(n)$. Also $|\bfa|_1 \leq \poly(n)\cdot \phi$ as all job sizes are polynomially bounded.  Therefore, in each iteration the oracle takes $\tilde O_{\epsilon}(|E|) =  \tilde O_{\epsilon}(n+\kappa)$ time, and  there are at most $O\left(\frac{{\bar m}\log {\bar m}}{\epsilon^2}\right) = \tilde O_{\epsilon}(1)$ iterations.  So, the running time of the algorithm is $\tilde O_{\epsilon}(n+\kappa)$, assuming $p_{\max} = \poly(n)$.  \smallskip

Before proceeding to the next section, we summarize the properties of our $\bfx \in [0, 1]^{J \times [0, D]}$.  Its value to \eqref{LP:prec} is at most $(1+O(\epsilon))\opt$.  $\bfx$ satisfies all constraints in LP\eqref {LP:prec}, except\eqref{LPC:prec-capacity}, which is satisfied with a factor of $1+O(\epsilon)$ on the right side.

\subsection{Rounding Algorithms}
\label{subsec:prec-rounding}
After we obtain the solution $\bfx$,  we round it to an integral one using problem-dependent algorithms. For every $j \in J$, we define 
\begin{align*}
	C_j := \sum_{d = 1}^{D} \tau_d(x_{jd} - x_{j(d-1)}) =  \sum_{d = 1}^{D-1}(\tau_{d} - \tau_{d+1}) x_{jd} + \tau_D = \tau_D - \sum_{d = 1}^{D-1}\eta_d x_{jd}
\end{align*}
to be the fractional completion time of $j$.  Then $\bfx$ has value $\sum_{j \in J} w_j C_j$ to the LP\eqref{LP:prec}.

\begin{restatable}{claim}{claimprecCtoq}
	\label{claim:prec-C-to-q}
	For a job $j \in J$, we have $C_j \geq q_j$. 
	For two jobs $j \prec j'$, we have $C_j \leq C_{j'}$. 
\end{restatable}
\begin{proof}
	To see the first statement, notice that $x_{jd} = 0$ if $\tau_d < q_j$. Thus $C_j = \sum_{d = 1}^D \tau_d (x_{jd} - x_{j(d-1)}) \geq \sum_{d = 1}^D q_j (x_{jd} - x_{j(d-1)}) = q_j$.
	The second statement follows from that $C_j = \tau_D - \sum_{d = 1}^{D-1} \eta_d x_{jd}, C_{j'} = \tau_D - \sum_{d = 1}^{D-1} \eta_d x_{j'd}$ and that $x_{jd} \geq x_{j'd}$ for every $d \in [0, D]$.
\end{proof}

\begin{restatable}{lemma}{lemmaboundcjbeforetau}
	\label{lemma:bound-cj-before-tau}
	Let ${C^*} \geq 0$ be a time point and let $J' := \set{j \in J: C_j \leq {C^*}}$. Then, we have
	\begin{align*}
		p(J') \leq (2 + O(\epsilon))m {C^*}.
	\end{align*}
\end{restatable}
\begin{proof}
	Let $\xi$ be the $1+O(\epsilon)$ term so that $\bfx$ satisfies \eqref{LPC:prec-capacity} with the right-side replaced by $\xi m \tau_d$. Let $D'$ be the minimum number such that $\xi m \tau_{D'} \geq p(J')$. Then $1 \leq D' \leq D$.  If $D' = 1$, then we have $p(J') \leq \xi m \leq (1+O(\epsilon)) m C^*$ if $C^* \geq 1$; if $C^*< 1$ then $J' = \emptyset$ by Claim~\ref{claim:prec-C-to-q}. So, we can assume $2 \leq D' \leq D$.
	\begin{align*}
		&\qquad {C^*} p(J') \quad\\
		&\geq\quad \sum_{j \in J'}C_j p_j \quad=\quad \sum_{j \in J'} p_j\left(\tau_D - \sum_{d = 0}^{D-1}\eta_d 	x_{jd}\right) \\
		&=\quad p(J')\tau_D - \sum_{j \in J'}p_j\sum_{d=0}^{D-1}\eta_d x_{jd} \quad=\quad p(J')\tau_D -\sum_{d= 0}^{D-1}\eta_d \sum_{j \in J'} p_j  x_{jd} \\
		&\geq \quad p(J')\tau_D - \sum_{d=0}^{D-1} \eta_d \min \big\{\xi m\tau_d ,p(J')\big\} 
		\quad =\quad p(J')\tau_D - \left(\xi m\sum_{d = 0}^{D' - 1} \tau_d\eta_d + \sum_{d = D'}^{D-1} \eta_d p(J')\right)\\
		&\geq \quad p(J')\tau_D - \xi m \int_{t = 0}^{\tau_{D'}} t \sfd t - (\tau_D - \tau_{D'})  p(J') \quad=\quad p(J')\tau_{D'} -\frac{\xi m\tau_{D'}^2}{2}\\
		&\geq \quad p(J')\tau_{D'} - \frac{\tau_{D'}(1+\epsilon)p(J')}{2} \quad=\quad \frac{(1-\epsilon)p(J')\tau_{D'}}{2} \quad\geq\quad \frac{(1-\epsilon)p^2(J')/(\xi m)}{2}.
	\end{align*}
	The inequality in the third line holds as $\sum_{j \in J'} p_j x_{jd} \leq \sum_{j \in J} p_j x_{jd} \leq \xi m \tau_d$ and $\sum_{j \in J'} p_j x_{jd} \leq p(J')$. To see the inequality in the fourth line, notice that $\tau_d\eta_d \leq \int_{t = \tau_{d}}^{\tau_{d+1}} t \sfd t$ for every $d \in [0, D'-1]$.  The first inequality in the last line used that $\xi m \tau_{D'} = \xi m(1+\epsilon)\tau_{D'-1} < (1+\epsilon)p(J')$ by the choice of $D'$.  The second inequality in the line used that $\tau_{D'} \geq p(J')/(\xi m)$.

	Therefore, we have $p(J') \leq \frac{2\xi m {C^*}}{1-\epsilon} = (2+O(\epsilon))m{C^*}$, as $\xi = 1+O(\epsilon)$.
\end{proof}

The lemma immediately gives us a $\tilde O_\epsilon(n + \kappa)$-time $(2+O(\epsilon))$-approximation for $1|\tprec|\sum_j w_jC_j$, finishing the proof of Theorem~\ref{thm:main-1-Prec-wC}. We schedule the jobs on the single machine in non-decreasing order of $C_j$ values, guaranteeing that if $j \prec j'$ then $j$ is scheduled before $j'$.  Then the completion time $\tilde C_{j^*}$ of a job $j^*$ is at most $p(\{j \in J: C_j \leq C_{j^*}\}) \leq (2+O(\epsilon))C_{j^*}$.  The weighted completion time of the schedule then is at most $(2+O(\epsilon))\sum_{j \in J} w_j C_j \leq (2+O(\epsilon))\opt$ as the value of $\bfx$ to LP\eqref{LP:prec} is at most $(1+O(\epsilon))\opt$.  \medskip

When $m>1$, we use a simple job-driven list scheduling algorithm as in \cite{Li20}.  In addition to the set $J$ of jobs with job sizes and precedence constraints, we are given a vector $(F_j)_{j \in J} \in \R_{\geq 0}^J$ that respects the precedence constraints: For every $j \prec j'$ we have $F_j \leq F_{j'}$.  Notice it is possible that $F_j = F_{j'}$ for $j \prec j'$. 

In the algorithm, for every job $j$ in non-decreasing order of $F_j$ values, breaking ties so that if $j' \prec j''$ then $j'$ is handled before $j''$, we schedule $j$ as early as possible without violating the $m$-machine constraint and the precedence constraints.  The pseudo-code is given in Algorithm~\ref{alg:list-scheduling}.  In the algorithm, the congestion of a set of scheduling intervals is the maximum number of intervals in the set covering a same unit-time slot.  

\begin{algorithm}[ht]
	\caption{$\listscheduling((F_j)_{j \in J})$} \label{alg:list-scheduling}
	\textbf{Input}:  a vector $(F_j)_{j \in J} \in \R_{\geq 0}^J$ respecting the precedence constraints  \\
	\textbf{Output}: a schedule of jobs, given by starting times $(\tilde S_j)_{j \in J}$ and completion times $(\tilde C_j = \tilde S_j + p_j) _{j \in J}$
	\begin{algorithmic}[1]
		\For{every $j \in J$ in non-decreasing order of $F_j$, breaking ties first using $\prec$ and then arbitrarily}
		\State $t \gets \max_{j' \prec j}\tilde C_{j'}$, assuming the maximum of an empty set is $0$
		\State find the minimum $t' \geq t$ such that we can schedule $j$ in interval $(t', t' + p_j]$, without increasing the congestion of the scheduling intervals to $m+1$ \label{step:list-scheduling-t'}
		\State ${\tilde S}_j \gets t', \tilde C_j \gets t' + p_j$, and schedule $j$ in $({\tilde S}_j, \tilde C_j]$
		\EndFor
		\State \Return $((\tilde C_j)_{j \in J})$
	\end{algorithmic}
\end{algorithm}

To guarantee that the algorithm runs in $\tilde O_\epsilon(n + \kappa)$ time, we need to show how to find the $t'$ in Step~\ref{step:list-scheduling-t'} in amortized $O(\log n)$ time. This is done by maintaining two self-balancing binary search trees. We defer the details to Section~\ref{subsec:data-structures-for-list-scheduling}.

Throughout this section, we fix a job $j^* \in J$ and analyze the completion time $\tilde C_{j^*}$ of the job in the constructed schedule. We focus on the moment where $\tilde S_{j^*}$ and $\tilde C_{j^*}$ are decided; that is, the end of the iteration in which we handle $j^*$. We call the scheduled constructed so far the schedule of interest (jobs handled after $j^*$ are not scheduled yet).  In the schedule, a unit time slot $(t-1, t]$ is said to be busy if exactly $m$ jobs have scheduling intervals covering $(t - 1,t]$; otherwise we say $(t - 1, t]$ is idle.  Let $T_\busy$ and $T_\idle$ be the number of busy and idle unit-time slots before $\tilde C_{j^*}$, w.r.t the schedule of interest. Then $\tilde C_{j^*} = T_\busy + T_\idle$.

\begin{claim}
	\label{claim:busy}
	$T_\busy \leq \frac{1}{m}p\big(\{j \in J: F_j \leq F_{j^*} \}\big)$.
\end{claim}
\begin{proof}
	The total size of jobs in the schedule of interest is at most $p\big(\{j \in J: F_j \leq F_{j^*} \}\big)$.  So, $mT_\busy \leq p\big(\{j \in J: F_j \leq F_{j^*} \}\big)$. Dividing both sides by $m$ gives the claim.
\end{proof}

\begin{lemma}\cite{MQS98, Li20}
	\label{lemma:idle-unit}
	When all jobs have unit sizes,  we have $T_\idle \leq q_{j^*}$.
\end{lemma}


So, if we let $F_j = C_j$ for every $j$ (notice that $(C_j)_{j \in J}$ respects the precedence constraints), and apply Claim~\ref{claim:busy} and Lemma~\ref{lemma:idle-unit}, we have 
\begin{align*}
	\tilde C_{j^*} &= T_\busy + T_\idle \leq \frac{1}{m}\big(\{j \in J: C_j \leq C_{j^*} \}\big) + q_{j^*} \leq (2+O(\epsilon)) C_{j^*} + C_{j^*} = \left(3+ O(\epsilon)\right) C_{j^*}.
\end{align*}
The second inequality used Lemma~\ref{lemma:bound-cj-before-tau}. This gives us a $\tilde O_\epsilon(n + \kappa)$-time $(3+O(\epsilon))$-approximation for $P|\tprec, p_j = 1|\sum_j w_j C_j$. In Section~\ref{subsubsec:improved}, we show the approximation ratio of $1 +\sqrt{2}$ due to \cite{Li20} can be recovered using our LP relaxation; this will prove Theorem~\ref{thm:main-P-Prec-pj1-wC}.\medskip

Finally, we focus on the general problem $P|\tprec|\sum_j w_j C_j$.  As our LP is weaker, we could not recover the approximation ratios of $4$  in \cite{MQS98} or $2+2\ln2$ in \cite{Li20}. Instead, we obtain a worse ratio of $6+O(\epsilon)$.  

\begin{restatable}{lemma}{lemmaidlegeneral}
	\label{lemma:idle-general}
	Let $\theta \in (0, 1)$ be a number such that for every $j \prec j'$, we have $F_{j'} - F_j \geq \theta p_j$. Then  $T_\idle \leq \frac{F_{j^*}}{\theta} + p_{j^*}$.
\end{restatable}
\begin{proof}
	We revisit the tools built in \cite{MQS98} and \cite{Li20} that bound $T_\idle$ when job sizes are arbitrary. The following lemma was proved in the two papers. (See, e.g., Lemma 2.2 in \cite{Li20}.)
	\begin{lemma}[\cite{MQS98}, \cite{Li20}]
		\label{lemma:PwC-previous-j}  Let $j \in J$ be a job in the schedule of interest with ${\tilde S}_j > 0$. Then we can find a job $j'$ such that either
		\begin{enumerate}[topsep=3pt,itemsep=0pt, label=(\ref{lemma:PwC-previous-j}\alph*), leftmargin=*]
			\item $j' \prec j$ and $(\tilde C_{j'}, {\tilde S}_j]$ is busy, or
			\item $F_{j'} \leq F_j, {\tilde S}_{j'} < {\tilde S}_{j}$ and $({\tilde S}_{j'}, {\tilde S}_j]$ is busy.
		\end{enumerate}
	\end{lemma}
	
	We start from $j = j^*$ and repeat the following process. While $\tilde S_j > 0$, we find a job $j'$ satisfying either (\ref{lemma:PwC-previous-j}a) or (\ref{lemma:PwC-previous-j}b), and update $j \gets j'$.  Notice that $F_j$ and $\tilde S_j$ only decrease from iteration to iteration. $\tilde S_j$ decreases from the initial value of $\tilde S_{j^*}$ to the final value of $0$, and $F_j$ decreases from the initial value of $F_{j^*}$ to some non-negative number.   
	
	In each iteration, we show that the number of idle slots in $(\tilde S_{j'}, \tilde S_{j}]$ is at most $\frac{F_j - F_{j'}}{\theta}$: In case (\ref{lemma:PwC-previous-j}a), we get at most $\tilde C_{j'} - \tilde S_{j'} = p_{j'}$ units of  idle time in $(\tilde S_{j'}, \tilde S_j]$, and $p_{j'} \leq \frac{F_j - F_{j'}}{\theta}$.  In case (\ref{lemma:PwC-previous-j}b), there are no idle slots in $(\tilde S_{j'}, \tilde S_{j}]$, and $0 \leq \frac{F_j - F_{j'}}{\theta}$.  So, the total number of idle time slots before $\tilde S_{j^*}$ is at most $\frac{F_{j^*}}{\theta}$, implying that the total amount of idle time before $\tilde C_{j^*}$ is at most $\frac{F_{j^*}}{\theta} + p_{j^*}$. This finishes the proof of Lemma~\ref{lemma:idle-general}.
\end{proof}

\cite{MQS98} used $F_j = C_j - \frac{p_j}{2}$ to obtain their $4$-approximation for the problem. However we are not guaranteed that the vector $(C_j - \frac{p_j}{2})_{j \in J}$ respects the precedence constraints. 
Instead, we define $F_j = C_j + q_j - p_j$ for every $j \in J$ and call $\listscheduling((F_j)_{j \in J})$.  $q_j - p_j$ is the maximum size of jobs in a precedence chain ending at some predecessor of $j$.  
If $j \prec j'$ then $F_{j'} - F_j = (C_{j'} + q_{j'} - p_{j'}) -  (C_j + q_j - p_j) \geq q_{j'} - p_{j'} - q_j + p_j \geq p_j$. So, $(F_j)_{j \in J}$ respects the precedence constraints, and it satisfies the condition in Lemma~\ref{lemma:idle-general} with $\theta = 1$. By the lemma, we have $T_\idle \leq F_{j^*} + p_{j^*} = C_{j^*} + q_{j^*} \leq 2C_{j^*}$. 
\begin{align*}
	T_\busy  &\leq \frac{1}{m}p\big(\{j \in J: F_j \leq F_{j^*} \}\big) = \frac{1}{m}\big(\{j \in J: C_j + q_j - p_j \leq C_{j^*} + q_{j^*} - p_{j^*}\}\big)\\
	&\leq \frac{1}{m}\big(\{j \in J: C_j \leq C_{j^*} + q_{j^*} - p_{j^*}\}\big) \leq (2+O(\epsilon)) (C_{j^*} + q_{j^*} - p_{j^*}) \leq (4+O(\epsilon))C_{j^*}.
\end{align*}
The first and third inequalities used Claim~\ref{claim:busy} and Lemma~\ref{lemma:bound-cj-before-tau}. Therefore, we have $\tilde C_{j^*} \leq (6+O(\epsilon)) C_{j^*}$, which gives a $\tilde O_\epsilon(n + \kappa)$-time $(6+O(\epsilon))$-approximation for $P|\tprec|\sum_j w_j C_j$, assuming $p_{\max} = \poly(n)$.


\section{Nearly-Linear Time Rounding Algorithm for Weighted Completion Time Scheduling on Unrelated Machines}
\label{sec:R-wC-rounding}
We show in Section~\ref{subsec:constructing-z} that the vector $\bfz$ in Theorem~\ref{thm:RwC-exists-zcirc} can be constructed in nearly-linear time. This proves Theorem~\ref{thm:main-R-wC-LP}.  To prove Theorem~\ref{thm:main-R-wC}, we show that the rounding algorithm of Im and Li \cite{IL23} runs in time nearly-linear in the size of the support. This is done in Section~\ref{subsec:rectangle-rounding}.

\subsection{Explicit Construction of $\bfz$ in Nearly-Linear Time}
\label{subsec:constructing-z}

In this section, we show how to construct the $\bfz$ explicitly in nearly-linear time. We use the following idea: If we discretize job lengths, and ignore job identities, then we only need to deal with $\tilde O_\epsilon(1)$ different mini-jobs for any machine $i$.

We focus on a fixed machine $i \in M$ from now on, and show how to construct $\bfz$ for this $i$.  We create two sets of rectangles: 
\begin{itemize}
	\item For every quadruple $deuv$, we create a rectangle of height $\sum_{j \in  N(i)} y_{ijdeuv}$ and horizontal span being the blocking interval of a mini-job for any $y_{ijdeuv}$; notice that all the mini-jobs have the same blocking interval as it only depends on $deuv$. We call the rectangle a \emph{cross-rectangle} and denote it as $deuv$.  This will stand for all the cross-mini-jobs for the variables $y_{ijdeuv}, j \in N(i)$. As they are all isomorphic, we do not need to distinguish them until the end of the algorithm. Let $\calR_{\textsf{cross}}$ be the set of cross-rectangles we created. Notice that the total height of the rectangles covering any time point $t$ is at most $1+\epsilon$. 
	\item  We cut down the length of each inside-mini-job for $i$ to the nearest integer in $\{T_1, T_2, \cdots, T_D\}$. First, each length is cut by at most a multiplicative factor of $\epsilon$, which can be ignored.  Second, the number of different lengths for inside-mini-jobs become $D = O(\frac{\log n}{\epsilon}) = \tilde O_\epsilon(1)$.  For every $d \in [D]$, every possible length $p$, we create a rectangle of height $\sum_{j \in N(i):p_{ij} = p} x_{ijd}$ and width $p$. Denote the rectangle by the pair $dp$ and we call it an inside-rectangle. Let $\calR_{\textsf{inside}}$ be the set of inside rectangles. Unlike a cross-rectangle, an inside-rectangle only has a width; it does not have a horizontal span.  
\end{itemize}

Notice that for the fixed $i$, the total number of rectangles in both $\calR_{\textsf{cross}}$ and $\calR_{\textsf{inside}}$ is bounded by $\tilde O_\epsilon(1)$.  The running time for constructing the rectangles, over all $i$, is linear in the support of $(\bfx, \bfy)$.

Then we simulate Algorithm~\ref{alg:schedule-on-mini-machines}.  We first construct the vector $\sigma : \text{cross-mini-jobs} \to \text{mini-machines}$ as in Step~\ref{step:mini-schedule-sigma} in $\tilde O_\epsilon(1)$ time. In each iteration, we choose a set of cross-rectangles with disjoint horizontal span, as follows. Choose the first rectangle as the one in $\calR_{\textsf{cross}}$ with the earliest starting time. Then for each $a \geq 2$, choose the $a$-th rectangle as the one in $\calR_{\textsf{cross}}$ with the earliest starting time, whose starting time is at least the ending time of the $(a-1)$-th rectangle. The procedure terminates when the $a$-th rectangle can not be found.  Then, let $g$ be the minimum height of all the rectangles we chose.  For each such rectangle, we split off a sub-rectangle of height $g$, with the same horizontal span; we assign all the sub-rectangles to the first $\Phi g$ mini-machines.  The height of the rectangles will be decreased by $g$. If the height of a rectangle becomes $0$, then we remove it.  Due to the greedy choices, after the iteration, the total height of rectangles in $\calR_{\textsf{cross}}$ covering any time point becomes at most $1+\epsilon - g$.  Then we repeat the procedure until all rectangles are removed, and we use $(1+\epsilon)\Phi$ mini-machines.  In every iteration, at least one rectangle disappears. So the running time of the procedure is $\tilde O_\epsilon(1)$ since the number of rectangles in $\calR_{\textsf{cross}}$ is $\tilde O_{\epsilon}(1)$. Each sub-rectangle of a rectangle in $\calR_{\textsf{cross}}$ is assigned to a consecutive set of mini-machines (assuming mini-machines are indexed from 1 to $(1+\epsilon)\Phi$). 


Then we simulate Loop \ref{step:mini-schedule-loop} of Algorithm~\ref{alg:schedule-on-mini-machines}. We maintain a partition of the mini-machines $[(1+\epsilon)\Phi]$ into intervals that we call bundles, where the mini-machines in each interval have the same load. Every time we try to handle a sub-rectangle of rectangle in $\calR_{\textsf{cross}}$, we find the bundles that overlap with the mini-machines it is assigned to according to $\sigma$, and increase their loads. A bundle may be split into two if necessary.  To handle a rectangle $dp$ in $\calR_{\textsf{inside}}$,  we repeatedly find the bundle with the smallest load, and schedule a sub-rectangle of $dp$ on the bundle by increasing the load. The last bundle may need to be split into two bundles. Therefore, every time we handle a sub-rectangle of rectangle in $\calR_{\textsf{cross}}$, or  a rectangle in $\calR_{\textsf{inside}}$, the number of bundles increases by at most 2. Therefore, the whole algorithm runs in time $\tilde O_\epsilon(1)$, as total number of sub-rectangles from $\calR_{\textsf{inside}}$ and rectangles in $\calR_{\textsf{inside}}$ is $\tilde O_\epsilon(1)$.  Moreover, along the way, we can increase the length of sub-rectangles of $deuv \in \calR_{\textsf{cross}}$ to $\tau^{(e+1)}_v - \tau^{(d)}_u$, which upper bounds the length of a job $j$ with $y_{ijdeuv}$ being a variable. So, this will cover Step~\ref{step:mini-schedule-extend} of Algorithm~\ref{alg:schedule-on-mini-machines} as well.

Once we have the assignment of sub-rectangles into bundles, we can then recover the vector $\bfz$ for $i$. For every rectangle $deuv \in \calR_{\textsf{cross}}$, we consider all its sub-rectangles, and then matching them to the cross-mini-jobs for variables $\{y_{ijdeuv}: j \in N(i)\}$ in a natural way. Construct two lists arbitrarily, one containing the sub-rectangles of $deuv$, the other containing jobs $j$ with $y_{ijdeuv} > 0$. We take the first sub-rectangle and the first job $j$ in the two lists.  Let $g$ be the minimum of $y_{ijdeuv}$ and the height of the sub-rectangle. Assume the sub-rectangle is scheduled with starting time $s$. Then we increase $z_{ijs}$ by $g$, decrease both $y_{ijdeuv}$ and the height of the sub-rectangle by $g$. If the sub-rectangle becomes empty, we move to the next sub-rectangle in the list; if $y_{ijdeuv}$ becomes $0$, we move to the next job in the job list.   We can handle the sub-rectangles of inside-jobs in a similar way. Notice that  the number of $z_{ijs}$ with positive values is at most the number of variables in the LP for the machine $i$ plus the number of sub-rectangles constructed.  Over all the machines, the support size of $\bfz$ is $\tilde O_\epsilon(|E|)$.

Finally, we may need to increase the length of jobs to their original length.  We should extend a job of length $p$ into a job of length $\floor{\frac{p}{1-\epsilon}}$.  To do so, we scale the time horizon by a factor of $\frac{1}{1-\epsilon}$: a job with scheduling interval $(S, C]$ will now be scheduled in $\Big(\frac{S}{1-\epsilon}, \frac{C}{1-\epsilon}\Big]$.  To make sure the starting and ending times are integers, we can change the scheduling interval to $\Big(\floor{\frac{S}{1-\epsilon}}, \floor{\frac{C}{1-\epsilon}}\Big]$.

\subsection{Rounding Algorithm for $R||\sum_j w_jC_j$ of Im and Li}
\label{subsec:rectangle-rounding}
In this section we  sketch the rounding algorithm of Im and Li \cite{IL23} for the scheduling problem that achieves the 1.45-approximation, and argue that it has running time nearly-linear in $|E|$.  They first defined a strong negative correlation scheme and designed a randomized algorithm to achieve the desired properties. In the setting, there is a set $M$ of machines, a set $J$ of jobs, a set $U$ of groups, a function $g: U \to M$ mapping groups to machines, and a vector $\bfy \in [0, 1]^{U \times J}$ such that $y(u, J) \leq 1$ for every $u \in U$, and $y(U, j) = 1$ for every $j \in J$, where $y(u, J') = \sum_{j \in J'}y_{uj}$ for every $u \in U$ and $J' \subseteq J$, and $y(U', j) = \sum_{u \in U'} y_{uj}$ for every $U'\subseteq U$ and $j \in J$.  A group $u \in U$ \emph{belongs to} the machine $g(u)$. A machine $i \in M$ is said to \emph{dominate} a job $j \in J$ if $y(g^{-1}(i), j) > \frac12$, where $g^{-1}(i)$ is defined as $\{u \in U: g(u) = i\}$.

The output of the scheme is an assignment $\sigma: J \to U$ of jobs to groups satisfying marginal probabilities, non-positive correlation for a same machine, and strongly negative correlation for a same group. Formally, for any $u \in U$ and $j \in J$, we need $\Pr[\sigma(j) = u] = y_{uj}$. For any two distinct jobs $j, j'$ and two (possibly identical) groups $u, u' \in U$ with $g(u) = g(u')$, we need $\Pr[\sigma(j) = u, \sigma(j') = u'] \leq y_{uj}y_{u'j'}$. For any two distinct jobs $j, j' \in J$ and group $u \in U$ such that $g(u)$ does not dominate any of $j$ and $j'$, we have $\Pr\left[\sigma(j) = \sigma(j') = u\right] \leq (1-\eta)y_{uj}y_{uj'}$,  where $\eta > 0$ is an absolute constant. Im and Li \cite{IL23} gives an algorithm for the scheme with $\eta = 0.1561$.

Let $F$ be the set of group-job pairs $uj$ with $y_{uj} > 0$.  It is easy to see that the algorithm of Im and Li for the strong negative correlation scheme can run in $O(|F|)$ time. In the algorithm, every job $j \in J$ randomly chooses two candidate groups $v^1_j$ and $v^2_j$ with $v^1_jj, v^2_jj \in E$, with probabilities satisfying some conditions. The two edges are called candidate edges for $j$.  $H^{\mathrm{cand}}$ is defined as the graph between $U$ and $J$ containing all the candidate edges. Then they independently mark each candidate edge with some probability. For every $u \in U$, they pair the marked edges incident to $u$.  They define  a graph $H^{\mathrm{split}}$ by splitting each group $u$ into multiple copies, one for a pair of marked edges, or a single unpaired candidate edge.  $H^{\mathrm{split}}$ is a bipartite graph where every vertex on the left side (they are obtained from the splitting operation) has degree $1$ or $2$, and every job on the right side has degree $2$.  So the graph is the disjoint union of many cycles and paths. Then Im and Li used some simple procedure for each cycle and path, to obtain the final assignment $\sigma$. 

Then we  proceed to discuss how Im and Li used the strong negative correlation scheme to round a solution $\bfz = [0, 1]^{E \times [0, T)}$ to LP\eqref{LP:rectangle}. They view the LP solution $\bfz$ as a collection of rectangles. For every $z_{ijs} > 0$, they use the triple $R_{ijs}$ to denote the rectangle with horizontal span $(s, s+p_{ij}]$ and height $z_{ijs}$. 
%
Let $\alpha = 0.3, \beta = 12.1$ and they randomly choose $\rho \in [1, 1+\beta)$ so that $\ln\beta$ is uniformly distributed in $[0, \ln(1+\beta))$. The time horizon is partitioned into infinite number of base windows of the form $(\rho(1+\beta)^{k-1}, \rho (1+\beta)^k], k \in \Z$, with grid points of the form $\rho (1+\beta)^k$.  They also choose a threshold $\tau_{ij}$ for every machine-job pair $ij \in E$.  A rectangle $R_{ijs}$ belongs to a base window $k$ if $s \leq \rho(1+\beta)^{k-1} < s+\tau_{ij} \leq \rho(1+\beta)^k$; let $\calR_k$ be the set of all rectangles belonging to the base window $k$. Notice that some rectangle may not belong to any base window. 

To create the instance for the strong negative correlation scheme, they add a group $u_{ik}$ for every machine $i$ and base window $k$, correspondent to the rectangles on machine $i$ and belonging to base window $k$.  They also create a group $v_{ij}$ for a pair $ij \in E$, correspondent to the rectangles on machine $i$ for job $j$ that do not belong to any base window.  With the correspondence between groups and sets of rectangles, the edges $F$ between the groups $U$ and the jobs $J$, and the vector $\bfy$ can be defined naturally: $y_{uj}$ for $u \in U, j \in J$ is the total height of all rectangles in the set $u$ for job $j$; if $y_{uj} = 0$ then there is no edge $uj$.  Clearly, the size of the instance is $\tilde O_\epsilon(|E|)$. Using the algorithm for the strong negative correlation scheme, they can find an assignment of jobs to groups, and thus an assignment of jobs to machines. Then scheduling jobs on each machine using the Smith rule gives the final schedule.  It is not hard to see the algorithm runs in nearly linear time. Finally, the $1+\epsilon$ violation on LP constraints lead to a $1+\epsilon$ multiplicative factor in the approximation ratio.

\section{Handling Super-Polynomial Integers in Input}
\label{sec:bigpw}
In this section, we show how to handle the cases when the sizes and/or weights are super-polynomial in $n$, for the two problems $P|\tprec|\sum_j w_jC_j$ and $R||\sum_j w_j C_j$.

\subsection{Handling Arbitrary Processing Times and Weights for $R||\sum_j w_jC_j$}
\label{subsec:wc-big-p-w}

In this section, we consider the unrelated machine weighted completion time problem, and remove the assumption that all weights and processing times are bounded by a polynomial function of $n$. 

\paragraph{Preprocessing}
First we need a $\poly(n)$-approximation for the problem, and this can be done easily: 
\begin{lemma}
	Assigning each job $j$ to the machine $i$ with the smallest $p_{ij}$ leads to a $\frac{n+1}{2}$-approximation for the weighted completion time problem. 
\end{lemma}
\begin{proof}
	Let $\sigma \in M^J$ be the assignment that assigns each job $j$ to the machine $i$ with the smallest $p_{ij}$. Notice that $Q:=\sum_{j \in J} w_j p_{\sigma_jj}$ is a lower bound for the weighted completion time of any schedule.  
	
	Using Smith's rule, it is well known that the weighted completion time of $\sigma$ is 
	\begin{align*}
		&\quad\quad \sum_{i, \{j, j'\}: \sigma_j = \sigma_{j'} = i} \min\Big\{w_j p_{ij'}, w_{j'} p_{ij}\Big\} \quad\leq\quad \sum_{\{j, j'\}} \min\Big\{w_j p_{\sigma_{j'}j'}, w_{j'} p_{\sigma_jj}\Big\}\\
		&\leq\quad \frac12\sum_{\{j, j'\}} \Big(w_j p_{\sigma_jj}+ w_{j'} p_{\sigma_{j'}j'}\Big)\quad=\quad\frac{n+1}{2}\sum_{j \in J} w_j p_{\sigma_jj}\quad=\quad\frac{n+1}{2}Q.
	\end{align*}
	Above, $\{j, j'\}$ is over all subsets of $J$ of size $1$ (in case $j = j'$) or $2$.  The summations are well-defined since all the terms inside  are symmetric w.r.t $j$ and $j'$.  The second inequality holds as $w_j p_{\sigma_jj}+ w_{j'} p_{\sigma_{j'}j'} \geq 2\sqrt{w_j p_{\sigma_jj}w_{j'} p_{\sigma_{j'}j'}} \geq 2\min\Big\{w_j p_{\sigma_{j'}j'}, w_{j'} p_{\sigma_jj}\Big\}$. 
\end{proof}

Then, we can assume we are given an upper bound $\Phi$ on the optimum weighted completion time and our goal is to find a schedule with weighted completion time $(1.45+O(\epsilon))\Phi$.  If some $ij \in E$ has $w_jp_{ij} > \Phi$, then we can remove $ij $ from $E$ since it can not be used.  For any job $j \in J$ for which there exists a machine $i \in M$ such that $w_jp_{ij} \leq \frac{\epsilon^2 \Phi}{n^3}$, we can then remove $j$ from $J$ (but keeping $n$ unchanged), and in the end we insert $j$ to this machine $i$ using the Smith's rule.  Let $J'$ be the set of remaining jobs and $J''$ be the set of jobs removed and inserted back in the end.  Assume we have a schedule for $J'$ with total weighted completion time at most $(1.45 + O(\epsilon))\cdot\Phi = O(1)\cdot \Phi$, and it obeys the Smith's rule. We prove
\begin{claim}
	Inserting $J''$ to the schedule for $J'$ increases the weighted completion time by at most $O(\epsilon)\cdot \Phi$.
\end{claim}
\begin{proof}
	%
	For two jobs $j, j' \in J$, we use $j \sim j'$ to denote that $j$ and $j'$ are assigned to the same machine.  Let $p'_j$ be the processing time of $j$ on its assigned machine in the final schedule.  The cost incurred by inserting jobs in $J''$ is
	\begin{align*}
		&\quad\quad \sum_{\{j, j'\}: j \sim j', \{ j, j'\} \cap J'' \neq \emptyset} \min\{w_j p'_{j'}, w_{j'} p'_j\}
		\quad\leq\quad \sum_{\{ j, j'\}: \{j, j'\} \cap J'' \neq \emptyset} \sqrt{w_jp'_{j}w_{j'}p'_{j'}}\\
		&\leq\quad \sum_{j \in J'}\sqrt{w_j p'_{j}} \sum_{j' \in J''}\sqrt{w_{j'}p'_{j'}} + \left(\frac{|J''|\cdot (|J''|+1)}{2}\right)\cdot \frac{\epsilon^2\Phi}{n^3}\\
		&\leq\quad \sqrt{n}\cdot\sqrt{\sum_{j \in J'}w_j p'_{j}}\cdot n \sqrt {\frac{\epsilon^2\Phi}{n^3}} + \frac{\epsilon^2\Phi}{n}\\
		&\leq\quad \sqrt{n}\cdot O(\sqrt{\Phi}) \cdot  \frac{\epsilon\sqrt{\Phi}}{\sqrt{n}}  + \frac{\epsilon^2\Phi}{n} \quad\leq\quad O(\epsilon) \cdot \Phi.
	\end{align*}
	Both inequalities in the second line used that every $j' \in J''$ has $w_j p'_{j'} \leq \frac{\epsilon^2 \Phi}{n^3}$.  The second inequality also used that $\sqrt{a_1} + \sqrt{a_2} + \cdots + \sqrt{a_n} \leq \sqrt{n(a_1+a_2+\cdots +a_n)}$ for any $a_1, a_2, \cdots, a_n \geq 0$.  The first inequality in the third line used that the schedule for $J'$ has weighted completion time at most $O(1)\cdot \Phi$. 
\end{proof}

Therefore, after removing jobs $J''$ from $J$, we can assume for every machine $ij  \in E$, we have $w_j p_{ij} \in \left(\frac{\Phi}{B}, \Phi\right]$ for $B =\frac{n^3}{\epsilon^2} = \poly(n, \frac1\epsilon)$.


\begin{lemma}
	\label{lemma:cj-relate-to-pij}
	In any schedule that respects the Smith's rule, if $j$ is scheduled on $i$, then its completion time is in $\left[p_{ij}, n\sqrt{B}p_{ij}\right]$.
\end{lemma}
\begin{proof}
	The completion time of $j$ is at least $p_{ij}$. 	On the other hand, if $j'$ is scheduled before $j$ on the same machine $i$, then we have $\frac{p_{ij'}}{w_{j'}} \leq \frac{p_{ij}}{w_j}$.   This implies that $p_{ij'}^2 = p_{ij'}w_{j'}\cdot \frac{p_{ij'}}{w_{j'}} \leq Bp_{ij}w_j \cdot \frac{p_{ij}}{w_j} = B p_{ij}^2$. So, $p_{ij'} \leq \sqrt{B} p_{ij}$.  This implies that the completion time of $j$ is at most $n\sqrt{B}p_{ij}$.
\end{proof}


\paragraph{Modifications to LP\eqref{LP:RwC}}
With the lemma, we can then show how to modify LP\eqref{LP:RwC} to make our running time nearly-linear.  

\begin{itemize}
	\item {\bf Restricting the set of $x$-variables.} With Lemma~\ref{lemma:cj-relate-to-pij}, we introduce a variable $x_{ijd}$ only if  $\Delta_d \leq p_{ij}$ and $T_d \leq n \sqrt{B}p_{ij}$.  There are at most $O\left(\frac{\log (n\sqrt{B})}{\epsilon}\right) = \tilde O_{\epsilon}(1)$ different variables $x_{ijd}$ for a fixed $ij  \in E$. Thus, the number of $x$ variables is at most $O_{\epsilon}(|E|)$.    
	
	\item {\bf Restricting the set of $y$-variables by reducing job lengths.} We decrease each $p_{ij}$ by $\floor{\epsilon p_{ij}}$ and only allow $j$ to start at or after $\floor{\epsilon p_{ij}}$ on machine $i$. On one hand, this will make the instance easier; on the other hand, extending the length from $p_{ij} - \floor{\epsilon p_{ij}}$ back to $p_{ij}$ for all $ij \in E$ increases the weighted completion time by at most a multiplicative factor of $1+\epsilon$.  After this, we introduce a $y_{ijdeuv}$ variable only if $t^d_u \geq \floor{\epsilon p_{ij}}$ and $t^d_u + p_{ij} - \floor{\epsilon p_{ij}} \leq n\sqrt{B}p_{ij}$.	 Then the number of $y_{ijdeuv}$ variables can be bounded by $\frac{|E|\log n}{\epsilon^2}$. Moreover, one can see that each variable appears in at most $O(\frac{\log n}{\epsilon})$ constraints. Thus, the number of non-zeros in the LP can still be bounded by $\tilde O_\epsilon(|E|)$.
\end{itemize}

Finally, we need to say something about the construction of $\bfz$ since now $D = \Theta(\frac{\log (np_{\max})}{\epsilon})$ might be large, and the construction of $\bfz$ as in Appendix~\ref{subsec:constructing-z} might not be fast any more.  We break the time horizon $[1, \infty)$ into phases of the form $[n^{6k}, n^{6(k + 1)}), k \in \Z_{\geq 0}$; we call $[n^{6k}, n^{6(k + 1)})$ the $k$-th phase. Let $L = \ceil{\frac{1}{\epsilon}} + 1$.  Then, we discard 1 out of every $L$ phases randomly: choose some integer $\ell \in [0, L-1]$ uniform at random, we discard the $(aL + \ell)$-th phase for every integer $a \geq 0$.  If in the fractional solution $(\bfx, \bfy)$, any fraction of job $j$ has completion time inside a discarded phase, then we discard $j$.   A job $j$ can only complete at a time in $[p_{ij}, n\sqrt{B}p_{ij}]$ on machine $i$, and all $p_{ij}$'s for the same $j$ and different $i$'s differ by at most a factor of $B$. Therefore, a job $j$ is discarded with probability at most $O(1/L)$.  A phase (job) that is not discarded is said to be alive. 

We call a maximal consecutive interval of alive phases an epoch.  An alive job $j$ belongs to an epoch if every fractional of $j$ completes inside the epoch, according to $(\bfx, \bfy)$; every alive job must belong to some epoch.  For every epoch, we can construct a solution $\bfz$ to the rectangle LP, for the jobs belonging to the epoch, using the algorithm in \ref{subsec:constructing-z}; notice that now the lengths of jobs belonging to an epoch differ by at most a factor of $n^{O(L)} = n^{O(1/\epsilon)}$.  Then, we concatenate all the solutions $z$ for all epochs into a solution for all the alive jobs. This only loses a negligible factor as for any epoch $o$, the total length of all alive jobs belonging to all previous epochs is very short compared to the length of a job belong to $o$. Finally, we can enumerate all possible $\ell$'s, and take the average of the solution $\bfz$ constructed over all $\ell$. In the solution, every job is scheduled by a fraction of $1- O(1/L) = 1-O(\epsilon)$. We can scale the variables by a factor of $1 + O(\epsilon)$; all jobs are scheduled to a fraction of 1 and \eqref{LPC:rectangle-capacity} is violated by a factor of $1+O(\epsilon)$. 

\subsection{Handling Super-Polynomial $p_{\max}$ for $P|\tprec|\sum_j w_jC_j$}
\label{sec:prec-super-p}
The main modification in this case is that we define $d^{\min}_j$ and $d^{\max}_j$ differently.  First, we need a $\poly(n)$-approximation for the scheduling instance.  Recall that $q_j$ is the maximum total size of jobs in a precedence chain ending at $j$.  The optimum schedule has weighted completion time at least $\Phi:=\sum_{j \in J}w_j q_j$.  On the other hand, if we schedule all the jobs in non-decreasing order of $q_j$ values on one machine (even in case we have $m$ machines) so that jobs respect the precedence constraints, the completion time of a job $j$ is at most $nq_j$ and thus the weighted completion time of the schedule is at most $n \Phi$.  Therefore, we have $\Phi \leq \opt \leq  n\Phi$, where $\opt$ is the optimum weighted completion time for the given instance.

For every $j \in J$,  and let $\tilde w_j:= \max_{j' \succ^* j} w_{j'}$ be the maximum weight of a job that directly or indirectly succeeds $j$: $j' \succ^* j$ means there is a precedence chain from $j$ to $j'$; we assume $j \succ^* j$. $\tilde w_j$'s can be computed in $O(|E|)$ time using dynamic programming.  We still define $\tau_0 = 0$ and $\tau_d = (1+\epsilon)^{d-1}$ for every integer $d \geq 1$. For every $j \in J$, define
\begin{align*}
	d^{\min}_j := \max\set{d:\tau_{d} \leq \frac{\epsilon\Phi}{n\tilde w_j}} \qquad\text{and}\qquad d^{\max}_j := \min \set{d : \tau_d \geq \frac{n\Phi}{\tilde w_j}}.
\end{align*}
Let $D = \max_{j \in J} d^{\max}_j$.  We use $\tilde w_j$ instead of $w_j$ in the definitions to guarantee that $(d^{\min}_j)_{j \in J}$ and $(d^{\max}_j)_{j \in J}$ respect the precedence constraints.

We still use  LP\eqref{LP:prec}, but with the new definitions of $d^{\min}_j$ and $d^{\max}_j$ values.  In the linear program that we actually solve, we only have a variable $x_{jd}$ for every $j \in J$ and integer $d \in (d^{\min}_j, d^{\max}_j)$, since the other variables are fixed to $0$ or $1$.  In the analysis, it is convenient for us to keep a variable $x_{jd}$ for every $j \in J$ and $d \in [0, D]$. 

We need to argue the validity of LP\eqref{LP:prec} again since now we forced $x_{jd} = 0$ for $d \leq d^{\min}_j$. In the correspondent 0/1-integer program with the requirement, $x_{jd}$ is intended to indicate whether $d > d^{\min}_j$ and $j$ has completion time at most $\tau_d$.  \eqref{LPC:prec-same-j} says if $j$ has completion time at most $\tau_d$ and $d > d^{\min}_j$, then it has completion time at most $\tau_{d+1}$ and $d + 1 > d^{\min}_j$. \eqref{LPC:prec-prec} requires that for two jobs $j \prec j'$ and $d > d^{\min}_{j'}$,  if $j'$ has completion time at most $\tau_d$, then so does $j$ and $d > d^{\min}_{j}$.  This is valid since $\tilde w_{j} \geq \tilde w_{j'}$, which imply $d^{\min}_j \leq d^{\min}_{j'}$. \eqref{LPC:prec-capacity} is valid since the total size of jobs with completion time at most $\tau_d$ is at most $m \tau_d$ for every $d \in [0, D]$ in any valid solution.  \eqref{LPC:prec-0} is from the intended meaning of $x_{jd}$'s and that a job $j$ can not complete before time $q_j$. \eqref{LPC:prec-1} is valid by the definition of $d^{\max}_j$: If a job $j$ has completion time more than $\tau_{d^{\max}_j}$, then it incurs a weighted completion time of more than $n\Phi \geq \opt$.

Then we show that forcing $x_{jd} = 0$ if $d \leq d^{\min}_j$ only incurs a multiplicative factor of $1+\epsilon$. Recall that $\lp$ and $\opt$ are respectively the values of LP\eqref{LP:prec} and the scheduling instance.

\begin{lemma}
	$\lp \leq (1+\epsilon)^2\opt$.
\end{lemma}

\begin{proof}
	Let $\tilde \bfx^* \in \{0, 1\}^{J \times [0, D]}$ be the solution correspondent to the optimum schedule: $\tilde x^*_{jd} \in \{0, 1\}$ indicates if $d >d^{\min}_j$ and $j$ has completion time at most $\tau_d$ in the schedule. Then, we have
	\begin{align*}
		\opt \quad&\geq\quad \sum_{j \in J}w_j \left(\sum_{d = 1}^{D} (\tilde x^*_{jd} - \tilde x^*_{j(d-1)}) \frac{\tau_{d}}{1+\epsilon} - \tau_{d^{\min}_j}\right) \\
		&\geq\quad \frac{1}{1+\epsilon}\sum_{j \in J}w_j \sum_{d = 1}^{D} (\tilde x^*_{jd} - \tilde x^*_{j(d-1)}) \tau_d - \epsilon\cdot\opt\\
		&=\quad \frac{1}{1+\epsilon}\sum_{j \in J} w_j \left(\sum_{d = 1}^{D-1} \tilde x^*_{jd}\big(\tau_d - \tau_{d+1}\big) + \tau_D\right) - \epsilon\cdot\opt \\
		&=\quad \frac{1}{1+\epsilon}\sum_{j \in J} w_j\left(\tau_D  - \sum_{d=1}^{D-1} \eta_d \tilde x^*_{jd}\right) -\epsilon\cdot \opt 
		\quad\geq\quad \frac{\lp}{1+\epsilon}  -\epsilon\cdot \opt.
	\end{align*}
	To see the first inequality in the first line, focus on a job $j \in J$ and the $d$ such that $\tilde x^*_{jd} = 1$ and $\tilde x^*_{j(d-1)} = 0$. If $d-1 > d^{\min}_j$, then the completion time of $j$ is in $(\tau_{d-1}, \tau_d]$ and thus is at least $\frac{\tau_d}{1+\epsilon}$. Otherwise $d-1 = d^{\min}_j$ and the completion time of $j$ is at least 1. In either case, the term inside the parentheses lower bounds the completion time.   The second inequality in the line holds since $\sum_{j \in J}w_j \tau_{d^{\min}_j} \leq \epsilon \cdot \opt$, as $w_j \tau_{d^{\min}_j} \leq \tilde w_j \tau_{d^{\min}_j} \leq \frac{\epsilon\Phi}{n} \leq \frac{\epsilon\opt}{n}$  for every $j \in J$.   The other arguments are the same as those in the proof of Lemma~\ref{lemma:prec-lp-to-opt}. In the end, we have $\lp \leq (1+\epsilon)^2\opt$, finishing the proof of the lemma.
\end{proof}


Again, we define the directed graph $G = (V, E)$ in Section~\ref{subsec:prec-LP}: $jd \in V$ if and only if $x_{jd}$ is not fixed to $0$ or $1$ in LP\eqref{LP:prec}, and there is an edge from $jd$ to $j'd'$ if we have a constraint $x_{jd} \leq x_{j', d'}$ in \eqref{LPC:prec-same-j} or \eqref{LPC:prec-prec}. For every $j \in J$, we have $d^{\max}_j - d^{\min}_j = O\left(\frac{\log n}{\epsilon}\right) = \tilde O_\epsilon(1)$.   Therefore we have $|V| \leq \tilde O_\epsilon(n)$.  The numbers of constraints in \eqref{LPC:prec-same-j} and \eqref{LPC:prec-prec} are respectively $\tilde O_\epsilon(n)$ and $\tilde O_\epsilon(\kappa)$. So $|E|\leq \tilde O_\epsilon(n + \kappa)$.  Each variable appears in exactly one constraint in \eqref{LPC:prec-capacity}, the matrix $\bfP$ defining \eqref{LPC:prec-capacity} has the number of non-zeros being ${\bar N} = \tilde O_\epsilon(n)$.  Let $a_{jd} = w_j \eta_d$ for every variable $jd \in V$. Then the LP is equivalent to $\max \bfa\bfx$ subject to $\bfx \in \calQ:=\{\bfx \in [0, 1]^V: x_v \leq x_u \forall vu \in E\}, \bfP\bfx \leq \bf1$.

We set $\phi = \epsilon\Phi \leq \epsilon \cdot \opt$. Then for every $jd\in V$, we have $a_{jd} = w_j \eta_d \leq \tilde w_j \tau_{d^{\max}_j-1} \leq n\Phi \leq \frac{n}{\epsilon}\cdot \phi$, as  $d < d^{\max}_j$. Therefore, we have $|\bfa|_1 \leq \poly(n)\cdot \phi$. Then each time the oracle given in Theorem~\ref{thm:prec-aggregate-main} still takes time $\tilde O_\epsilon(n + \kappa)$. However,  Loop~\ref{step:lp-solver-main-loop} in the template algorithm need to run for $O\left(\frac{\bar m\log \bar m}{\epsilon^2}\right) = \tilde O_{\epsilon}(D) = \tilde O_{\epsilon}(\log p_{\max})$ iterations, where $\bar m$ is the number of rows of $\bfP$, i.e., the number of constraints in \eqref{LPC:prec-capacity}. Overall, the running time of the algorithm is $\tilde O_\epsilon((n + \kappa)\log p_{\max})$.



\section{Other Omitted Analysis}
\label{sec:other}

\subsection{Reducing $R||C_{\max}$ to Promise Version}\label{subsec:RCmax-to-promise}
We first show how to reduce the general problem to the promise version.   Let $\calA$ be the algorithm for the promise version of the problem. If we are given a $P \geq \opt$ to $\calA$, the it will successfully output a schedule of makespan at most $(2+O(\epsilon))P$. However, when $P < \opt$, the algorithm may or may not succeed. 

Assigning each job $j$ to the machine $i$ with the smallest $p_{ij}$ value gives us an $m$-approximation. Then we can create a geometric sequence of $\floor{\log_{1+\epsilon}m} + 1$ numbers $P$, such that one of them has $\opt \leq  P < (1+\epsilon)\opt$.   Via binary search among these numbers\footnote{Simply enumerating all values of $P$ is sufficient, but binary search gives a better dependence.}, we can run $\calA$ for $O\big(\log(\floor{\log_{1+\epsilon}n}+1)\big) = O(\log \log n)$ times, to find a $P < (1+\epsilon)\opt$ for which $\calA$ succeeds. So, the schedule for this $P$ has makespan at most $(2+O(\epsilon))P \leq (2+O(\epsilon))\opt$.  \medskip

\subsection{$(1+\sqrt{2}+\epsilon)$-Approximation Algorithm for $P|\tprec, p_j=1|\sum_j w_jC_j$}
\label{subsubsec:improved}
In this section, we show that the $(1+\sqrt{2})$-approximation of \cite{Li20} can be obtained using our weaker LP relaxation. 
For every $j \in J$ and $\theta \in (0, 1]$, define $D^\theta_j = \tau_d$ where $d$ is the minimum integer in $[0, D]$ such that $x_{jd} \geq \theta$.   This is the time that $\theta$ fraction of job $j$ is completed.  Notice that $C_j = \sum_{d = 1}^D \tau_d(x_{jd} - x_{j(d-1)}) = \int_{\theta = 0}^{1} D^\theta_j\sfd \theta$. Our algorithm for the problem chooses $\theta$ uniformly at random from $(0, 1]$, and then call $\listscheduling((D^\theta_j)_{j \in J})$ and output the returned schedule. Let $\tilde C_j$ be the completion time of the job $j$ in the constructed schedule.  Focus on a job $j^* \in J$ from now on and we bound $\E_\theta\left[\frac{\tilde C_{j^*}}{C_{j^*}}\right]$.


We shall use $g(\theta) = D^\theta_{j^*}$ for every $\theta \in (0, 1]$ and so $C_{j^*} = \int_{\theta = 0}^{1}g(\theta) \sfd \theta$. Let $g(0) = \lim_{\theta \to 0^+}g(\theta)$; that is, $g(0)$ is $\tau_d$ for the smallest $d$ such that $x_{j^*, d} > 0$.   By \eqref{LPC:prec-0}, we have $g(0) \geq q_{j^*}$. For every $j \in J$ and $\theta \in [0, 1]$, define $\displaystyle h_j(\theta):= x_{jd}$ for the $d$ satisfying $\tau_d = g(\theta)$. This is the fraction of  job $j$ that is completed when $\theta$ fraction of $j^*$ is completed.  Thus, we have $\sum_{j \in J} h_j(\theta) \leq (1+O(\epsilon)) m g(\theta)$ for every $\theta \in [0, 1]$. Notice that $D^\theta_j \leq D^\theta_{j^*}$ if and only if $h_j(\theta) \geq \theta$. So, by Claim~\ref{claim:busy}, Lemma~\ref{lemma:idle-unit} and that $q_{j^*} \leq g(0)$, we have $\tilde C_{j^*} \leq g(0) + \frac1m\sum_{j \in J}\mathbf{1}_{h_j(\theta) \geq \theta}$. Thus, we can bound $\frac{\E\left[\tilde C_{j^*}\right]}{C_{j^*}}$ by the supreme of 
\begin{align}
	\frac{g(0) +  \frac1m\sum_{j \in J}\int_{\theta=0}^1\mathbf{1}_{h_j(\theta)\geq \theta} \sfd \theta}{\int_{\theta=0}^1g(\theta)\sfd\theta} \label{inequ:PwC-uniform-supreme}
\end{align} 
subject to 
\begin{enumerate}[leftmargin=*, topsep=3pt,itemsep=3pt,label=(\ref{inequ:PwC-uniform-supreme}.\arabic*)]
	\item $g:[0, 1] \to [1, \infty)$ is piecewise linear, left-continuous and non-decreasing, \label{property:g-increasing}
	\item $\forall j \in J$, $h_j:[0,1]\to[0,1]$ is piecewise linear, left-continuous and non-decreasing, \label{property:h-increasing}
	\item 	$\displaystyle	\sum_{j \in J} h_j(\theta) \leq  (1+O(\epsilon)) mg(\theta), \quad \forall \theta \in [0, 1]$.\label{property:h-less-than-g}
\end{enumerate}

\begin{lemma}[\cite{Li20}]
	The supreme of \eqref{inequ:PwC-uniform-supreme} satisfying the three properties is at most $(1 + \sqrt{2})(1+O(\epsilon))$. 
\end{lemma}
Indeed, \cite{Li20} proved that the supreme is exactly $1+\sqrt{2}$ if the $(1+O(\epsilon))$-term in \ref{property:h-less-than-g} is $1$.  So, we can define $g' = (1+O(\epsilon))g$ and $g'$ satisfies the property holds with $(1+O(\epsilon))$ replaced to $1$; moreover \ref{property:g-increasing} and \ref{property:h-increasing} remain satisfied for $g'$. Then the supreme of $\frac{g'(0) +  \frac1m\sum_{j \in J}\int_{\theta=0}^1\mathbf{1}_{h_j(\theta)\geq \theta} \sfd \theta}{\int_{\theta=0}^1g'(\theta)\sfd\theta} = \frac{g(0) +  \frac1{(1+O(\epsilon))m}\sum_{j \in J}\int_{\theta=0}^1\mathbf{1}_{h_j(\theta)\geq \theta} \sfd \theta}{\int_{\theta=0}^1g(\theta)\sfd\theta}$ is $1+\sqrt{2}$.  So, the supreme of \eqref{inequ:PwC-uniform-supreme} is at most $(1+\sqrt{2})(1+O(\epsilon)) = 1+\sqrt{2}+O(\epsilon)$.


\subsection{Implementation of $\listscheduling$ for Identical Machine Precedence Constrained Scheduling in $O((n + \kappa)\log n)$ Time} 
\label{subsec:data-structures-for-list-scheduling}

\newcommand{\II}{{\mathsf{Idle}\mhyphen\mathsf{Intervals}}}
\newcommand{\find}{{\mathsf{find}}}
\newcommand{\remove}{{\mathsf{remove}}}
\newcommand{\sfinsert}{{\mathsf{insert}}}
\newcommand{\increase}{{\mathsf{increase}}}
\newcommand{\CC}{{\mathsf{Critical}\mhyphen\mathsf{Counters}}}

We define two data structures. The first data structure, which we call $\II$, maintains the set of idle unit-time slots. (Recall that a slot $(t-1, t]$ is idle if the number of jobs processed during the slot is at most $m-1$, and busy otherwise.)  Initially, all unit-time slots $(t-1, t], t \geq 1$ are idle. The data structure supports the following two operations: 	
\begin{itemize}
	\item $\find(t, p)$: given two integer $t \geq 0$ and $p > 0$, return the smallest $t' \geq t$ such that $(t', t'+p]$ is idle, i.e., all unit-time slots in $(t', t'+p]$ are idle.
	\item $\remove(t, t')$: given an idle interval $(t, t']$, mark all unit-time slots in $(t, t']$ as busy.
\end{itemize}

To implement the data structure, we maintain the set $\calI$ of inclusion-wise maximal idle intervals.  We store $\calI$ in a self-balancing binary-search tree  (BST) with the left-to-right order.  For each node in the BST, we maintain the maximum length of intervals in $\calI$ stored in the sub-tree rooted at the node. 

With this data structure, both operations can be done in $O(\log s)$ time, where $s$ is the maximum possible size of $\calI$.   For $\find(t, p)$, we first try to find the interval $I \in \calI$ containing $t$. If it exists and containing $(t, t+p]$, we return $t$. If the algorithm does not return, we find the left-most interval $I \in \calI$ to the right of $t$ with length at least $D$.  Both steps can be implemented in $O(\log s)$-time. For $\remove(t, t')$, we need to find the interval in $I = (\tilde t, \tilde t'] \in \calI$ containing $(t, t']$, which is guaranteed to exist, remove $I$ from $\calI$, adding $(\tilde t, t]$ and/or $(t', \tilde t']$ to $\calI$, if they are not empty. So $\remove(t, t')$ can be done in $O(\log s)$ time.\medskip

The second data structure, which we call $\CC$, maintains a set $\calT$ of critical time points, and a counter $m_t$ for every $t \in \calT$, which is an integer in $[0, m]$. In the $\listscheduling$ algorithm, a time point $t$ is critical if $t = 0$ or some job completes at $t$.  Notice that the starting time $\tilde S_j$ of a job  $j$ is either $0$ or the completion time $\tilde C_{j'}$ of some other job $j'$.  $m_t$ for a $t \in \calT$ is supposed to be the number of jobs $j$ such that $(t, t+1] \subseteq (\tilde S_j, \tilde C_j]$. The data structure supports the following operations:
\begin{itemize}[itemsep=0pt]
	\item $\sfinsert(t)$: if $t \notin \calT$, then we add $t$ to $\calT$, with $m_t = m_{t'}$ where $t'$ is the last critical time point before $t'$.
	\item $\increase(\tau, \tau')$: increase the $m_t$ values of all time points $t \in \calT \cap [\tau, \tau')$ by 1, and return the list of those points $t$ whose new $m_t$ values become $m$, as well as their respective next critical time point.  It is guaranteed that before the updates, every $t \in \calT \cap [\tau, \tau')$ has $m_t < m$.
\end{itemize}

We again use a self-balancing BST to store $\calT$ and their $m_t$ values, with the natural integer order for the time points.  We setup some notations first before describing the values maintained at the nodes of the tree. For each node $v$ in the BST, let $A_v$ be the set of ancestor nodes of $v$ in the BST, including $v$ itself, let $\Lambda_v$ be the descendant nodes of $v$, including $v$.  For every $u \in \Lambda_v$, let $Pvu$ be the set of nodes in the path from $v$ to $u$ in the BST.  For every node $v$ in the BST, let $t_v$ be the time point stored at $v$, and for every $t \in \calT$, let $v_t$ be the node in the BST storing $t$. So $v_t = v$ if and only if $t_v = t$. We maintain three values for each node $v$: $m_\textrm{inc}(v), m_\textrm{self}(v)$ and $m_{\max}(v)$.  We guarantee the following two properties:
\begin{itemize}[itemsep=0pt]
	\item For every $t \in \calT$, we have $m_t = \sum_{v \in A_{v_t}} m_\textsf{inc}(v) + m_\textsf{self}(v_t)$.
	\item For every $v$ in the BST, $m_{\max}(v) = \max_{u \in \Lambda_v}\left(\sum_{u' \in Pvu}m_\textrm{inc}(u') + m_\textrm{self}(u)\right)$, which is equal to $\max_{u \in \Lambda_v} m_{t_u} - \sum_{v' \in A_v \setminus \{v\}} m_\textsf{inc}(v')$.  
\end{itemize}
Given $m_\textsf{inc}(v)$'s and $m_\textsf{self}(v)$'s for all $v$, the $m_{\max}(v)$'s can be defined as follows: for every node $v$ in the BST with left child $v'$ and right child $v''$, let $m_{\max}(v) = \max\{m_{\max}(v') , m_{\max}(v''), m_\textsf{self}(v)\} + m_\textsf{inc}(v)$, where we assume if $v'$ or $v''$ does not exist, then its $m_{\max}$ value is $-\infty$. This guarantees that if we rotate the tree, $m_{\max}(v)$'s can be updated efficiently.

Again let $s$ be the maximum possible size of $\calT$. Then, $\sfinsert(t)$ takes $O(\log s)$ time.
It takes $O(\log s)$-time for $\increase(\tau, \tau')$ to update the data structure:  Let $\calT_v = \{t_u: u \in \Lambda_v\}$. We need to update the information for a node $v$ only if $\calT_v \cap[\tau, \tau') \neq \emptyset$ and $\calT_v \setminus [\tau, \tau') \neq \emptyset$, or $v$ is the topmost node with $\calT_v \subseteq [\tau, \tau')$. 
In the former case, we may need to update $m_\textsf{self}(v)$ and $m_\textsf{max}(v)$. In the latter case, we shall increase $m_\textsf{inc}(v)$ and $m_{\max}(v)$ by 1. There are only $O(\log s)$ nodes $v$ whose information will be updated.  It takes $O((r+1)\log s)$-time for $\increase(\tau, \tau')$ to return the list of critical time points $t \in [\tau, \tau')$ with $m_t = m$, where $r$ is the number of such points. 
\medskip

Now we show how to implement the $\listscheduling$ algorithm using the two data structures.  Consider the iteration for scheduling $j$. To find the minimum $t' \geq t$ such that $(t', t' + p_j]$ is idle, we call  $t' \gets \II.\find(t, p_j)$, and $\CC.\sfinsert(t'+p_j)$. To schedule $j$ in $(t', t'+p_j]$,  we call $\CC.\increase(t', t'+p_j)$.  For every returned time point $\tau$ and its next critical point $\tau'$, we call $\II.\remove(\tau, \tau')$.   $\II.\find$ and $\CC.\sfinsert$ are called once for every $j$, and each of the two operations has running time $O(\log n)$, as there are at most $n + 1$ critical time points.  Every critical time point is returned at most once by $\CC.\increase$,  since once a $t \in \calT$ has $m_t = m$, $m_t$ will never be increased in the future. This also implies that $\II.\remove$ will only be called $O(n)$ times. Therefore, the running time of the $\listscheduling$ algorithm is $O(\kappa + n \log n)$ using the two data structures. 


\section{Approximate Oracle for \eqref{LP:packing-aggregate}: Proof of Theorem~\ref{thm:prec-aggregate-main}}
\label{subsec:oracle}

In this section, we prove Theorem~\ref{thm:prec-aggregate-main}. Throughout this section, we fix $G = (V, E), \calQ, \bfb, \bfa, \bfy^*, \epsilon$ and $\phi$ as in Theorem~\ref{thm:prec-aggregate-main}, among which $\bfy^*$ is not given to our algorithm. For any directed graph $H = (V_H, E_H)$, and two subsets $U, U' \subseteq V_H$, $U\leadsto_H U'$ holds if there is a path from some vertex in $U$ to some vertex in $U'$ in $H$. If there is no such a path, then $U \not\leadsto_H U'$ holds.  If $U$ or $U'$ is a singleton set, we can replace it with the vertex it contains.  When $H = G$, the subscripts $H$ in the notations defined above can be omitted. \smallskip

Let $S = \{s \in V: a_s > 0\}$ and $T = \{t \in V: b_t > 0\}$.  We prove that the following properties can be assumed w.l.o.g.
\begin{restatable}{lemma}{lemmanetworktransformation}
	\label{lemma:network-transformation}
	To prove Theorem~\ref{thm:prec-aggregate-main}, we can w.l.o.g assume 
	\begin{itemize}[itemsep=0pt]
		\item $S \cap T = \emptyset$ and there are no edges from $S$ to $T$,
		\item $\delta^-(s) = \emptyset$ for every $s \in S$, 
		\item $\delta^+(t) = \emptyset$ for every $t \in T$, and
		\item for every $v \in V$, we have $S \leadsto v$ and $v \leadsto T$.
	\end{itemize}
	Finally, we can w.l.o.g replace the constraint $\bfy \in [0, 1]^V$ by $y_s \leq 1$ for every $s \in S$ and $y_t \geq 0$ for every $t \in T$.
\end{restatable}
\begin{proof}
	W.l.o.g, we assume every vertex $v \in V$ has $S \leadsto v$.  It is the best to set $y_v = 0$ for the vertices $v$ with $S \not\leadsto v$ and remove them.  Similarly, we assume every vertex $v \in V$ has $v \leadsto T$: It is the best to set $y_v = 1$ for  vertices $v$ with $v \not\leadsto T$ and remove them.
	
	Now for every $s \in S$, we can add a new vertex $s'$ and a new edge $(s', s)$ to $G$. We set $a_{s'} = a_s, b_{s'} = 0$, change $a_s$ to $0$, and update $S$ to $S \cup \{s'\} \setminus \{s\}$.  This does not change the instance since we have $y_{s'} \leq y_s$ and it is the best to set $y_{s'} = y_s$.   Similarly, for every $t \in T$, we add a new vertex $t'$ and a new edge $(t, t')$ to $G$. We set $a_{t'} = 0, b_{t'} = b_t$, change $b_t$ to $0$, and update $T$ to $T \cup \{t'\} \setminus \{t\}$. This does not change the instance since we have $y_{t'} \geq y_t$ and it is the best to $y_{t'} = y_t$. After this modification, the four properties in the list of the lemma are satisfied.
	
	We show w.l.o.g the constraint $\bfy \in [0, 1]^V$ can be replaced by $y_s \leq 1, \forall s \in S$ and $y_t \geq 0, \forall t \in T$. Fixing $(y_s)_{s \in S}$, we can assume $y_v = \max_{s\in S:s\leadsto v} y_s, \forall v \in V\setminus S$, as this will minimize $\bfb\bfy$. As we have $S \leadsto v$, the condition $y_v \leq 1, \forall v \in V$ is implied by that $y_v \leq 1, \forall v \in S$.  Similarly, fixing $(y_t)_{t \in T}$, we can assume $y_v = \min_{t \in T: v \leadsto t} y_t, \forall v \in V \setminus T$, as it maximizes $\bfa\bfy$.  So, the condition $y_v \geq 0, \forall v \in V$ is also implied.
\end{proof}

With Lemma~\ref{lemma:network-transformation}, the LP in Theorem~\ref{thm:prec-aggregate-main} becomes LP\eqref{LP:prec-aggregate}, which has constraints (\ref{LPC:prec-aggregate-b}-\ref{LPC:prec-aggregate-0}).

\noindent
\begin{minipage}[t]{0.4\textwidth}
	\begin{framed}
		\begin{equation}
			\max \qquad \sum_{s \in S}a_s y_s \label{LP:prec-aggregate}
		\end{equation}\vspace*{-15pt}
		\begin{align}
			\sum_{t \in T}b_t y_t &\leq 1 \label{LPC:prec-aggregate-b}\\
			y_v &\leq y_u & &\forall vu \in E \label{LPC:prec-aggregate-prec}\\
			y_s  &\leq 1 & &\forall s \in S \label{LPC:prec-aggregate-1}\\
			y_t  &\geq 0 & &\forall t \in T \label{LPC:prec-aggregate-0}
		\end{align}
	\end{framed}
\end{minipage}\hfill
\begin{minipage}[t]{0.6\textwidth}
	\begin{framed}
		\begin{equation}
			\min \qquad \gamma + \sum_{s \in S} r_s \label{LP:prec-aggregate-dual}
		\end{equation}\vspace*{-20pt}
		\begin{align}
			f\left(\delta^+(s)\right) + r_s &= a_s  & &\forall s \in S \label{LPC:prec-aggregate-dual-S}\\
			\gamma b_t - f\big(\delta^-(t)\big) &\geq 0 & &\forall t \in T\label{LPC:prec-aggregate-dual-T}\\
			f\left(\delta^+(v)\right) - f\big(\delta^-(v)\big) & = 0 & &\forall v \notin (S \cup T)\label{LPC:prec-aggregate-dual-other}\\
			\gamma &\geq 0 \\
			f_{vu} &\geq 0 & &\forall vu \in E\\
			r_s &\geq 0 & &\forall s \in S
		\end{align}
	\end{framed}
\end{minipage}\bigskip

The dual of LP\eqref{LP:prec-aggregate} is LP\eqref{LP:prec-aggregate-dual}, where variables $\gamma$, $(f_{vu})_{vu \in E}$, and $(r_s)_{s \in S}$ correspond to constraints \eqref{LPC:prec-aggregate-b}, \eqref{LPC:prec-aggregate-prec} and \eqref{LPC:prec-aggregate-1} respectively. Constraints~\eqref{LPC:prec-aggregate-dual-S}, \eqref{LPC:prec-aggregate-dual-T} and \eqref{LPC:prec-aggregate-dual-other} correspond to variables $(y_s)_{s \in S}, (y_t)_{t \in T}$ and $(y_v)_{v \in V \setminus (S\cup T)}$ respectively.  We only require $y_t$ for $t \in T$ to be non-negative. So \eqref{LPC:prec-aggregate-dual-S} and \eqref{LPC:prec-aggregate-dual-other} are equalities.

If we fix $\gamma$, the dual LP is equivalent to a network flow problem NFP$_\gamma$, defined as follows. 
\begin{definition}
	For any $\gamma \geq 0$, we use NFP$_\gamma$ to denote the following single-commodity flow problem. We are given the network $G = (V, E)$, with sources $S$ and sinks $T$. Each source $s \in S$ has a supply of $a_s > 0$, each sink $t \in T$ has a demand $\gamma b_t$, and the capacities of all edges in $G$ are infinite.  
	
	So, a valid flow for NFP$_\gamma$ is a vector $\bff \in \R_{\geq 0}^E$ satisfying $f\left(\delta^+(s)\right) \leq a_s$ for every $s \in S$,  $f(\delta^-(t)) \leq \gamma b_t$ for every $t \in T$, and $f(\delta^+(v)) = f(\delta^-(v))$ for every $v \in V \setminus (S \cup T)$.   Let $\calF_\gamma \subseteq \R_{\geq 0}^E$ denote the set of valid flows $\bff$ for NFP$_\gamma$.   The value of a flow $\bff \in \calF_\gamma$, denoted as $\val(\bff)$, is defined as $f\big(\delta^+(S)\big) = f\big(\delta^-(T)\big)$. 
\end{definition}

The value of the dual LP\eqref{LP:prec-aggregate-dual} is the minimum, over all $\gamma \geq 0$, of $\gamma + |\bfa|_1 - \opt_\gamma$, where $\opt_\gamma$ is the value of the maximum flow for NFP$_\gamma$.  To understand better why $\gamma + |\bfa|_1 - \opt_\gamma$ is an upper bound on the value of \eqref{LP:prec-aggregate}, assume in the optimum flow for NFP$_\gamma$, each $s \in S$ sends $a'_s$ units flow. Then, the flow, \eqref{LPC:prec-aggregate-prec} and \eqref{LPC:prec-aggregate-0} imply $\bfa'\bfy  \leq \gamma\bfb\bfy$, which is at most $\gamma$ by \eqref{LPC:prec-aggregate-b}.  Then $(\bfa - \bfa')\bfy \leq |\bfa - \bfa'|_1 = |\bfa|_1 - \opt_\gamma$ by \eqref{LPC:prec-aggregate-1}.  If we are not concerned with the running time, then NFP$_\gamma$ is equivalent to a fractional maximum bipartite matching problem, on the bipartite graph between $S$ and $T$ where $st$ exists if and only if $s \leadsto t$. 

We state the main theorem that solves NFP$_{\gamma}$ approximately.  As we can only lose an additive factor in the objective value in Theorem~\ref{thm:prec-aggregate-main}, the values in $\bfa$ need to be respected. Values in $\bfb$ can be approximated within a factor of $1+O(\epsilon)$.  
%
In the statement, for a subset $S' \subseteq S$ of sources, $T(S') := \{t \in T : S' \leadsto t\}$ denotes the set of sinks that can be reached from $S'$. 
\begin{restatable}{theorem}{thmmfmc}
	\label{thm:mfmc}
	Let $\gamma \geq 0, \epsilon > 0$. There is an $O\left(\frac1\epsilon\cdot|E|\cdot \log |V|\cdot \log \frac{|\bfa|_1}{\phi}\right)$-time algorithm that outputs a flow $\bff \in \calF_\gamma$ and a set $S' \subseteq S$ such that $a(S \setminus S') + \frac{\gamma b(T(S'))}{1+\epsilon} \leq \val(\bff) + \frac{\phi}{3}$.
\end{restatable}
The value of the maximum flow for NFP$_\gamma$ is $\min_{S'\subseteq S}\left(a(S \setminus S') + \gamma b(T(S'))\right)$, by the maximum flow minimum cut theorem. The theorem finds a flow whose value is at least that of the maximum flow for the instance with $\bfb$ replaced by $\frac{\bfb}{1+\epsilon}$. 
We leave the proof of Theorem~\ref{thm:mfmc}, 
which involves many technical definitions, to Section~\ref{sec:appendix-networkflow}. 

\begin{proof}[Proof of Theorem~\ref{thm:prec-aggregate-main} using Theorem~\ref{thm:mfmc}]
	To gain insights into the proof, we first ignore the running time requirement. Let $\Psi$ be the value of LP\eqref{LP:prec-aggregate}. Then by LP duality, we have $\Psi = \min_{\gamma \geq 0}(\gamma + |\bfa|_1 - \opt_\gamma)$, where $\opt_\gamma$ is the value of the optimum flow for NFP$_\gamma$.  By MFMC theorem, for every $\gamma \geq 0$, there is a set $S' \subseteq S$ such that $\opt_\gamma = a(S \setminus S')\gamma b(T(S'))$. So, for every $\gamma \geq 0$, there is a set $S' \subseteq S$ such that $\gamma + a(S') - \gamma b(T(S')) \geq \Psi$. Then one can show that there are two subsets $S', S'' \subseteq S$ and a real number $z \in [0, 1]$ such that $\gamma + z(a(S') - \gamma b(T(S'))) + (1-z)(a(S'') - \gamma b(T(S''))) \geq \Psi$ for every $\gamma \geq 0$.  That means the coefficient $1 - z b(T(S')) + (1-z)b(T(S''))$ for $\gamma$ is  non-negative, and $z a(S') + (1-z) a(S'') \geq \Psi$. Then setting $y_v = z1_{S' \leadsto v} + (1-z)1_{S'' \leadsto v}, \forall v \in V$ gives a solution to LP\eqref{LP:prec-aggregate} of value $z a(S') + (1-z) a(S'') \geq \Psi$.
	
	Now we take the running time and the approximation parameters $\epsilon$ and $\phi$ into account. 
	We define an interesting set $\Gamma$ of $\gamma$'s as follows. Start from $\gamma = \frac{\phi}{3}$ and $\Gamma = \{\gamma\}$. While $\gamma < |\bfa|_1$ we do the following: $\gamma \gets (1+\epsilon)\gamma, \Gamma \gets \Gamma \cup \{\gamma\}$. Notice that we have $|\Gamma|  = O\left(\frac{\log (|\bfa|_1/\phi)}{\epsilon}\right) = \tilde O_{\epsilon} \left(\log \frac{|\bfa|_1}{\phi}\right)$. For every $\gamma \in \Gamma$, let $\opt_\gamma$ be the value of the optimum flow  for NFP$_\gamma$. We use Theorem~\ref{thm:mfmc} to find a flow $\bff^\gamma \in \calF_{\gamma}$, and a set $S'_\gamma \subseteq S$ with $a(S \setminus S'_\gamma) + \frac{\gamma b(T(S'_\gamma))}{1+\epsilon} \leq \val(\bff^\gamma) + \frac{\phi}{3} \leq \opt_\gamma +  \frac{\phi}{3}$. Then,
	\begin{align}
		\bfa\bfy^* &\leq  \min_{\gamma \in \Gamma} (|\bfa|_1 - \opt_\gamma + \gamma) \leq \min_{\gamma \in \Gamma} \left(|\bfa|_1 - a(S \setminus S'_\gamma)-\frac{\gamma b(T(S'_\gamma))}{1+\epsilon} + \gamma\right) + \frac{\phi}{3}\nonumber\\
		&= \min_{\gamma \in \Gamma} \left(a(S'_\gamma) - \frac{\gamma b(T(S'_\gamma))}{1+\epsilon} + \gamma\right) + \frac{\phi}{3}. \label{inequ:cy*-to-Psi}
	\end{align}
	The first inequality is by that for a fixed $\gamma \in \Gamma$, LP\eqref{LP:prec-aggregate-dual} can attain value $|\bfa|_1 - \opt_\gamma + \gamma$.
	
	Define $\Psi := \min_{\gamma \in \Gamma} \left(a(S'_\gamma) - \frac{\gamma b(T(S'_\gamma))}{1+\epsilon} + \gamma\right) \leq |\bfa|_1 + \frac{\phi}{3}$. Define two vectors $\tilde \bfa, \tilde \bfb \in \R_{\geq 0}^{\Gamma}$ as follows: For every $\gamma \in \Gamma$, let $\tilde a_\gamma = a(S'_\gamma)$ and $\tilde b_\gamma = b(T(S'_\gamma))$. 
	\begin{lemma}
		\label{lemma:at-least-Psi}
		The value of the following LP is at least $\Psi - \frac{2\phi}{3}$: maximize $\tilde \bfa\bfz$ subject to $\bfz \in \R_{\geq 0}^\Gamma, |\bfz|_1\leq 1$ and $\frac{\tilde \bfb \bfz}{(1+\epsilon)^2}\leq 1$.
	\end{lemma}
	\begin{proof}
		Assume otherwise. By duality, there is an $\alpha \in [0, \Psi - \frac{2\phi}{3}]$ such that $\frac{\alpha\tilde\bfb}{(1+\epsilon)^2} + (\Psi - \frac{2\phi}{3}-\alpha){\bf1} \geq \tilde\bfa$.   Consider the largest $\gamma \in \Gamma$ that is at most $\alpha$, or consider $\gamma = \frac{\phi}{3}$ if $\alpha < \frac{\phi}{3}$. So, $\alpha < (1+\epsilon)\gamma$ and $\alpha \geq \gamma - \frac{\phi}{3} > \gamma - \frac{2\phi}{3}$. Then $\frac{\alpha \tilde b_\gamma}{(1+\epsilon)^2} + \Psi-\frac{2\phi}{3}-\alpha \geq \tilde a_\gamma$ implies $\frac{\gamma \tilde b_\gamma}{1+\epsilon} + \Psi - \gamma > \tilde a_\gamma$, which is $\Psi > \tilde a_\gamma - \frac{\gamma \tilde b_\gamma}{1+\epsilon} + \gamma$. This contradicts the definition of $\Psi$.
	\end{proof}
	
	There are only two non-trivial linear constraints in LP in Lemma~\ref{lemma:at-least-Psi}. So its value can be attained by a $\bfz$ with at most $2$ non-zero coordinates. Then in time $\tilde O_{\epsilon}\left(\log^2\frac{|\bfa|_1}{\phi}\right)$, which is not a bottleneck, we can find a solution $\bfz$ to the LP, with value at least $\Psi - \frac{2\phi}{3}$. 
	
	For every $v \in V$,  let $y_v:= \sum_{\gamma: S'_\gamma \leadsto v}z_{\gamma}$. Then $y_v \leq |\bfz|_1 \leq 1$. Clearly $y_v \leq y_u$ for every $vu \in E$ as $S'_\gamma \leadsto v$ implies $S'_\gamma \leadsto u$. $\bfa\bfy = \sum_{\gamma \in \Gamma}z_{\gamma}a(S'_\gamma) = \tilde \bfa\bfz \geq \Psi - \frac{2\phi}{3}$. Similarly, $\bfb\bfy = \tilde \bfb\bfz \leq (1+\epsilon)^2$. Notice that $\bfa\bfy^*\leq\Psi + \frac{\phi}{3}$ by \eqref{inequ:cy*-to-Psi}, so $\bfa\bfy \geq \bfa\bfy^* - \phi$.  The running time of the algorithm is $\tilde O_\epsilon\left(\big(|E|\big)\log^2 \frac{|\bfa|_1}{\phi}\right)$ as we need to run the algorithm in Theorem~\ref{thm:mfmc} for $|\Gamma| = \tilde O_\epsilon\left(\log\frac{|\bfa|_1}{\phi}\right)$ times. 
\end{proof}

\section{Approximate Network Flow Algorithm: Proof of Theorem~\ref{thm:mfmc}} 
\label{sec:appendix-networkflow}


We highlight some key ideas behind the proof of Theorem~\ref{thm:mfmc}. 
Consider the special case where the graph $G$ is a directed bipartite graph from $S$ to $T$ and for convenience assume $\bfa$ and $\bfb$ are all-1 vectors. The problem then becomes fractional bipartite-matching.  
If we are allowed to scale up the capacities of the sinks $T$ to $1+\epsilon$, then there are augmenting paths of length $O(\frac{\log n}{\epsilon}) = \tilde O_\epsilon(1)$, if the current matching is not perfect.   We can use the shortest augmenting path algorithm of Hopcroft and Karp \cite{HK73} or Dinic \cite{Din70}: In each iteration one can in nearly-linear time find a ``blocking flow'' in the residual graph restricted to shortest augmenting paths. Then augmenting the matching increases the length of the shortest augmenting path by at least $2$ and so the algorithm terminates in $\tilde O_\epsilon(1)$ iterations.

Now, we move to the case where $G = (V \supseteq S \cup T, E)$ is a general graph. Let $H = (S \cup T, E_H)$ be the transitive closure of $G$ where $st \in E_H$ if and only if there is a path from $s$ to $t$ in $G$. So the maximum flow problem over $G$ is equivalent to that over $H$.  However, we can not construct and maintain $H$ explicitly as it may have quadratic number of edges.  Instead, we try to mimic the shortest augmenting path algorithm.  A shortest augmenting path in $H$ corresponds to an augmenting path in $G$ with the minimum number of switches between forward and backward edges. Then it is tempting to run the algorithm where in each iteration we do the following: Construct the residual graph for the current flow $\bff$,  restrict it to the augmenting paths with the minimum number of switches, find a blocking flow in the graph and use it to augment $\bff$.

However, unlike the bipartite graph case, the above operations do not necessarily increase the minimum number of switches in an augmenting path for $\bff$. This is due to the interference between the forward and backward ``segments'' of an augmenting path.  To address the issue, we separate forward and backward segments, using the structure of  ``handled graphs'': A handle is a copy of a sub-graph of $G$, with sources and sinks identified with those in $G$. A handled graph $G'$ is the graph $G$ with many handles.  When augmenting the flow $\bff$ over the graph $G'$, we only consider the forward edges in $G$, and backward edges in the handles of $G'$. This way, the forward and backward edges have disjoint supports. By carefully constructing the handles, we show that augmenting $\bff$ by a blocking flow in the handled graph can increase the minimum number of switches in an augmenting path increases by at least 2.  The algorithm then repeats the procedure $\tilde O_\epsilon(1)$ times.

From now on, we focus on the proof of  Theorem~\ref{thm:mfmc}. For convenience, it is repeated below. 
\thmmfmc*

In Theorem~\ref{thm:mfmc}, $\gamma$ is fixed. So we omit the subscript $\gamma$ in $\calF_\gamma$ and simply use $\calF$ for $\calF_\gamma$.

\subsection{Handled Graphs and Shortcut Graphs}
\label{subsec:networkflow-notation}
In this section, we introduce two important structures: handled graphs and shortcut graphs. 
\begin{definition}
	Let $\hat G = (\hat V, \hat E)$ be a sub-graph of $G$. 
	A directed graph $\tilde G = (\tilde V, \tilde E)$ is \emph{a copy of $\hat G$} if 
	\begin{itemize}
		\item $\tilde V \cap (S \cup T) = \hat V \cap (S \cup T)$,
		\item $\tilde V \setminus (S\cup T)$ and $\hat V \setminus (S\cup T)$ are disjoint, and
		\item there is a bijection $\pi:\tilde V \to \hat V$ such that $\pi(v) = v$ for every $v \in \tilde V \cap (S \cup T) = \hat V \cap (S \cup T)$, and $uv  \in \tilde E$ if and only if $\pi(u)\pi(v) \in \hat E$ for every $u, v \in \tilde V$.
	\end{itemize}
	For every $v \in \tilde V$, we say $\pi(v)$ is the pre-image of $v$. 
\end{definition} 
So $\tilde G$ satisfies the definition if it is obtained from $\hat G$ by copying everything except the sources and sinks.\footnote{There are no edges from sources to sinks in $G$ so the first condition in the definition implies $E \cap \tilde E = \emptyset$.}  
Throughout the paper, we shall use $\pi$ to denote the function that maps all vertices in all copies of sub-graphs we ever defined to their pre-images.  For convenience, we also let $\pi(v) = v$ for every $v \in V$.  We guarantee the $\pi(v)$ of every $v$ is known to our algorithm.

\begin{definition}
	Let  $G^1 = (V^1, E^1), G^2 = (V^2, E^2), \cdots, G^k = (V^k, E^k)$ be $k$ copies of sub-graphs of $G$ for some integer $k \geq 0$, such that for every $1 \leq i < j \leq k$ we have $V^i \cap V^j \subseteq S \cup T$. 
	We say $G' = (V', E')$, where $V' = V \cup V^1 \cup V^2 \cup \cdots \cup V^k$ and $E' = E \uplus E^1 \uplus E^2 \uplus \cdots \uplus E^k$,  is a \emph{handled graph}, and $G^1, G^2, \cdots, G^k$ are called the \emph{handles} of $G'$.
\end{definition} 

See Figure~\ref{fig:handled-graph} for an illustration of handled graphs. Notice that the handles of $G'$ do not improve the connectivity from sources to sinks: For every $s \in S$ and $t \in T$, we have $s \leadsto_{G'} t$ if and only if $s \leadsto t$. 
The usefulness of handled graphs will be discussed later.

\begin{figure}
	\centering
	\begin{minipage}{0.45\textwidth}
		\centering
		\includegraphics[width=0.6\textwidth]{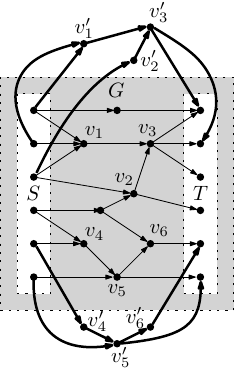}
		\caption{A handled graph. $G$ is the graph contained in the gray rectangle.  $S$ and $T$ are respectively the vertices in the left and right small white rectangles. The thick edges are edges in handles. $v'_1, v'_2, v'_3, v'_4, v'_5$ and $v'_6$ are copies of $v_1, v_2, v_3, v_4, v_5$ and $v_6$ respectively.}
		\label{fig:handled-graph}
	\end{minipage}\hfill%
	\begin{minipage}{0.45\textwidth}
		\centering
		\includegraphics[width=0.5\textwidth]{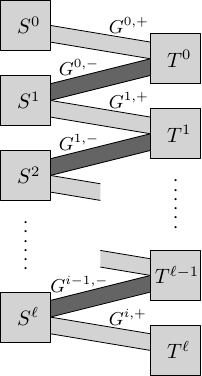}
		\caption{$S^i$'s, $T^i$'s, $G^{i, +}$'s and $G^{i,-}$'s defined in Algorithm~\ref{alg:inclen}. The
			light gray components are inside $G$, and the dark gray components are copies of $G$.}
		\label{fig:inclen-graphs}
	\end{minipage}
\end{figure}

We extend the definition of a valid flow to sub-graphs, copies of sub-graphs and handled graphs:
\begin{definition}
	Let $G' = (V', E')$ be a sub-graph of $G$, or a copy of a sub-graph of $G$, or a handled graph.  A valid flow $\bff'$ for $G'$ is a vector in $\R_{\geq 0}^{E'}$ satisfying $f'\big(\delta^+_{G'}(s)\big) \leq a_s$ for every $s \in V' \cap S$, $f'\big(\delta^-_{G'}(t)\big) \leq \gamma b_t$ for every $t \in V' \cap T$, and $f'\big(\delta^+_{G'}(v)\big) = f'\big(\delta^-_{G'}(v)\big)$ for every $v \in V' \setminus (S \cup T)$.  Let $\calF^{G'}$ be the set of all valid flows for $G'$. The value of a $\bff' \in \calF^{G'}$, denoted as $\val(\bff')$,  is defined as $f'(\delta^+_{G'}(V' \cap S)) = f'\big(\delta^-_{G'}(V' \cap T)\big)$.
\end{definition}
Given $\bff' \in \calF^{G'}$ for some $G'$, we define the support of $\bff'$, denoted as $\supp(\bff')$, as the following sub-graph $G'' = (V'', E'')$ of $G'$: $E''$ is the set of edges with positive $\bff'$ values, and $V''$ is the set of vertices incident to at least one edge in $E''$.  We use $V_{\supp(\bff')}$ and $E_{\supp(\bff')}$ to denote the vertices and edges in $\supp(\bff')$ respectively. 

\medskip

\begin{definition}[Shortcut Edges and Graphs]
	Given $s \in S, t \in T$ with $s \leadsto t$, we say $st$ is a shortcut edge. Let $G' = (V', E')$ be a handled graph and $\bff' \in \calF^{G'}$.  For any $s \in S, t \in T$, we say $ts$ is a backward shortcut edge w.r.t $\bff'$ if $s \leadsto_{\supp(\bff')} t$. A shortcut edge w.r.t $\bff'$ is defined as either a forward shortcut edge or a backward shortcut edge w.r.t $\bff'$.
	
	Let $G' = (V', E')$ be a handled graph and $\bff' \in \calF^{G'}$, the shortcut graph $H$ for $\bff'$ is the graph $(S\cup T, E_H)$, where $E_H$ is the set of shortcut edges w.r.t $\bff'$.
\end{definition}
Notice that we can not maintain a shortcut graph explicitly since its size might be quadratic in $|V|$. 

\begin{definition}[Augmenting Shortcut Path and Alternating Shortcut Path]
	Let $G' = (V', E')$ be a handled graph and $\bff' \in \calF^{G'}$, and let $H$ be the shortcut graph for $\bff'$.  We say a source $s \in S$ is satisfied w.r.t $\bff'$ if $f(\delta^+(s)) = a_s$ and unsatisfied otherwise. We say a sink $t \in T$ is saturated w.r.t $\bff'$ if $f(\delta^-(t)) = b_t$ and unsaturated otherwise. 
	
	An alternating shortcut path in $H$ is a simple path starting from an unsatisfied vertex $s \in S$. An alternating shortcut path in $H$ is said to be an augmenting shortcut path in $H$ if it ends at an unsaturated vertex $t \in T$. 
\end{definition}

\subsection{Long Augmenting Shortcut Paths Imply $(1+\epsilon)$-Approximate Flow}
In this section, we prove that if we have a flow $\bff'$ in some handled graph whose shortcut graph does not contain a short augmenting path, then we can find a set $S' \subseteq S$ satisfying the property of Theorem~\ref{thm:mfmc}. In the proof and throughout the rest of  Section~\ref{sec:appendix-networkflow}, we use $S^{\leq i}$ as a shorthand for $\union_{i' \leq i}S^i$. $S^{<i}, S^{\geq i}, S^{> i}, T^{\leq i}, T^{<i}, T^{\geq i}$ and $T^{>i}$ are defined similarly.   
\begin{lemma}
	\label{lemma:find-cut} 
	Let $G' = (V', E')$ be a handled graph and $\bff' \in \calF^{G'}$. 
	Let $L = \floor{\log_{1+\epsilon}\frac{3|\bfa|_1}{\phi}} = \tilde O_\epsilon\left(\log \frac{|\bfa|_1}{\phi}\right)$.  Assume the shortcut graph $H$ for $\bff'$ does not contain an augmenting path of length at most $2L+1$. 
	Then in time $O_\epsilon(|E'|)$ we can find a set $S' \subseteq S$  such that $a(S \setminus S') + \frac{\gamma b(T(S'))}{1+\epsilon} \leq \val(\bff') + \frac{\phi}{3}$. 
\end{lemma}	
\begin{proof}
	Let $H = (S \cup T, E_H)$ be the shortcut graph for $\bff'$.  For every integer $\ell \in [0, L]$, let $S^\ell$ be the set of vertices in $S$ to which the shortest alternating path in $H$ has length exactly $2\ell$, and let $T^\ell$ be the set of vertices in $T$ to which the shortest alternating path in $H$ has length exactly $2\ell+1$.  Let $S^{L+1} = S \setminus S^{\leq L}$ and $T^{L+1} = T \setminus T^{\leq L}$, so that $(S_\ell)_{\ell \in [0, L+1]}$ and $(T_\ell)_{\ell \in [0, L+1]}$ form partitions of $S$ and $T$ respectively.
	
	We find the ${\ell^*} \in [0, L]$ with minimum $a(S^{>{\ell^*}}) + \frac{\gamma b(T^{\leq {\ell^*}})}{1+\epsilon}$ and output $S' := S^{\leq {\ell^*}}$. We prove that $a(S \setminus S') + \frac{\gamma b(T(S'))}{1+\epsilon} = a(S^{>{\ell^*}}) + \frac{\gamma b(T^{\leq {\ell^*}})}{1+\epsilon} \leq \val(\bff') + \frac{\phi}{3}$.  Assume towards the contradiction that $a(S^{>{\ell^*}}) + \frac{\gamma b(T^{\leq {\ell^*}})}{1+\epsilon} > \val(\bff') + \frac{\phi}{3}$.  Then, for every $\ell \in [0, L]$, we have $a(S^{>\ell}) + \frac{\gamma b(T^{\leq \ell})}{1+\epsilon}>\val(\bff') + \frac{\phi}{3}$ by the way we choose $\ell^*$, which is $\val(\bff') - a(S^{>\ell}) + \frac{\phi}{3} < \frac{\gamma b(T^{\leq \ell})}{1+\epsilon}$.
	
	
	We prove  $\gamma b(T^{\leq \ell})\leq \val(\bff') - a(S^{> \ell +1})$ for every $\ell \in [0, L]$. In the flow $\bff'$, all sinks in $T^{\leq \ell}$ are saturated  as there are no augmenting path of length at most $2\ell+1$ in $H$.  So $\gamma b(T^{\leq \ell})$ units flow are sent to $T^{\leq \ell}$ in $\bff'$. The senders of this flow are in $S^{\leq \ell + 1}$.  Moreover, as the sources in $S^{>\ell+1}$ are saturated, the sources in $S^{\leq \ell+1}$ sent $\val(\bff) - a(S^{>\ell+1})$ units of flow.
	Hence $\gamma b(T^{\leq \ell})\leq \val(\bff) - a(S^{>\ell+1})$. 
	
	Therefore, we have proved that 	$\val(\bff) - a(S^{>\ell+1}) > (1+\epsilon)\left(\val(\bff) - a(S^{>\ell}) + \frac{\phi}{3}\right)$ for every $\ell \in [0, L]$.  This gives us $\val(\bff) = \val(\bff) - a(S^{>L+1}) > (1+\epsilon)^{L + 1}\big(\val(\bff) - a(S^{>0}) + \frac{\phi}{3}\big) \geq (1+\epsilon)^{L+1}\frac{\phi}{3}$. However, as $\val(\bff') \leq |\bfa|_1$, we have a contradiction by our definition of $L$.
	
	To construct the partitions $(S_\ell)_{\ell \in [0, L+1]}$ and $(T_\ell)_{\ell \in [0, L+1]}$, we define the following directed graph.
	We have vertices $\{(v, o): v \in V', o \in \{0 ,1\}\}$. For every $t \in T$, we have an edge $((t, 0), (t, 1))$ of length 1. For every $s \in T$, we have an edge $((s, 1), (s, 0))$ of length $1$. For every $vu\in E'$, we have an edge $((v, 0), (u, 0))$ of length 0. For every $vu \in \supp(\bff')$, we have an edge $((u, 1), (v, 1))$ of length $0$. 	
	Then $(S_\ell)_{\ell \in [0, L+1]}$ and $(T_\ell)_{\ell \in [0, L+1]}$ can be constructed using a variant of BFS that takes care of length-$0$ edges.  The running time of the algorithm is $O(|E'|)$.
\end{proof}
With Lemma~\ref{lemma:find-cut}, it remains to construct a handled graph $G'$ and a flow $\bff' \in \calF^{G'}$ satisfying the condition of the lemma. To achieve the goal, we need more definitions and tools. 

\subsection{Sub-Flows, Projections and Blocking Flows}

\begin{definition}
	\label{def:sub-flow}
	Let $G' = (V', E')$ be a handled graph, and $\bff' \in \calF^{G'}$. Let $S' \subseteq S$ be a subset of sources. Then, a sub-flow of $\bff'$ sent from $S'$ is a flow $\bff'' \in \calF^{G'}$ satisfying
	\begin{enumerate}[label=(\ref{def:sub-flow}\alph*),leftmargin=*]
		\item $f''_e \leq f'_e$ for every $e \in E'$,
		\item $f''_e = f'_e$ for every $e \in \delta^+_{G'}(S')$, and
		\item $f''_e = 0$ for every $e \in \delta^+_{G'}(S \setminus S')$.
	\end{enumerate}
	
	Let $T' \subseteq T$ be a subset of sinks. Similarly, a sub-flow of $\bff'$ received by $T'$ is a flow $\bff'' \in \calF^{G'}$ satisfying 
	\begin{enumerate}[label=(\ref{def:sub-flow}\alph*),leftmargin=*,start=4]
		\item $f''_e \leq f'_e$ for every $e \in E'$, 
		\item $f''_e = f'_e$ for every $e \in \delta^-_{G'}(T')$, and
		\item $f''_e = 0$ for every $e \in \delta^-_{G'}(T \setminus T')$.
	\end{enumerate} 
\end{definition}

\begin{lemma}
	Let $G' = (V', E')$ be a handled graph, and $\bff' \in \calF^{G'}$, and $S' \subseteq S$. Then we can find a sub-flow $\bff''$ of $\bff'$ sent from $S'$ in time $O\big(|S'| + |V_{\supp(\bff'')}|\cdot\log |V'| + |E_{\supp(\bff'')}|\big)$.
\end{lemma}
\begin{proof}
	Initially, we have $f'(\delta^+_{G'}(s))$ units of commodity at any $s \in S'$, and $0$ units of commodity elsewhere. We then process the vertices $V' \setminus T$ in topological order one by one.  When processing a vertex $v \in V \setminus T$, we push the commodity at $v$ to its out-neighbors along edges in $\delta^-_{G'}(v)$, with the only constraint being $\bff'' \leq \bff'$. 
	
	We assume the topological ordering of $G$ is computed at the beginning of the whole algorithm and for every $v \in V$ we know the rank of $v$ in the ordering. Notice that the rank of vertices in $G$ can be extended to rank of vertices in $G'$. Then we use a priority-queue data structure to store the vertices which hold the commodity, with rank being the priority function.  The overall running time of the algorithm can be bounded by $O\big(|S'| + |V_{\supp(\bff'')}|\cdot\log |V'| + |E_{\supp(\bff'')}|\big)$.
\end{proof}
Similarly, the following lemma can be proved:
\begin{lemma}
	Let $G' = (V', E')$ be a handled graph, and $\bff' \in \calF^{G'}$, and $T' \subseteq T$. Then we can find a sub-flow $\bff''$ of $\bff'$ received by $T'$ in time $O\big(|T'| + |V_{\supp(\bff'')}|\cdot\log |V'| + |E_{\supp(\bff'')}|\big)$.
\end{lemma}

\begin{definition}
	Let $G' = (V', E')$ be a handled graph, and $\bff' \in \calF^{G'}$.  Let $\tilde G = (\tilde V, \tilde E)$ be a sub-graph of $G$, or a copy of a sub-graph of $G$.  If for every $uv  \in \supp(\bff')$, we have some $(\tilde u, \tilde v) \in \tilde E$ with $\pi(u) = \pi(\tilde u)$ and $\pi(v) = \pi(\tilde v)$, then we say $\bff'$ can be projected to $\tilde G$. (Notice that the $(\tilde u, \tilde v)$ is unique, if it exists.) Otherwise, we say $\bff'$ can not be projected to $\tilde G$.  In the former case, we define the projection of $\bff'$ to $\tilde G$ to be the vector $\tilde \bff \in \R_{\geq 0}^{\tilde E}$ satisfying:
	\begin{align*}
		\tilde f_{\tilde u\tilde v} = f'\Big(\big\{uv  \in E': \pi(u) = \pi(\tilde u), \pi(v) = \pi(\tilde v)\big\}\Big) , \qquad \forall \tilde u\tilde v \in \tilde E.
	\end{align*}
\end{definition}
Clearly, in the above definition, the projection $\tilde \bff$ of $\bff'$ to $\tilde G$ has $\tilde\bff \in \calF^{\tilde G}$ and $\val(\tilde \bff) = \val(\bff')$.

Dinic \cite{Din70} introduced the notion of blocking flows, which is a $s$-$t$ flow such that every $s$-$t$ path in the graph $G$ has an edge that is full. 
\begin{definition}[\cite{Din70}]
	\label{def:blocking-flows}
	Let $R = (V_R, E_R)$ be a directed graph with two special vertices $s^*, t^* \in V_R$ such that $\delta^-(s^*) = \emptyset$ and $\delta^+(t^*) = \emptyset$. Let $\bfc \in [0, \infty]^{E_R}$ be a capacity vector on $E_R$. A blocking flow in $(R, \bfc)$ is a vector $\bfg \in \R_{\geq 0}^{E_R}$ satisfying $\bfg \leq \bfc$, $g(\delta^+(v)) = g(\delta^-(v))$ for every $v \in V_R \setminus \{s^*, t^*\}$, and every path $P$ from $s^*$ to $t^*$ in $R$ contains an edge $e$ with $g_e = c_e$. 
\end{definition}
So, a flow $\bfg$ is a block flow if we can not increase its value by only increasing $g_e$ values. In an influential paper of Sleator and Tarjan \cite{ST83}, they developed the dynamic tree (also known as link cut tree) data structure and showed how it can be used to find a blocking flow in $O(|E_R|\log |V_R|)$-time.
\begin{theorem}[\cite{ST83}]
	\label{thm:find-blocking-flow}
	Let $R = (V_R, E_R)$, $s^*, t^*$ and $\bfc$ as defined in Definition~\ref{def:blocking-flows}.
	There is $O\big(|E_R|\log |V_R|\big)$-time algorithm that finds a blocking flow in $(R, \bfc)$.
\end{theorem}

\subsection{Augmenting using Shortest Shortcut Paths}
Now we show how to find $G'$ and $\bff'$ satisfying the property of Lemma~\ref{lemma:find-cut}. We maintain $G'$ and $\bff' \in \calF^{G'}$ and repeatedly augment $\bff'$ along shortest augmenting paths in the shortcut graph $H'$ for $\bff'$. During the procedure, we need to change the handled graph $G'$ from iteration to iteration. The following core theorem states that we can increase the length of the shortest augmenting path in $H'$ by $2$ in nearly-linear time.
\begin{theorem}
	\label{thm:inclen-main}
	Let $G^\circ = (V^\circ, E^\circ)$ be a handled graph, $\bff^\circ \in \calF^{G^\circ}$, and $H^\circ$ be the shortcut graph for $f^\circ$. Let $\ell \geq 0$ be an integer such that the length of the shortest augmenting path in $H^\circ$ is at least $2\ell+1$.   Given $\ell, G^\circ, f^\circ$ and $\epsilon > 0$, there is an $O(|E^\circ|\log|V^\circ| )$-time algorithm that outputs a handled graph $G' = (V', E')$ with $|V'| \leq 3|V|$ and $|E'| \leq 3|E|$,\footnote{We use concrete constants here to avoid abuse of $O(\cdot)$ notation caused by applying the theorem repeatedly.} and a flow $\bff' \in \calF^{G'}$, such that in the shortcut graph $H'$ for $\bff'$, the shortest augmenting path has length at least $2\ell + 3$.
\end{theorem}

The algorithm is described in Algorithm~\ref{alg:inclen}. 
In the algorithm, the sub-graph of $G$ between a set $S' \subseteq S$ and $T' \subseteq T$ is defined as the sub-graph of $G$ induced by $\{v \in V: S' \leadsto v, v \leadsto T'\}$.

\begin{algorithm}[ht]
	\caption{$\inclen(\ell, G^\circ = (V^\circ, E^\circ), \bff^\circ)$}
	\label{alg:inclen}
	\begin{algorithmic}[1]
		\Require{integer $\ell \geq 0$, a handled graph $G^\circ = (V^\circ, E^\circ)$, a flow $\bff^\circ \in \calF^{G^\circ}$, which defines the shortcut graph $H^\circ$ for $\bff^\circ$.  The shortest augmenting path in $H^\circ$ has length at least $2\ell+1$.}
		\Ensure{a handled graph $G'$ and a flow $\bff' \in \calF^{G'}$ such that the shortest augmenting shortcut path in the shortcut graph $H'$ for $\bff'$ has length at least $2\ell+3$.} \medskip
		\State let $S^i$ be the sources to which the shortest alternating shortcut path in $H^\circ$ has length $2\ell$, $\forall i \in [0, \ell]$ \label{step:inclen-define-Si}
		\State let $T^i$ be the sinks to which the shortest alternating shortcut path in $H^\circ$ has length $2\ell+1$, $\forall i \in [0, \ell]$ \label{step:inclen-define-Ti}
		\State let $G^{i, +}$ be the sub-graph of $G$ between $S^i$ and $T^i$, $\forall i \in [0, \ell]$ \label{step:inclen-define-Gi+}
		\State let $G^{i, -}$ be a copy of the sub-graph of $G$ between $S^{i+1}$ and $T^i$, $\forall i \in [0, \ell-1]$ \label{step:inclen-define-Gi-}
		\State let $G' = (V', E')$ be the handled graph with handles $\{G^{i, -}:i \in [0, \ell-1]\}$ \label{step:inclen-define-G'}
		\For{$i \gets 0$ to $\ell$} \label{step:inclen-loop}
		\State find a sub-flow $\bff^{\circ(i, +)}$ of $\bff^\circ$ sent by $S^i$, \quad$\bff^\circ \gets \bff^\circ  - \bff^{\circ(i, +)}$ \label{step:inclen-find-y+}
		\State $\bff'^{(i, +)} \gets $ projection of $\bff^{\circ(i, +)}$ to $G^{i, +}$  \Comment{See Claim~\ref{claim:how-fcirc-changes}}
		\If{$i =\ell$} \textbf{break} \EndIf
		\State find a sub-flow $\bff^{\circ(i, -)}$ of $\bff^\circ$ received by $T^i$, \quad$\bff^\circ \gets \bff^\circ - \bff^{\circ(i, -)}$ \label{step:inclen-find-y-}
		\State $\bff'^{(i, -)} \gets $ projection of $\bff^{\circ(i, -)}$ to $G^{i, -}$  \Comment{See Claim~\ref{claim:how-fcirc-changes}}
		\EndFor\State\textbf{end for}
		\State let $\bff' \in \calF^{G'}$ be defined as $\bff' := \sum_{i = 0}^{\ell}\bff'^{(i, +)} + \sum_{i = 0}^{\ell-1}\bff'^{(i, -)}$, assuming coordinates that do not exist are 0 
		\label{step:inclen-sum}
		
		\State $R = (V' \cup \{s^*, t^*\}, E_R)$ and $\bfc \in [0, \infty]^{E_R}$ be a graph where $E_R$ and $\bfc$ are constructed as follows: \label{step:inclen-R}
		\begin{itemize}[leftmargin=*]
			\item for every $i \in [0, \ell]$ and edge $e$ in $G^{i, +}$, add $e$ to $E_R$ and let $c_e = \infty$,
			\item for every $i \in [0, \ell-1]$ and edge $vu$ in $G^{i, -}$, add $uv $ to $E_R$ and let $c_{uv} = f'_{vu}$,
			\item for every $s \in S^0$, add $s^*s$ to $E_R$, and let $c_{s^*s} =  a_s - f'(\delta^+(s))$,
			\item for every $t \in T^{\ell}$, add $tt^*$ to $E_R$, and let $c_{tt^*} = b_t - f'(\delta^-(t))$.
		\end{itemize}
		\State find a blocking flow $\bfg$ for $(R, \bfc)$ using Theorem~\ref{thm:find-blocking-flow} \label{step:inclen-find-blocking-flow}
		\State for every edge $e$ in $G^{i,+}$ for some $i$, let $f'_e \gets f'_e + g_e$, and for every edge $vu$ in $G^{i, -}$ for some $i$, let $f'_{uv} \gets f'_{uv} - g_{vu}$ \label{step:inclen-use-blocking-flow}
		\State \Return $(G', \bff')$
	\end{algorithmic}
\end{algorithm}

In Steps~\ref{step:inclen-define-Si} to \ref{step:inclen-define-G'}, we define $S^i$'s, $T^i$'s, $G^{i, +}$'s, $G^{i,-}$'s and $G'$. See Figure~\ref{fig:inclen-graphs} for an illustration. By our assumption that every augmenting path in $H^\circ$ has length at least $2\ell+1$,  all sinks in $T^{< \ell}$ are saturated by the initial $\bff^\circ$.  Notice that $G^{i, +}$'s are sub-graphs of $G$, but $G^{i, -}$'s are copies of sub-graphs and are included in $G'$ as handles. In Loop~\ref{step:inclen-loop} of Algorithm~\ref{alg:inclen}, we construct the flows $\bff^{\circ(i, +)}$'s, $\bff'^{(i, +)}$'s, $\bff^{\circ(i, -)}$'s and $\bff'^{(i, -)}$'s.  The following claimed can be proved via mathematical induction:
\begin{claim}
	\label{claim:how-fcirc-changes}
	Focus on the iteration $i$ of Loop~\ref{step:inclen-loop}.
	\begin{itemize}
		\item At the beginning of the iteration, $\bff^\circ$ is a flow from $S^{\geq i}$ to $T^{\geq i}$.
		\item $\bff^{\circ(i, +)}$ can be projected to $G^{i, +}$.
		\item Before Step~\ref{step:inclen-find-y-} in the iteration, $\bff^\circ$ is a flow from $S^{\geq i + 1}$ to $T^{\geq i}$. 
		\item $\bff^{\circ(i, -)}$ can be projected to $G^{i, -}$.
	\end{itemize}
\end{claim}
\begin{proof}
	Assume that at the beginning of iteration $i$ of Loop~\ref{step:inclen-loop}, $\bff^\circ$ is a flow from $S^{\geq i}$ to $T^{\geq i}$; this holds for $i = 0$.  Then $\bff^{\circ(i, +)}$ is a flow from $S^i$ to $T^i$ since $S^i \not\leadsto_{G^\circ}T^{>i}$. So it can be projected to $G^{i, +}$. Then before Step~\ref{step:inclen-find-y-} in the iteration,  $\bff^\circ$ is a flow from $S^{\geq i+1}$ to $T^{\geq i}$. $\bff^{\circ(i, -)}$ is a flow from $S^{i+1}$ to $T^i$ since $S^{\geq i+2} \not\leadsto_{\supp(\bff^\circ)} T^i$.  So, it can be projected to $G^{i, -}$. In the end of the iteration $i$ and thus at the beginning of iteration $i+1$, $\bff^\circ$ is a flow from $S^{\geq i+1}$ to $T^{\geq i+1}$.
\end{proof}	

In Step~\ref{step:inclen-sum}, we construct $\bff'$ by summing up $\bff'^{(i, +)}$'s and $\bff'^{(i, -)}$'s. Then, we define a residual graph $R$, find a blocking flow $\bfg$ in the graph, augment $\bff'$ using $\bfg$ in Steps~\ref{step:inclen-R}, \ref{step:inclen-find-blocking-flow} and \ref{step:inclen-use-blocking-flow}. We return $(G', \bff')$ in the end.

Now we can prove the key lemma that establishes the correctness of the algorithm:
\begin{lemma}
	Let $H'$ be the shortcut graph w.r.t $\bff'$ returned by Algorithm~\ref{alg:inclen}.  Then any augmenting path in $H'$ has length at least $2\ell+3$.
\end{lemma}
\begin{proof}
	To avoid confusion, we use $\bar \bff'$ be the flow $\bff'$ obtained after Step~\ref{step:inclen-sum}, that is, before it is augmented. We use $\bff'$ be the final $\bff'$ returned by the algorithm. 
	Let $\bar H'$ and $H'$ be the shortcut graphs for $\bar \bff'$ and $\bff'$ respectively.  Let $S^{\ell+1} := S \setminus S^{\leq \ell}$ and $T^{\ell+1} := T \setminus T^{\leq \ell}$, so that $(S_i)_{i \in [0, \ell+1]}$ and $(T_i)_{i \in [0, \ell+1]}$ are partitions of $S$ and $T$ respectively.
	
	The following two properties hold:
	\begin{enumerate}[label=(P\arabic*),leftmargin=*]
		\item If a forward edge in $\bar H'$ connects $S^i$ to $T^{i'}$, then we have $i' \leq i$. 
		\item If a backward edge in $\bar H'$ connects $T^i$ to $S^{i'}$, then $i' \in \{i, i+1\}$. 
	\end{enumerate}
	(P1)  follows from the definition of $S^i$'s and $T^i$'s.  (P2) follows from that $G^{i, +}$'s and $G^{i,-}$'s are internally disjoint, and that $\bar\bff'$ has support in these graphs. 
	
	If we focus on the sequence $S^0, T^0, S^1, T^1, \cdots, S^{\ell}, T^{\ell}$ of vertex sets, every edge in $\bar H'$ can only increase the position of the vertex in the sequence by $1$.  All unsatisfied sources are in $S^0$ and all  unsaturated sinks are in $T^{\ell} \cup T^{\ell+1}$. Therefore, an augmenting path in $\bar H'$ has length at least $2\ell+1$. For it to have length exactly $2\ell+1$, it can only use \emph{useful} shortcut edges in $\bar H'$: A forward shortcut edge is useful if it connects $S^i$ to $T^i$ for some $i \in [0, \ell]$, and a backward shortcut edge is useful if it connects $T^i$ to $S^{i+1}$ for some $i \in [0, \ell-1]$.  Moreover, if there is a backward edge from $T^i$ to $S^{i+1}$ in $\bar H'$, the correspondent path from $S^{i+1}$ to $T^i$ in $\supp(\bar\bff')$ must be completely in $G^{i, -}$.
	
	We then consider how augmenting $\bar \bff'$ using $\bfg$ to $\bff'$ changes the set of backward shortcut edges. The set of backward shortcut edges created by paths in the handles can only shrink, since the augmenting operation can only decrease $\bff'_e$ values of edges in the handles. The backward shortcut edges created by paths in $G^{i, +}$ are not useful, as they connect $T^i$ to $S^i$. Therefore the set of useful backward edges can only shrink from $\bar H' $ to $H'$. Suppose there is an augmenting path of length $2\ell + 1$ in $H'$. It must be an augmenting path in $\bar H'$.  However, as $\bfg$ is a blocking flow in $R$, one of the (backward) shortcut edge in the path must be broken in $H'$, a contradiction.  Therefore, there are no augmenting paths of length $2\ell+1$ in $H'$. The lemma follows since an augmenting shortcut path has length being an odd number.	
\end{proof}

We remark that it is crucial for us to make $G^{i, -}$'s and $G$ internally disjoint using handles. If we let $G^{i,-}$'s be sub-graphs of $G$, then increasing flow values in $G^{i, +}$ may create new backward shortcut edges from $T^i$ to $S^{i+1}$. Though augmenting $\bff'$ by $\bfg$ destroys the old augmenting paths of length $2\ell+1$ in $\bar H'$, it may create new ones.  At the other extreme, we could decompose $\bff^{\circ(i, -)}$'s completely into paths, and maintain a set of source-sink pairs, each with the amount of flow sent between them. But this way we could not bound the number of such pairs as the algorithm proceeds.  Moreover, in the end, we have to realize the flows sent between the pairs in $G$. Even assuming in the end we have, say, $O(|E|)$ $st$-pairs with positive amount of flow sent in between, we do not know how to realize the flows in $G$ in nearly-linear time. So, the handled graph gives an approach between the two extremes, which can guarantee that the length of the shortest augmenting shortcut paths increases from iteration to iteration, and that the graphs we maintain have nearly-linear size. 

We now show the algorithm has nearly-linear running time. The following are two simple but useful observations we can make:
\begin{obs}
	\label{obs:inclen-a-few-times}
	For each vertex $v \in V$, let $i$ be the smallest integer such that $S^i \leadsto v$. Then $v$ is not in $G^{i', +}$ for any $i' \neq i$. $v$ does not appear in $G^{i', -}$ as a copy for any $i' \notin \{i-1, i\}$.
\end{obs}
\begin{proof}
	For any $i' < i$, we have $S^{i'} \not\leadsto v$. All the vertices $t \in T$ with $v \leadsto t$ are included in $T^{\leq i}$.  So, for every $i' > i$ we have $v \not\leadsto T^{i'}$. Therefore the two statements follow.
\end{proof}

\begin{obs}
	\label{obs:inclen-a-few-times-1}
	Any $v \in V^\circ$ is in the support of at most $3$ flows in $\{ \bff^{\circ(i, +)}: i \in [0, \ell] \} \cup \{\bff^{\circ(i, -)}: i \in [0, \ell-1]\}$.
\end{obs}
\begin{proof}
	Let $i$ be the smallest integer such that $S^i\leadsto_{G^\circ} v$. Then, $S^{i'} \not\leadsto_{G^\circ} v$ for every $i' < i$. Also $v \not \leadsto_{G^\circ} T^{i'}$ for every $i' > i$ since otherwise $S^i \leadsto_{G^\circ} T^{i'}$, which implies $S^i \leadsto_G T^{i'}$, a contradiction.  So, if $v \in V_{\supp(\bff^{\circ(i', +)})}$, then $i' = i$.  If $v \in V_{\supp(\bff^{\circ(i', -)})}$, then $i' \in \{i - 1, i\}$.
\end{proof}

Steps \ref{step:inclen-define-Si} and \ref{step:inclen-define-Ti} can be implemented in $O(|E^\circ|)$ time using BFS.    By Observation~\ref{obs:inclen-a-few-times}, we have $|V'| \leq 3|V|$ and $|E'| \leq 3|E|$, and Step \ref{step:inclen-define-Gi+}, \ref{step:inclen-define-Gi-} and \ref{step:inclen-define-G'} can be implemented in time $O(|E^\circ|)$. Constructing the sub-flows $\bff^{\circ(i, +)}$ and $\bff^{\circ(i, -)}$ takes time $O\big(|S^i|  + |V_\supp(\bff^{\circ(i, +)})|\cdot \log |V^\circ| + |E_{\supp(\bff^{\circ(i, +)})}|\big)$ and $O\big(|T^i|  + |V_\supp(\bff^{\circ(i, -)})|\cdot \log |V^\circ| + |E_{\supp(\bff^{\circ(i, -)})}|\big)$ respectively. By Observation~\ref{obs:inclen-a-few-times-1}, Loop~\ref{step:inclen-loop} takes time $O(|V^\circ| \log |V^\circ| + |E^\circ|)$. Steps~\ref{step:inclen-sum} and \ref{step:inclen-R} take time $O(|E^\circ|)$. The bottleneck of the algorithm is Step~\ref{step:inclen-use-blocking-flow}, which runs in time $O(|E^\circ|\log |V^\circ|)$ using Theorem~\ref{thm:find-blocking-flow}.   This finishes the proof of Theorem~\ref{thm:inclen-main}.

	%

\subsection{Finishing Proof of Theorem~\ref{thm:mfmc}}
\label{subsec:networkflow-finish}
In this section, we can wrap up the proof of Theorem~\ref{thm:mfmc}. Let $L := \floor{\log_{1+\epsilon}\frac{3|\bfa|_1}\phi} $ as in Lemma~\ref{lemma:find-cut}. We run Algorithm~\ref{alg:mfmc} defined below:
\begin{algorithm}[H]
	\caption{Construction of $\bff \in \calF$ and $S' \subseteq S$ satisfying properties of Theorem~\ref{thm:mfmc}}
	\label{alg:mfmc}
	\begin{algorithmic}[1]
		\State $G^{(0)} \gets G, \bff^{(0)} \gets $ all-0 vector over domain $E$
		\For{every $\ell = 0$ to $L := \floor{\log_{1+\epsilon}\frac{3|\bfa|_1}\phi}$}
		\State $(G^{(\ell+1)}, \bff^{(\ell+1)}) \gets \inclen(\ell, G^{(\ell)}, \bff^{(\ell)})$
		\EndFor
		\State \Return the projection $\bff$ of $\bff^{(L+1)}$ to $G$, and $S'$ obtained using Lemma~\ref{lemma:find-cut} over $G^{(L+1)}$ and $\bff^{(L+1)}$ 
	\end{algorithmic}
\end{algorithm}

By Theorem~\ref{thm:inclen-main}, for every $\ell \in [0, L+1]$, we have that $G^\ell$ is a handled graph, $\bff^{\ell} \in \calF^{G^{\ell}}$ and the shortest augmenting path in the shortcut graph for $\bff^\ell$ has length at least $2\ell+1$. Then by Lemma~\ref{lemma:find-cut},  we can find a set $S' \subseteq S$ with $a(S \setminus S') + \frac{\gamma b(T(S'))}{1+\epsilon} \leq \val(\bff^{L+1}) + \frac{\phi}{3} = \val(\bff) + \frac{\phi}{3}$. By Theorem~\ref{thm:inclen-main}, all handled graphs $G^{\ell}$ constructed have at most $3|V|$ vertices and $3|E|$ edges. So, the overall running time is $O\left(L\cdot|E|\cdot\log|V|\right) = O\left(\frac1\epsilon\cdot|E|\cdot \log |V|\cdot \log \frac{|\bfa|_1}{\phi}\right)$. This finishes the proof of Theorem~\ref{thm:mfmc}.

\end{document}